\documentclass[11pt]{amsart}

\usepackage{amsmath,amssymb,amsthm,mathtools}
\usepackage{bm}
\usepackage{xcolor}

\usepackage{pgfplots}
\pgfplotsset{compat=1.16}
\usepackage{enumitem}
\usepackage[hidelinks]{hyperref}

\usepackage{algorithm}
\usepackage{algpseudocode}

\numberwithin{equation}{section}
\newtheorem{theorem}{Theorem}[section]
\newtheorem{lemma}[theorem]{Lemma}
\newtheorem{proposition}[theorem]{Proposition}
\newtheorem{corollary}[theorem]{Corollary}

\newtheorem{example}{Example}[section]
\theoremstyle{remark}
\newtheorem{remark}[theorem]{Remark}

\algrenewcommand\algorithmicrequire{\textbf{Input:}}
\algrenewcommand\algorithmicensure{\textbf{Output:}}

\DeclareMathOperator{\dist}{dist}
\DeclareMathOperator{\diag}{diag}

\newcommand{\xkh}[1]{\left(#1\right)}
\newcommand{\dkh}[1]{\left\{#1\right\}}
\newcommand{\zkh}[1]{\left[#1\right]}

\newcommand{\norm}[1]{\left\|{#1}\right\|_2}

\newcommand{\abs}[1]{\lvert#1\rvert}
\newcommand{\Abs}[1]{\left\lvert#1\right\rvert}

\newcommand{\vx}{{\bm x}}

\newcommand{\vv}{{\bm v}}
\newcommand{\vw}{{\bm w}}
\newcommand{\vz}{{\bm z}}
\newcommand{\vf}{{\bm f}}
\newcommand{\vd}{{\bm d}}

\newcommand{\vh}{{\bm h}}

\newcommand{\E}{{\mathbb E}}

\newcommand{\C}{{\mathbb C}}
\newcommand{\R}{{\mathbb R}}

\DeclareMathOperator{\Rea}{Re}
\DeclareMathOperator{\Ima}{Im}
\DeclareMathOperator{\supp}{supp}
\newcommand{\Prob}{\mathbb P}

\title[Gauss--Newton for Coded Diffraction Phase Retrieval]
{The Gauss--Newton Method for Phase Retrieval from Coded Diffraction Patterns}

\author{Meng Huang}
\address{School of Mathematical Sciences, Beihang University, Beijing, 100191, China}
\email{menghuang@buaa.edu.cn}

\author{Tianjiao Wang}
\address{School of Artificial Intelligence, University of Science and Technology Beijing, China}
\email{tianjwang6-c@my.cityu.edu.hk}

\date{}
\begin{document}
\begin{abstract}
We study the minimum-norm Gauss--Newton method for phase retrieval from random coded diffraction patterns with additive intensity noise. For bounded symmetric masks satisfying the standard complex moment conditions, we establish convergence guarantees with quadratic error recursions in both the resampled and fixed-data settings.
In the resampled setting, $O(\log^2 n)$ masks per independent block suffice to obtain a spectral initializer within a constant-relative-radius neighborhood and to establish a local quadratic error recursion with an additive noise term. Consequently, in the noiseless case, the method achieves relative accuracy $\epsilon$ using $O(\log\log(1/\epsilon)\log^2 n)$ masks. In the fixed-data setting, $O(\log^6(en))$ masks suffice for spectral initialization within a neighborhood whose relative radius is proportional to $\log^{-5/2}(nL)$. Within this neighborhood, the method satisfies a linear contraction bound and eventually a quadratic error recursion. In particular, in the noiseless case, the method initially converges linearly and enters a quadratic convergence phase after $O(\log n)$ iterations, without acquiring additional masks. Numerical experiments on synthetic signals and a natural image demonstrate the method's efficiency relative to state-of-the-art algorithms.
\end{abstract}
\maketitle

\section{Introduction}
\label{sec:introduction}

\subsection{Problem setup and motivation}
Phase retrieval seeks to reconstruct a signal from measurements that record magnitude but not phase. In diffraction imaging, this leads to the recovery of an object from Fourier intensity data.  However, this problem is generally ill-posed \cite{beinert2018,bendory,edidin2019,huangJoC}.  To address this issue, a common approach to ensuring uniqueness is to use multiple masks that introduce redundancy into the acquired data. This setup is called the {\em coded diffraction pattern} (CDP) model \cite{candes2015cdp,gross2017improved}. Specifically, fix an unknown signal $\vx_\star \in \C^n\setminus\{\bm0\}$ with $n\ge2$. CDP phase retrieval aims to recover $\vx_\star$ from the noisy measurements 
\begin{equation*}
 y_{\ell,k}= \Abs{\sum_{t=0}^{n-1} x_\star(t) \bar{d}_\ell(t) e^{-\frac{2\pi i k t}{n}}}^2 + \xi_{\ell,k}, \quad 0 \le k < n, \; 1 \le \ell \le L,
\end{equation*}
{Here, $\vd_\ell \in \C^n$ are the $L$ known masks, $\xi_{\ell,k}$ are real-valued intensity perturbations, and $\bar d_\ell(t)$ denotes the complex conjugate of $d_\ell(t)$. With $m=nL$, the model can be written as}
\begin{equation}  \label{eq:model}
 y_{\ell,k}=|\bm f_k^*{\bm D}_\ell^*\bm x_\star|^2+\xi_{\ell,k},
 \qquad \bm\xi=(\xi_{\ell,k})_{\ell,k}\in\mathbb R^m, \quad {0 \le k < n, \; 1 \le \ell \le L},
\end{equation}
where 
\begin{equation}
 \bm f_k=(e^{2\pi\mathrm i kj/n})_{j=0}^{n-1},\qquad
 {\bm D}_\ell=\diag(d_{\ell,0},\ldots,d_{\ell,n-1}).
 \label{eq:fourier-vector0}
\end{equation}
%There are several experimental techniques to generate masked Fourier measurements in optical setups, such as inserting a mask or a phase plate after the object \cite{liu2008phase,Loewen}. 
Phase retrieval has attracted significant attention across diverse fields of physical science and engineering, including X-ray crystallography \cite{Xcry1,Xcry2}, diffraction imaging \cite{chai2010array}, microscopy \cite{miao2008extending}, as well as optics and acoustics \cite{walther1963question, balan2010signal, balan2006signal}.

A common approach to reconstructing \(\bm x_\star\in \C^n\)  from the measurements \(\{y_{\ell,k}\}_{\ell,k}\) in \eqref{eq:model}  is to solve
\begin{equation} \label{eq:minls}
 \min_{\bm z\in\mathbb C^n}\quad \frac1{2m}\sum_{\ell=1}^L\sum_{k=0}^{n-1}
 \bigl(|\bm f_k^*{\bm D}_\ell^*\bm z|^2-y_{\ell,k}\bigr)^2.
\end{equation}
Several algorithms with provable guarantees have been developed for \eqref{eq:minls}, including the seminal work of Cand\`es, Li, and Soltanolkotabi \cite{candes2015wf}. They introduced the Wirtinger flow (WF) method, which combines spectral initialization with gradient updates for \eqref{eq:minls}, and showed that, under random masking with $L\ge O(\log^4 n)$, WF converges linearly to the target solution \(\bm x_\star\) with high probability.  Recently,  Li, Li, and Xia \cite{li2022sampling} analyzed resampled truncated Wirtinger flow and obtained relative accuracy \(\epsilon\) using \(O(\log(1/\epsilon)\log^2 n)\) masks in total.

To date, almost all  provable nonconvex methods for CDP phase retrieval use first-order updates. The least-squares structure of \eqref{eq:minls} makes Gauss--Newton a natural second-order method. Establishing its convergence, however, presents two difficulties. First, phase invariance forces the Jacobian at each iterate \(\bm z_t\) to annihilate \(\mathrm i\bm z_t\), so the linearized subproblem has nonunique solutions whenever \(\bm z_t\ne\bm0\). Choosing the minimum-norm solution resolves this ambiguity, but the analysis must account for the resulting null direction. Second, Fourier measurements generated by the same mask share its random entries. With fixed data, the iterates also depend on the masks defining their Jacobians, so concentration estimates at independent deterministic points cannot be applied directly along the trajectory. Resampling simplifies this dependence. For example, WF uses fresh masks during its theoretical initialization procedure, whereas the analysis in \cite{li2022sampling} uses fresh masks at each update. However, this device changes the data-use protocol and makes the total number of masks depend on the number of independent blocks.

These observations motivate the following question: 

{\em Can Gauss--Newton achieve quadratic convergence for coded diffraction phase retrieval using a fixed polylogarithmic collection of masks?}

\subsection{Related work}
\subsubsection{Convex recovery for CDP phase retrieval.}
Over the past decade, numerous algorithms with provable performance guarantees have been developed for the CDP phase retrieval problem. One prominent approach is convex relaxation, which lifts the unknown signal to the rank-one Hermitian matrix \(\bm x_\star\bm x_\star^*\) and formulates recovery as a semidefinite program. Cand\`es, Li, and Soltanolkotabi \cite{candes2015cdp} established exact recovery from \(O(\log^4 n)\) random masks. Gross, Krahmer, and Kueng \cite{gross2017improved} subsequently reduced the required number of masks to \(L \ge C\omega\log^2 n\), with failure probability at most \(e^{-\omega}\). More recently, Huang and Li \cite{huangGao} further improved this bound to the optimal sampling rate of \(O(\log n)\). For noisy CDP measurements, Huang, Wen, and Zhang \cite{huang2026cdp} established estimation error bounds for PhaseLift under both adversarial and subgaussian noise, making progress toward Soltanolkotabi's conjecture. Despite these theoretical advances, the computational cost of solving the resulting semidefinite programs remains a major limitation of convex methods, particularly for large-scale problems.

\subsubsection{First-order methods for CDP phase retrieval.}
To improve computational efficiency, another line of research directly optimizes a nonconvex loss function in the signal space. A pioneering example is the WF method \cite{candes2015wf}, which requires \(O(\log^4 n)\) random masks for the theoretical guarantee. The CDP guarantee for WF differs from its Gaussian-measurement guarantee: it uses resampling during initialization before subsequent geometrically convergent updates. Li, Li, and Xia \cite{li2022sampling} analyzed resampled truncated Wirtinger flow for masked Fourier measurements using a truncated spectral initializer. Li and Li \cite{li2025taf} studied revised truncated amplitude flow for CDPs, proving noiseless convergence and stability under bounded amplitude noise. Their analysis also uses resampling \cite{li2025taf}.  For these results, the number of masks per independent block must be distinguished from the total acquisition budget.

\subsubsection{Gauss--Newton methods for phase retrieval.}
The Gauss--Newton method is widely used for nonlinear least-squares problems.  For phase retrieval, Gao and Xu \cite{gao2017phaseless} established quadratic convergence with resampling for real-valued signals measured by complex Gaussian sensing vectors. Huang \cite{huang2026near} extended the analysis to complex-valued signals using the minimum-norm Gauss--Newton method and proved near-quadratic convergence without sample splitting under complex Gaussian measurements.  That analysis uses leave-one-out arguments to control the data-dependent trajectory. In contrast, the Fourier sensing vectors within a single CDP mask are dependent, and controlling the approximation error after removing an entire mask is difficult when only polylogarithmically many masks are available. Thus, the Gaussian result does not directly yield a fixed-data CDP guarantee, and its leave-one-out argument does not directly apply to the CDP setting considered here.

\subsection{{Our contributions}}

To our knowledge, existing nonconvex methods with provable guarantees for CDP phase retrieval use first-order updates, and their analyses rely heavily on resampling to ensure independence between the random masks and the iterates.  In this paper, we analyze the minimum-norm Gauss--Newton method and establish quadratic convergence guarantees, with particular emphasis on the fixed-data setting.  Because the iterates depend on the masks and Fourier measurements generated by the same mask share its random entries, concentration estimates at independent deterministic points cannot be applied directly along a fixed-data trajectory.  To overcome this challenge,  at each iterate $\bm z_t$, we partition the measurements $\dkh{{\bm D}_{\ell}\bm f_k}_{\ell, k}$ into the ``bad" set  \( S=\{(\ell,j):|\bm f_j^*{\bm D}_\ell^*\xkh{\vz_t-\widetilde{\vx}_\star}|>c_0\}\) and its ``good" complement $S^c$, for a constant $c_0>0$. 
When $\bm z_t$ is sufficiently close to $\vx_\star$, we show that $S$ has small cardinality and control the aggregate contribution of the ``bad" measurements uniformly over all such $S$, using bounds for suprema of chaos processes from \cite{kmr2014}.  The resampled analysis provides intuition for the fixed-data analysis. We therefore study the minimum-norm Gauss--Newton method under both data-use protocols. Our main results are:
\begin{enumerate}
\item \emph{Resampled convergence.}  We prove noisy spectral-initialization and one-step quadratic-error bounds using \(O(\log^2 n)\) masks per independent block. In particular, in the noiseless case, the minimum-norm Gauss--Newton method achieves relative accuracy \(\epsilon\) using \(O(\log\log(1/\epsilon)\log^2 n)\) masks.
\item \emph{Fixed-data convergence.} Under mild conditions on the noise vector, with \(O(\log^6(en))\) masks in total, spectral initialization enters a neighborhood of relative radius proportional to \(\log^{-5/2}(nL)\), and the minimum-norm Gauss--Newton method satisfies a linear contraction bound. Under a stronger noise condition, it also satisfies a quadratic error recursion after \(O(\log n)\) iterations. In particular, in the noiseless case, the minimum-norm Gauss--Newton method converges linearly at first and enters a quadratic convergence phase after \(O(\log n)\) iterations, without acquiring new masks.
\end{enumerate}
 Numerical experiments examine recovery accuracy, convergence, and noise robustness for synthetic signals and a natural image. 
 %The uniform bound for $\mbox{sup}_{S} \left\|\frac1m\sum_{(\ell,k)\in S} {\bm D}_\ell\bm f_k\bm f_k^*{\bm D}_\ell^*\right\|_2$ established in Lemma~\ref{lem:subset} is also of independent interest. 

\subsection{Notation}
Vectors are written in bold, with scalar components denoted by the
corresponding subscripts. The symbols \((\cdot)^\top\),
\((\cdot)^*\), and \(\overline{(\cdot)}\) denote transpose,
conjugate transpose, and entrywise conjugation, respectively;
\(\Rea\) and \(\Ima\) denote real and imaginary parts.
For a vector, \(\|\cdot\|_2\) is the Euclidean norm; for a
matrix or a linear operator, it is the induced operator norm.
The Frobenius norm is denoted by \(\|\cdot\|_F\).
For a vector, \(\|\cdot\|_\infty\) denotes its largest component
magnitude.
The notation \(\supp\bm v\) denotes the set of nonzero
coordinates of \(\bm v\), \(|S|\) is the cardinality of a finite
set, and \(\diag\) forms a diagonal matrix from its arguments.
For self-adjoint operators, \(\bm A\preceq\bm B\) means that
\(\bm B-\bm A\) is positive semidefinite.  We write \(\E\) and \(\Prob\) for expectation and probability.
The constants \(C_M,C_{M,\beta},c_{M,\beta}>0\) may change between displays.  The notation \(O(\cdot)\)
suppresses constants independent of the dimension.

We view \(\mathbb C^n\) as a real Euclidean space with inner product
\begin{equation}
 \langle\bm u,\bm v\rangle_{\mathbb R}=\Rea(\bm u^*\bm v).
 \label{eq:real-inner-product}
\end{equation}
Adjoints and Moore--Penrose pseudoinverses of real-linear Jacobians
are taken with respect to this real structure.  Multiplying \(\bm x_\star\) by \(e^{\mathrm i\theta}\), for any \(\theta\in\mathbb R\), leaves
all measurements unchanged. Recovery is therefore understood up to
this unavoidable global phase. We define the distance modulo global phase by
\begin{equation}
 \dist(\bm z,\bm x_\star)=\min_{\theta\in\mathbb R}
     \|\bm z-e^{\mathrm i\theta}\bm x_\star\|_2.
 \label{eq:phase-distance}
\end{equation}
For \(\bm z^*\bm x_\star\ne0\), let
\(\widetilde{\bm x}_\star(\bm z)\) be the representative satisfying
\begin{equation}
 \bm z^*\widetilde{\bm x}_\star(\bm z)=|\bm z^*\bm x_\star|\in\mathbb R_+.
 \label{eq:alignment}
\end{equation}
If \(\bm z^*\bm x_\star=0\), choose any representative attaining
the phase distance.

Throughout this paper, we assume that all \(nL\) mask entries are independent copies of a complex random variable \(d\) satisfying
\begin{equation}
 \begin{gathered}
 d\overset{\mathrm{law}}{=}-d,\qquad \E d=\E d^2=0,\\
 \E|d|^2=1,\qquad \E|d|^4=2,\qquad |d|\le M.
 \end{gathered}
 \label{eq:mask}
\end{equation}
The first condition states that \(d\) and \(-d\) have the same distribution, while the normalization \(\E|d|^2=1\) is imposed for convenience. These assumptions are standard in phase retrieval from random coded diffraction patterns; see \cite{candes2015cdp,gross2017improved}.

\subsection{Organization}
Section~\ref{sec:algorithm} describes spectral initialization and the minimum-norm Gauss--Newton update. Sections~\ref{sec:resampled-results} and~\ref{sec:fixed-results} state the resampled and fixed-data convergence guarantees, respectively. Section~\ref{sec:experiments} presents numerical experiments. The resampled results are proved in Section~\ref{sec:resampled-proof}, and the fixed-data results in Section~\ref{sec:fixed-proof}. Section~\ref{sec:discussion} discusses the scope and limitations of the guarantees. Appendix~\ref{sec:proLesubset} proves the uniform subset estimate, and Appendix~\ref{app:auxiliary} collects auxiliary lemmas.

\section{Gauss--Newton algorithm}
\label{sec:algorithm}

\subsection{Gauss--Newton refinement}
For \(k=0,\ldots,n-1\), let \(\bm{f}_k\in\mathbb C^n\) be the unnormalized discrete Fourier vector
\begin{equation}
\bm{f}_k =\xkh{1,\exp(2\pi\mathrm{i}k/n), \ldots, \exp(2\pi\mathrm{i}k(n-1)/n)}^\top.
\label{eq:fourier-vector}
\end{equation}
Fourier orthogonality gives \(\sum_{k=0}^{n-1}\bm{f}_k\bm{f}_k^*=nI_n\). We write the \(\ell\)-th random mask as
\begin{equation}
{\bm D}_\ell =
\diag(d_{\ell,0},\ldots,d_{\ell,n-1}),
\qquad \ell=1,\ldots,L.
\label{eq:fourier-vector1}
\end{equation}
Let \(m=nL\). A common approach to reconstructing \(\bm x_\star\)  from the measurements \(\{y_{\ell,k}\}_{\ell,k}\) in \eqref{eq:model}  is to solve
\begin{equation}
 \min_{\bm z\in\mathbb C^n}\quad
 \Phi(\bm z)=\frac1{2m}\sum_{\ell=1}^L\sum_{k=0}^{n-1}
 \bigl(|\bm f_k^*{\bm D}_\ell^*\bm z|^2-y_{\ell,k}\bigr)^2
 =\frac12\|\bm F(\bm z)\|_2^2,
 \label{eq:objective}
\end{equation}
{where the $(\ell,k)$-th component of $\bm F(\vz)$ is}
\begin{equation}
 [\bm F(\bm z)]_{\ell,k}
 =\frac{|\bm f_k^*{\bm D}_\ell^*\bm z|^2-y_{\ell,k}}{\sqrt m}.
 \label{eq:residual}
\end{equation}
To solve this nonlinear least-squares problem, we apply the
Gauss--Newton method. Using Wirtinger calculus, the real first-order
expansion at the current iterate \(\bm z_t\) is
\begin{equation}
 \bm F(\bm z_t+\bm h)
 =\bm F(\bm z_t)+\mathcal J_{\bm z_t} \vh  +O(\|\bm h\|_2^2),
 \label{eq:taylor}
\end{equation}
where the real-linear operator \(\mathcal J_{\bm z_t}:\mathbb C^n\to\mathbb R^m\)
is defined by
\begin{equation}
 \mathcal J_{\bm z_t}\bm v
 =\frac2{\sqrt m}
 \begin{bmatrix}
 \Rea(\overline{\bm f_0^*{\bm D}_1^*\bm z_t}\,\bm f_0^*{\bm D}_1^*\bm v)\\
 \vdots\\
 \Rea(\overline{\bm f_{n-1}^*{\bm D}_L^*\bm z_t}\,
                       \bm f_{n-1}^*{\bm D}_L^*\bm v)
 \end{bmatrix}.
 \label{eq:jac}
\end{equation}
{At the current iterate \(\bm z_t\), Gauss--Newton replaces the residual by its first-order approximation and solves}
\begin{equation}
\mbox{min}_{\bm h\in\mathbb C^n}\quad
 \frac12\left\| \mathcal J_{\bm z_t} \vh  +\bm F(\bm z_t)\right\|_2^2,
 \label{eq:linearized}
\end{equation}
{The next iterate is $\vz_{t+1}:=\vz_t+\vh_t$, where $\vh_t$ solves \eqref{eq:linearized}. Since \(\mathcal J_{\bm z_t}(\mathrm i\bm z_t)=\bm0\), the Jacobian is rank deficient and the least-squares solution is nonunique whenever \(\bm z_t\ne\bm0\). We select the minimum-norm solution, which gives}
\begin{equation}
 \bm z_{t+1}  =\bm z_t-\mathcal J_{\bm z_t}^{\dagger}\bm F(\bm z_t).
 \label{eq:gn}
\end{equation}
The pseudoinverse is taken with respect to the real inner product given in \eqref{eq:real-inner-product}. An equivalent matrix representation is obtained by defining
\begin{equation}
 J_{\bm z_t}=\frac1{\sqrt m}
 \begin{bmatrix}
 \bm z_t^*{\bm D}_1\bm f_0\bm f_0^*{\bm D}_1^*
 &\bm z_t^\top\overline{{\bm D}_1\bm f_0}({\bm D}_1\bm f_0)^\top\\
 \vdots&\vdots\\
 \bm z_t^*{\bm D}_L\bm f_{n-1}\bm f_{n-1}^*{\bm D}_L^*
 &\bm z_t^\top\overline{{\bm D}_L\bm f_{n-1}}({\bm D}_L\bm f_{n-1})^\top
 \end{bmatrix}
 \in\mathbb C^{m\times2n},
 \label{eq:jacobian}
\end{equation}
so that $\mathcal J_{\bm z_t} \vh={J_{\bm z_t}} \begin{bmatrix} \bm h\\ \overline{\bm h} \end{bmatrix} $. Therefore, the update \eqref{eq:gn} can also be written as
 \begin{equation}
  \bm z_{t+1}  =\bm z_t-\bigl(J_{\bm z_t}^{\dagger}\bm F(\bm z_t)\bigr)_{1:n},
  \label{eq:gn1}
\end{equation}
where \((\cdot)_{1:n}\) denotes the first \(n\) entries, and $ J_{\bm z_t}^{\dagger}$ is the pseudoinverse of the matrix $J_{\bm z_t}$.

\subsection{Spectral initialization}
Due to the nonconvexity of \(\Phi\), a good initial point is required.
We adopt the spectral method used in \cite{candes2015wf}. Set
\begin{equation}
 Y=\frac1m\sum_{\ell=1}^L\sum_{k=0}^{n-1}
       y_{\ell,k}{\bm D}_\ell\bm f_k\bm f_k^*{\bm D}_\ell^*,
 \qquad \lambda=\frac1m\sum_{\ell=1}^L\sum_{k=0}^{n-1}y_{\ell,k}.
 \label{eq:init}
\end{equation}
In the absence of noise, under the admissible mask assumption \eqref{eq:mask}, the population matrix has the form
\[
 \E Y=\bm x_\star\bm x_\star^*+\|\bm x_\star\|_2^2I_n,
 \qquad \E\lambda=\|\bm x_\star\|_2^2.
\]
{Thus \(\bm x_\star\) is a leading eigenvector of \(\E Y\), motivating the estimator \(\sqrt\lambda\,\bm v_0\), where \(\bm v_0\) is a unit leading eigenvector of \(Y\). In the noiseless case, \cite{candes2015wf} shows that \(O(\log^3 n)\) masks suffice for \(\dist(\sqrt\lambda\,\bm v_0,\bm x_\star)\le\|\bm x_\star\|_2/8\) with high probability. Under the noise conditions below, \(\lambda>0\) also holds with high probability, so the initializer is well defined.}

{Our minimum-norm Gauss--Newton method combines spectral initialization with the update \eqref{eq:gn}, as summarized in Algorithm~\ref{rs:alg:sampled-gn}.}

\begin{algorithm}[htbp]
\caption{{Minimum-norm} Gauss--Newton for coded diffraction phase retrieval}
\label{rs:alg:sampled-gn}
\begin{algorithmic}[1]
\Require Masks ${\vd_\ell \in \C^n,\ \ell=1,\ldots,L}$; observations ${\{y_{\ell,k}\}_{1\le\ell\le L,\,0\le k<n}}$; the maximum number of
iterations $T$.
\State Compute \(\bm z_0=\sqrt\lambda\,\bm v_0\), where \(\bm v_0\) is a unit leading eigenvector of \(Y\) and $\lambda$ is defined in
\eqref{eq:init}.
\For{\(t=0,\ldots,T-1\)}
\[
 \bm z_{t+1}=\bm z_t- \mathcal J_{\bm z_t}^{\dagger}\bm F(\bm z_t),
\]
where $\mathcal J_{\bm z_t}$ and ${\bm F}(\bm z_t)$ are given in \eqref{eq:jac} and \eqref{eq:residual}, respectively.
\EndFor
\Ensure \(\bm z_T\).
\end{algorithmic}
\end{algorithm}

A direct implementation of the minimum-norm step uses a singular value decomposition, but explicitly forming the pseudoinverse can be computationally expensive. In practice, the linearized subproblem can instead be solved approximately using an iterative least-squares method. We use LSQR in the numerical experiments. Each LSQR iteration requires one application of the Jacobian and one application of its adjoint. Exploiting the CDP structure, we compute both products using fast Fourier transforms (FFTs). Each product requires $O(L)$ Fourier transforms of length $n$ and $O(m)$ additional operations, resulting in a cost of $O(Ln\log n)=O(m\log n)$. Thus, with at most $T_0\geq 1$ inner LSQR iterations, each Gauss--Newton iteration costs $O(T_0m\log n)$. For a fixed inner iteration budget, such as $T_0=10$, this becomes $O(m\log n)$, matching the asymptotic cost of an FFT-based gradient iteration.

\section{Convergence with resampling}
\label{sec:resampled-results}
{This section establishes a noisy quadratic error bound with \(O(\log^2 n)\) masks per block and a complete noiseless convergence guarantee with resampling. Section~\ref{sec:fixed-results} treats the linear-to-quadratic behavior of Algorithm~\ref{rs:alg:sampled-gn} with fixed data.}

{We first extend the noiseless spectral-initialization guarantee in \cite{candes2015wf} to bounded intensity noise and reduce the number of masks from \(O(\log^3 n)\) to \(O(\log^2 n)\).}

\begin{theorem} \label{rs:cor:spectral-initialization}
Let ${\vx_\star\ne\bm0}$ be an unknown signal. Assume that the entries of the masks $\dkh{\vd_l}_{l=1}^L$ are independent copies of a complex random variable $d$ satisfying \eqref{eq:mask}.  Consider the noisy measurements $y_{\ell,k}=\abs{\vf_k^* {\bm D}_{\ell}^* \vx_\star}^2+\xi_{\ell,k}, 1\le \ell\le L, {0\le k<n}$.  
If the noise vector $\bm\xi\in \mathbb R^{nL}$ satisfies
\[
\|\bm\xi\|_\infty\le\frac1{128 M^2}\|\bm x_\star\|_2^2,
\qquad
\frac{\|\bm\xi\|_1}{nL}\le\frac1{128}\|\bm x_\star\|_2^2,
\]
then with probability at least \(1-2n^{-3}\), the initial point $\vz_0:=\sqrt\lambda \vv_0$  obeys 
\begin{equation}
\dist(\bm z_0,\bm x_\star) <\frac18\|\bm x_\star\|_2,
\label{eq:log-squared-initialization}
\end{equation}
provided  \(L\ge C_{M} \log^2 n\). Here, \(\bm v_0\) is a unit-norm leading eigenvector of \(Y\), and \(C_M\) is a constant depending only on \(M\), the bound specified in \eqref{eq:mask}.
\end{theorem}

We next establish the following one-step convergence guarantee.

\begin{theorem} \label{rs:thm:complete}
Let $\vx_\star \neq \bm 0$ be an unknown signal. Assume that the entries of the masks $\dkh{\vd_l}_{l=1}^L$ are independent copies of a complex random variable $d$ satisfying \eqref{eq:mask}.  Consider the noisy measurements $y_{\ell,k}=\abs{\vf_k^* {\bm D}_{\ell}^* \vx_\star}^2+\xi_{\ell,k}, 1\le \ell\le L, {0\le k<n}$.   Suppose that the current {iterate $\vz_t \in \C^n$} is independent of the random masks $\dkh{\vd_l}_{l=1}^L$ and satisfies ${\dist(\bm z_t,\bm x_\star)} \le 1/8\cdot \|\bm x_\star\|_2$. For any \(\delta\in(0,1/2]\),  if 
 $L\ge {C_{M,\delta}} \log^2 n$ for some constant ${C_{M,\delta}}$ depending only on $\delta,M$, then
 with probability at least \(1-2n^{-3}\), the {iterate $\vz_{t+1}$} given by \eqref{eq:gn} obeys
\begin{equation*}
\frac{\dist(\bm z_{t+1},\bm x_\star)}{\norm{\vx_\star}}
\le \kappa_\delta \xkh{\frac{\dist(\bm z_{t},\bm x_\star)}{\norm{\vx_\star}}}^2 +
\sqrt\frac{ 2}{{(1-\delta)}}
\frac{\|\bm\xi \|_2} {\sqrt{nL}\,\|\bm x_\star\|_2^2},
\end{equation*}
where $\kappa_\delta =2\sqrt{\frac{2+\delta}{2(1-\delta)}} \le \sqrt{10}$, {and $\bm\xi=(\xi_{\ell,k})_{1\le\ell\le L,\,0\le k<n}$ is the noise vector.}
\end{theorem}

\begin{proof}
{See Section~\ref{sec:resampled-onestep-proof}.}
\end{proof}

{To satisfy the independence assumption of Theorem~\ref{rs:thm:complete}, Algorithm~\ref{alg:resampled-gn} uses block \(B_0\) for initialization and a fresh block \(B_{t+1}\) for the update at \(\bm z_t\).}

\begin{algorithm}[H]
\caption{Resampled Gauss--Newton for Coded Diffraction Phase Retrieval}
\label{alg:resampled-gn}
\begin{algorithmic}[1]
\Require Masks ${\vd_\ell \in \C^n,\ \ell=1,\ldots,L}$; observations ${\{y_{\ell,k}\}_{1\le\ell\le L,\,0\le k<n}}$; relative precision ${0<\epsilon<1/8}$.
\State {Set $T=\lceil\log_2\log_2(1/\epsilon)\rceil$, where $\log_2$ denotes the base-two logarithm.}
\State {Form $T+1$ disjoint mask blocks $B_0,\ldots,B_T$, each of size $\lfloor L/(T+1)\rfloor$, discarding any unused masks.}
\State {Set $\vz_0=\sqrt\lambda \vv_0$, where $\lambda=\frac{1}{n |B_0|} \sum_{\ell\in B_0}\sum_{k=0}^{n-1} y_{\ell,k}$, and $\vv_0$ is a unit leading eigenvector of $Y_0:=\frac{1}{n |B_0|} \sum_{\ell\in B_0}\sum_{k=0}^{n-1} y_{\ell,k} ({\bm D}_\ell \vf_k) ({\bm D}_\ell \vf_k)^*$.}
\For{\(t=0,\ldots,T-1\)}
 \[
 \bm z_{t+1}=\bm z_t-\mathcal J_{t+1, \bm z_t}^{\dagger}{\bm F_{t+1}(\bm z_t)},
 \]
  {where $\mathcal J_{t+1,\bm z_t}$ and $\bm F_{t+1}(\bm z_t)$ are given by \eqref{eq:jac} and \eqref{eq:residual} using block $B_{t+1}$ and normalization $1/\sqrt{n|B_{t+1}|}$.}
\EndFor
\Ensure \(\bm z_T\).
\end{algorithmic}
\end{algorithm}

{
\begin{corollary} \label{rs:cor:trajectory}
Let \(\bm x_\star\ne\bm0\) be fixed, and assume the mask model \eqref{eq:mask} with noiseless measurements \(y_{\ell,k}=|\bm f_k^*{\bm D}_\ell^*\bm x_\star|^2\). For \(0<\epsilon<1/8\), set \(T=\lceil\log_2\log_2(1/\epsilon)\rceil\). If \(L\ge C_M(T+1)\log^2 n\), then, with probability at least \(1-2(T+1)n^{-3}\), Algorithm~\ref{alg:resampled-gn} outputs \(\bm z_T\) satisfying
\[
\dist(\bm z_T,\bm x_\star)<\epsilon\|\bm x_\star\|_2.
\]
Here, \(C_M>0\) depends only on \(M\), and the total number of masks is \(O(\log^2 n[1+\log\log(1/\epsilon)])\).
\end{corollary}

\begin{proof}
Increasing \(C_M\) if necessary ensures that each block has enough masks for Theorems~\ref{rs:cor:spectral-initialization} and~\ref{rs:thm:complete}, the latter with \(\delta=1/2\). The initialization event gives \(\dist(\bm z_0,\bm x_\star)<\|\bm x_\star\|_2/8\). Conditional on all previous blocks and their successful updates, \(\bm z_t\) is independent of \(B_{t+1}\). The one-step bound, with \(\kappa_{1/2}=\sqrt{10}\), preserves this neighborhood since \(\sqrt{10}/8<1/2\). Induction therefore gives
\[
\frac{\dist(\bm z_t,\bm x_\star)}{\|\bm x_\star\|_2}\le\frac1{\sqrt{10}}\left(\frac{\sqrt{10}}8\right)^{2^t}<2^{-2^t},\qquad 0\le t\le T.
\]
The sum of the initialization failure probability and the \(T\) conditional step-failure probabilities is at most \(2(T+1)n^{-3}\). Finally, the definition of \(T\) gives \(2^{-2^T}\le\epsilon\).
\end{proof}
}

\section{Convergence with fixed data} \label{sec:fixed-results}
In this section, we establish convergence guarantees for the Gauss--Newton method with fixed data, using the same \(L\) masks for initialization and every update. With \(O(\log^6(en))\) masks, Algorithm~\ref{rs:alg:sampled-gn} admits a linear contraction bound followed by a quadratic error recursion in a smaller neighborhood. The main technical ingredient is the uniform subset estimate in Lemma~\ref{lem:subset}, which controls the pseudoinverse throughout the attraction neighborhood. The guarantee holds simultaneously for all admissible noise vectors, even when they depend on the masks.

\begin{theorem} \label{thm:noisy}
Fix \(n\ge2\), \(\beta>0\), and let \(\bm x_\star\in\mathbb C^n\setminus\{\bm0\}\) be an unknown signal.
Under the mask model \eqref{eq:mask}, there are constants
\(C_{M,\beta},c_{M,\beta}>0\) such that, if \(  L\ge C_{M,\beta}\log^6(nL)\) and {the noise vector satisfies} \( \|\bm\xi\|_\infty \le c_{M,\beta} \log^{-5/2}(nL) \|\bm x_\star\|_2^2\), then with probability at least \(1-(nL)^{-\beta}\), the following {statements hold} for Algorithm \ref{rs:alg:sampled-gn}:
\begin{enumerate}
 \item The initial vector $\vz_0$ satisfies 
 \begin{equation*}
 \frac{\dist(\bm z_0,\bm x_\star)}{\|\bm x_\star\|_2} \le C_{M,\beta} \log^{-5/2}(nL).
 \end{equation*}
 \item  For every \(t\ge0\),
 \begin{equation*}
 \dist(\bm z_{t+1},\bm x_\star)
 \le\frac{15}{16}\dist(\bm z_t,\bm x_\star) +\frac{\|\bm\xi\|_2}{\sqrt m\,\|\bm x_\star\|_2}.
 \end{equation*}
  \item If in addition \(
 \frac{\|\bm\xi\|_2}{\sqrt m\,\|\bm x_\star\|_2^2}
 \le\frac1{512M^2\sqrt n} \),  there is a deterministic \(T\le C_{M,\beta}\log(en)\) such that, for every \(t\ge T\),
\[
 \frac{\dist(\bm z_t,\bm x_\star)}{\|\bm x_\star\|_2}
       \le\frac1{16M^2\sqrt n}
 \]
and 
\[       
 \frac{\dist(\bm z_{t+1},\bm x_\star)}{\|\bm x_\star\|_2}
 \le2M^2\sqrt n  \left(\frac{\dist(\bm z_t,\bm x_\star)}
                        {\|\bm x_\star\|_2}\right)^2
       +\frac{\|\bm\xi\|_2}{\sqrt m\,\|\bm x_\star\|_2^2}.
\]
\end{enumerate}
Here, $m=nL$.
\end{theorem}

\begin{remark}
Theorem~\ref{thm:noisy} implies
\[
 \limsup_{t\to\infty}\dist(\bm z_t,\bm x_\star) \le16\frac{\|\bm\xi\|_2}{\sqrt m\,\|\bm x_\star\|_2}.
\]
{This is a deterministic stability bound. A similar error bound $O(\|\bm\xi\|_2/\sqrt m)$ has been established in \cite{huangPerfm} for phase retrieval with Gaussian random measurements, where the bound is proved to be sharp.
}
\end{remark}

\begin{remark}
In the noiseless case \(\bm\xi=\bm0\), Theorem~\ref{thm:noisy} guarantees that, with high probability, the minimum-norm Gauss--Newton method in Algorithm~\ref{rs:alg:sampled-gn} converges at least linearly during the first $O(\log n)$ iterations and quadratically thereafter. Consequently, it reaches relative accuracy $\epsilon$ within
\[
O(\log n+ \log\log(1/\epsilon))
\]
iterations.
\end{remark}

\section{Numerical experiments} \label{sec:experiments}
We evaluate the effectiveness and robustness of the Gauss--Newton method for coded diffraction phase retrieval and compare it with Wirtinger flow (WF) \cite{candes2015wf}, truncated Wirtinger flow (TWF)
\cite{chen2017twf}, and truncated amplitude flow (TAF)
\cite{wang2018taf}.   The masks are octanary:
\[
 \Prob\{d=\omega/\sqrt2\}=\frac15,\qquad
 \Prob\{d=\sqrt3\,\omega\}=\frac1{20},
 \qquad \omega\in\{1,-1,\mathrm i,-\mathrm i\}.
\]
Thus \eqref{eq:mask} holds with \(M=\sqrt3\). Within each trial, the four methods share both the observations and
the same spectral initial point, computed by 50 power iterations
for \eqref{eq:init}.  The GN subproblem is solved by LSQR \cite{paige1982lsqr}, starting
from zero, with tolerance \(10^{-5}\) and at most 10 inner
iterations for vectors and 5 for images. We use the step sizes
\[
 \mu_{\rm WF}(t)=\min\{1-\exp(-t/330),\,0.2\},\qquad
 \mu_{\rm TWF}=0.2,\qquad \mu_{\rm TAF}=0.8.
\]
The maximum WF step size is increased to \(0.4\) in the image experiment.
The TWF truncation parameters are \(0.3,5,5\), and the TAF truncation parameter is \(0.7\).
The experiments use MATLAB R2023a  on an Apple M1 Pro
computer with 16~GB of memory.  The error reported below is
\[
 \frac{\dist(\bm z_t,\bm x_\star)}{\|\bm x_\star\|_2}.
\]

\subsection{Recovery from noiseless measurements}

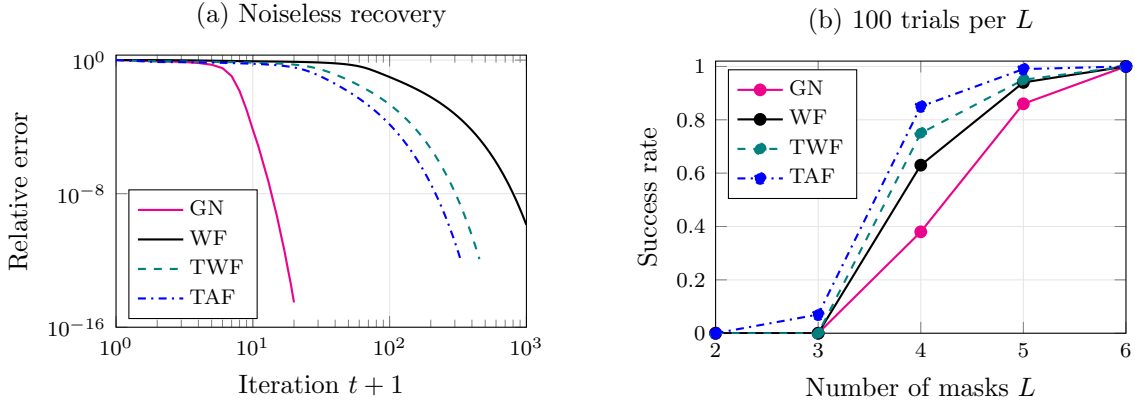
\begin{figure}[htbp]
\centering
\begin{tikzpicture}
\begin{loglogaxis}[width=0.46\textwidth,height=0.34\textwidth,tick label style={font=\scriptsize},label style={font=\small},title style={font=\small},grid=major,grid style={black!10},legend style={font=\scriptsize,cells={anchor=west}},legend pos=south west,xlabel={Iteration $t+1$},ylabel={Relative error},title={(a) Noiseless recovery},xmin=1,xmax=1001,ymin=1e-16,ymax=2]
\addplot[magenta,thick] coordinates {(1,1.017186) (2,0.8471575) (3,0.7924552) (4,0.6761145) (5,0.5299613) (6,0.3333593) (7,0.1110927) (8,0.01419528) (9,0.0010618) (10,6.893824e-05) (11,6.666864e-06) (12,5.335304e-07) (13,5.46005e-08) (14,4.507934e-09) (15,4.680734e-10) (16,3.911586e-11) (17,4.087872e-12) (18,3.43713e-13) (19,3.602704e-14) (20,3.073909e-15)};
\addlegendentry{GN}
\addplot[black,thick] coordinates {(1,1.017186) (2,1.002967) (3,0.9790011) (4,0.9522108) (5,0.9273851) (6,0.9063402) (7,0.8891639) (8,0.8752776) (9,0.8639615) (10,0.85456) (11,0.8465396) (12,0.8394889) (13,0.8330989) (14,0.8271408) (15,0.8214454) (16,0.8158877) (17,0.8103752) (18,0.8048392) (19,0.7992287) (20,0.7935067) (21,0.7876465) (22,0.7816301) (23,0.7754461) (24,0.7690883) (25,0.7625544) (26,0.7558448) (27,0.7489615) (28,0.7419076) (29,0.734686) (30,0.7272997) (31,0.7197506) (33,0.7041674) (36,0.6795594) (38,0.6622903) (42,0.6254949) (45,0.5957848) (48,0.5642169) (52,0.5194115) (57,0.4599409) (61,0.4106621) (66,0.3487452) (71,0.2895767) (77,0.2275693) (83,0.179536) (90,0.1379874) (98,0.1040015) (105,0.08233024) (114,0.06190439) (123,0.0471888) (133,0.03534541) (144,0.02603154) (155,0.0193633) (168,0.01378727) (181,0.009902725) (196,0.006817659) (212,0.004615176) (229,0.003070888) (248,0.001961218) (267,0.001259808) (289,0.0007589778) (312,0.0004491966) (338,0.0002495566) (365,0.0001361729) (394,7.134224e-05) (426,3.510082e-05) (461,1.622184e-05) (498,7.198855e-06) (538,3.000812e-06) (581,1.17498e-06) (628,4.228545e-07) (679,1.398803e-07) (734,4.253415e-08) (793,1.188852e-08) (857,2.989147e-09) (926,6.760941e-10) (1001,1.346374e-10)};
\addlegendentry{WF}
\addplot[teal,thick,dashed] coordinates {(1,1.017186) (2,0.8438542) (3,0.8359345) (4,0.8028313) (5,0.7882545) (6,0.7700875) (7,0.7553371) (8,0.7399287) (9,0.7252154) (10,0.7104726) (11,0.6955832) (12,0.6794537) (13,0.6639489) (14,0.6480225) (15,0.6323665) (16,0.6156192) (17,0.5981105) (18,0.5804134) (19,0.5612236) (20,0.5405762) (21,0.5181152) (22,0.4948652) (23,0.4709698) (24,0.447653) (25,0.4237542) (26,0.3996343) (27,0.3756502) (28,0.3518263) (29,0.328535) (30,0.3058375) (31,0.2843688) (33,0.2443345) (36,0.192216) (38,0.1631916) (41,0.1279698) (44,0.1006904) (47,0.07957008) (50,0.06315667) (54,0.04683514) (58,0.03498396) (62,0.02630372) (66,0.01990329) (71,0.01416501) (76,0.01014586) (81,0.007318576) (87,0.004976865) (93,0.003403676) (100,0.002195847) (107,0.001419023) (115,0.0008651698) (123,0.0005292668) (132,0.0003047767) (141,0.0001755921) (151,9.547603e-05) (162,4.889865e-05) (173,2.511767e-05) (186,1.146229e-05) (199,5.235298e-06) (213,2.256239e-06) (228,9.175302e-07) (245,3.308507e-07) (262,1.192988e-07) (281,3.811623e-08) (301,1.147935e-08) (322,3.255601e-09) (345,8.192932e-10) (369,1.945692e-10) (396,3.859307e-11) (424,7.24547e-12) (454,1.212527e-12)};
\addlegendentry{TWF}
\addplot[blue,thick,dashdotted] coordinates {(1,1.017186) (2,0.7733778) (3,0.7563824) (4,0.7347695) (5,0.7117534) (6,0.6895717) (7,0.667884) (8,0.6467054) (9,0.6262677) (10,0.6060158) (11,0.58547) (12,0.5648444) (13,0.5433802) (14,0.5214253) (15,0.4974294) (16,0.4720439) (17,0.4461711) (18,0.4200514) (19,0.3923567) (20,0.3646635) (21,0.3371446) (22,0.3088812) (23,0.2804773) (24,0.2526616) (25,0.2270218) (26,0.202338) (27,0.179665) (28,0.1592213) (29,0.1411823) (30,0.1249044) (31,0.1104221) (32,0.09768827) (34,0.07652002) (36,0.0603418) (39,0.04259871) (41,0.03401485) (44,0.02451467) (47,0.01784912) (50,0.01311134) (54,0.008793887) (57,0.00656561) (61,0.004482474) (65,0.003083916) (70,0.001949456) (74,0.001358611) (79,0.0008703456) (85,0.0005138041) (90,0.0003328615) (97,0.0001824352) (103,0.0001095058) (110,6.066463e-05) (118,3.105445e-05) (125,1.735292e-05) (134,8.248138e-06) (143,3.936992e-06) (153,1.737964e-06) (163,7.698982e-07) (174,3.154223e-07) (186,1.195372e-07) (198,4.542325e-08) (211,1.596214e-08) (226,4.787638e-09) (241,1.439075e-09) (257,4.00001e-10) (275,9.491717e-11) (293,2.255908e-11) (313,4.577492e-12) (334,8.58867e-13)};
\addlegendentry{TAF}
\end{loglogaxis}
\end{tikzpicture}
\hfill
\begin{tikzpicture}
\begin{axis}[width=0.46\textwidth,height=0.34\textwidth,tick label style={font=\scriptsize},label style={font=\small},title style={font=\small},grid=major,grid style={black!10},legend style={font=\scriptsize,cells={anchor=west}},legend pos=south west,xlabel={Number of masks $L$},ylabel={Success rate},title={(b) 100 trials per $L$},xmin=2,xmax=6,ymin=0,ymax=1.02,xtick={2,3,4,5,6},legend pos=north west]
\addplot[magenta,thick,mark=*] coordinates {(2,0) (3,0) (4,0.38) (5,0.86) (6,1)};
\addlegendentry{GN}
\addplot[black,thick,mark=*] coordinates {(2,0) (3,0) (4,0.63) (5,0.94) (6,1)};
\addlegendentry{WF}
\addplot[teal,thick,dashed,mark=*] coordinates {(2,0) (3,0) (4,0.75) (5,0.95) (6,1)};
\addlegendentry{TWF}
\addplot[blue,thick,dashdotted,mark=*] coordinates {(2,0) (3,0.07) (4,0.85) (5,0.99) (6,1)};
\addlegendentry{TAF}
\end{axis}
\end{tikzpicture}
\caption{Noiseless CDP experiments with $n=1000$. Panel (a) uses $L=6$ and a fixed seed; curves end when the stopping rule is met. Panel (b) counts recovery to relative error $10^{-5}$ in 100 paired trials for each $L$.}
\label{fig:cdp-noiseless}
\end{figure}

\begin{example}
In this example, we compare the convergence {behavior} of the four methods.
We draw \(\bm x_\star\sim\mathcal{CN}(\bm0,I_n)\), with \(n=1000\),
and use \(L=6\) masks. Figure~\ref{fig:cdp-noiseless}(a) shows that  GN reaches relative errors \(10^{-5}\) and \(10^{-10}\) in
10 and 15 outer iterations, respectively, and terminates after
19 iterations with error \(3.07\times10^{-15}\).
WF, TWF, and TAF first reach \(10^{-5}\) in 483, 188, and 131
iterations, respectively. 
\end{example}

\begin{example}
We next examine how the empirical recovery success rate depends on the number of masks. For each \(L\in\{2,3,4,5,6\}\), we draw 100 independent instances
with \(n=1000\), and use the same instance and initial point for all four methods. Recovery is counted as successful when the relative
error first falls below \(10^{-5}\). Figure~\ref{fig:cdp-noiseless}(b) reports the resulting rates.
All four methods succeed in all 100 trials when \(L=6\).
\end{example}

\subsection{Robustness to Gaussian and Poisson noise}
We now investigate the performance of the Gauss--Newton method with noisy measurements, considering Gaussian and Poisson noise.  In the Gaussian experiment,
\[
 y_{\ell,k}=|\bm f_k^*{\bm D}_\ell^*\bm x_\star|^2+\xi_{\ell,k},
 \qquad \xi_{\ell,k}\overset{\rm iid}{\sim}N(0,0.1^2).
\]
In the Poisson experiment, conditional on the masks and target,
the observations are independent with
\[
 y_{\ell,k}\sim\operatorname{Poisson}
                  (|\bm f_k^*{\bm D}_\ell^*\bm x_\star|^2).
\]
Treated as an additive perturbation, the latter noise is
\(\xi_{\ell,k}=y_{\ell,k}-|\bm f_k^*{\bm D}_\ell^*\bm x_\star|^2\).
{The realized SNR and squared relative error are expressed in decibels as}
\[
 \operatorname{SNR}
 =10\log_{10}
   \frac{\sum_{\ell,k}|\bm f_k^*{\bm D}_\ell^*\bm x_\star|^4}
        {\|\bm\xi\|_2^2},
 \qquad
 \operatorname{MSE}
 =10\log_{10}
   \frac{\dist^2(\bm z_T,\bm x_\star)}{\|\bm x_\star\|_2^2}.
\]

\begin{example}
We set \(n=1000\) and \(L=6\). 
Figure~\ref{fig:cdp-noisy-curves} displays the convergence curves.
The realized SNRs are \(83.13\)~dB and \(33.12\)~dB for the Gaussian
and Poisson data, respectively. GN reaches its error plateau in substantially fewer outer iterations.
The final errors of GN and WF are approximately \(5.59\times10^{-5}\)
for the Gaussian data and \(2.00\times10^{-2}\) for the Poisson data.
\end{example}

\begin{figure}[htbp]
\centering
\begin{tikzpicture}
\begin{loglogaxis}[width=0.46\textwidth,height=0.34\textwidth,tick label style={font=\scriptsize},label style={font=\small},title style={font=\small},grid=major,grid style={black!10},legend style={font=\scriptsize,cells={anchor=west}},legend pos=north east,xlabel={Iteration $t+1$},ylabel={Relative error},title={Gaussian intensity noise},xmin=1,xmax=1001,ymin=1e-5,ymax=2]
\addplot[magenta,thick] coordinates {(1,1.017189) (2,0.8471654) (3,0.7924645) (4,0.6761227) (5,0.5299689) (6,0.3333656) (7,0.1111039) (8,0.01419905) (9,0.001067001) (10,8.93806e-05) (11,5.675238e-05) (12,5.589032e-05) (13,5.588564e-05) (14,5.588169e-05) (15,5.588172e-05) (16,5.588173e-05) (17,5.588173e-05) (18,5.588173e-05) (19,5.588173e-05) (20,5.588173e-05) (21,5.588173e-05)};
\addlegendentry{GN}
\addplot[black,thick] coordinates {(1,1.017189) (2,1.00297) (3,0.9790035) (4,0.9522129) (5,0.927387) (6,0.9063422) (7,0.889166) (8,0.8752798) (9,0.8639638) (10,0.8545624) (11,0.846542) (12,0.8394913) (13,0.8331015) (14,0.8271433) (15,0.8214479) (16,0.8158903) (17,0.8103778) (18,0.8048418) (19,0.7992314) (20,0.7935094) (21,0.7876493) (22,0.7816329) (23,0.7754489) (24,0.7690912) (25,0.7625573) (26,0.7558477) (27,0.7489644) (28,0.7419104) (29,0.7346889) (30,0.7273025) (31,0.7197535) (33,0.7041701) (36,0.6795619) (38,0.6622926) (42,0.6254964) (45,0.5957856) (48,0.5642171) (52,0.5194106) (57,0.4599389) (61,0.4106593) (66,0.3487416) (71,0.2895729) (77,0.2275662) (83,0.179534) (90,0.1379866) (98,0.1040016) (105,0.08233066) (114,0.06190496) (123,0.04718934) (133,0.03534585) (144,0.02603186) (155,0.01936354) (168,0.01378746) (181,0.009902933) (196,0.006817944) (212,0.004615605) (229,0.003071541) (248,0.001962226) (267,0.001261325) (289,0.0007613788) (312,0.000453067) (338,0.0002561964) (365,0.000147676) (394,9.106833e-05) (426,6.63131e-05) (461,5.836702e-05) (498,5.642902e-05) (538,5.599949e-05) (581,5.590913e-05) (628,5.588886e-05) (679,5.588372e-05) (734,5.588224e-05) (793,5.588179e-05) (857,5.588166e-05) (926,5.588163e-05) (1001,5.588162e-05)};
\addlegendentry{WF}
\addplot[teal,thick,dashed] coordinates {(1,1.017189) (2,0.8438578) (3,0.8359365) (4,0.8028334) (5,0.7882563) (6,0.7700891) (7,0.7553382) (8,0.7399296) (9,0.7252158) (10,0.7104725) (11,0.6955825) (12,0.6794527) (13,0.6639477) (14,0.6480213) (15,0.6323653) (16,0.6156178) (17,0.5981093) (18,0.5804119) (19,0.561222) (20,0.5405742) (21,0.5181131) (22,0.4948631) (23,0.4709676) (24,0.4476507) (25,0.4237518) (26,0.3996317) (27,0.3756585) (28,0.3518328) (29,0.3285423) (30,0.3058536) (31,0.2843845) (33,0.2442975) (36,0.1922929) (38,0.1632479) (41,0.1280063) (44,0.1007219) (47,0.07960983) (51,0.05858153) (54,0.04683484) (58,0.03498051) (62,0.0263087) (67,0.01858466) (71,0.0141649) (76,0.01014514) (82,0.006858528) (88,0.004671076) (94,0.003198752) (101,0.002064344) (108,0.001336136) (116,0.0008150683) (124,0.0005015466) (133,0.0002954689) (142,0.0001822748) (152,0.0001192655) (163,9.058507e-05) (175,8.040119e-05) (187,7.769448e-05) (200,7.696484e-05) (215,7.667945e-05) (230,7.657807e-05) (246,7.655403e-05) (264,7.655559e-05) (283,7.655402e-05) (303,7.655324e-05) (325,7.655297e-05) (348,7.65529e-05) (373,7.655289e-05) (399,7.655288e-05) (428,7.655288e-05) (458,7.655288e-05)};
\addlegendentry{TWF}
\addplot[blue,thick,dashdotted] coordinates {(1,1.017189) (2,0.7733823) (3,0.7565154) (4,0.7348877) (5,0.7116456) (6,0.689406) (7,0.6677453) (8,0.6468125) (9,0.6262612) (10,0.6058477) (11,0.5850318) (12,0.5644749) (13,0.5431637) (14,0.5213898) (15,0.4974814) (16,0.4719718) (17,0.4452355) (18,0.4190142) (19,0.3917792) (20,0.3639885) (21,0.335672) (22,0.3072076) (23,0.2791248) (24,0.2513754) (25,0.2252088) (26,0.200739) (27,0.1782473) (28,0.15806) (29,0.1398393) (30,0.1237791) (31,0.1093817) (32,0.09670837) (34,0.0756694) (36,0.05963931) (39,0.04207871) (41,0.03357572) (44,0.02416882) (47,0.01757723) (50,0.01289669) (54,0.008637275) (57,0.006442545) (61,0.004394373) (65,0.003021757) (70,0.001910926) (74,0.001334358) (79,0.0008607252) (85,0.0005202893) (90,0.0003539383) (97,0.0002281696) (103,0.0001792715) (110,0.000156449) (118,0.0001484854) (125,0.0001465931) (134,0.0001459985) (143,0.0001459028) (153,0.0001458993) (163,0.0001459071) (174,0.0001459125) (186,0.0001459151) (198,0.000145916) (211,0.0001459164) (226,0.0001459165) (241,0.0001459166) (257,0.0001459166) (275,0.0001459166) (293,0.0001459166) (313,0.0001459166) (334,0.0001459166)};
\addlegendentry{TAF}
\end{loglogaxis}
\end{tikzpicture}
\hfill
\begin{tikzpicture}
\begin{loglogaxis}[width=0.46\textwidth,height=0.34\textwidth,tick label style={font=\scriptsize},label style={font=\small},title style={font=\small},grid=major,grid style={black!10},legend style={font=\scriptsize,cells={anchor=west}},legend pos=north east,xlabel={Iteration $t+1$},ylabel={Relative error},title={Poisson counts},xmin=1,xmax=1001,ymin=1e-2,ymax=2]
\addplot[magenta,thick] coordinates {(1,1.018209) (2,0.8471558) (3,0.7931251) (4,0.6762714) (5,0.5298802) (6,0.3348997) (7,0.1139043) (8,0.02484431) (9,0.02001194) (10,0.02001431) (11,0.02001471) (12,0.02001511) (13,0.02001512) (14,0.02001512) (15,0.02001512) (16,0.02001512) (17,0.02001512) (18,0.02001512) (19,0.02001512) (20,0.02001512)};
\addlegendentry{GN}
\addplot[black,thick] coordinates {(1,1.018209) (2,1.003683) (3,0.9793511) (4,0.9523786) (5,0.9275506) (6,0.9065833) (7,0.8894971) (8,0.8756873) (9,0.8644288) (10,0.8550679) (11,0.8470749) (12,0.8400419) (13,0.8336629) (14,0.8277108) (15,0.8220178) (16,0.81646) (17,0.8109457) (18,0.8054065) (19,0.7997921) (20,0.7940656) (21,0.788201) (22,0.7821804) (23,0.775993) (24,0.7696327) (25,0.7630978) (26,0.7563889) (27,0.7495085) (28,0.7424599) (29,0.7352466) (30,0.7278717) (31,0.7203376) (33,0.7047957) (36,0.6802799) (38,0.6630931) (42,0.6265155) (45,0.5970244) (48,0.5657316) (52,0.5213671) (57,0.4624524) (61,0.4134862) (66,0.3517131) (71,0.292478) (77,0.2301815) (83,0.1817575) (90,0.1398469) (98,0.1056993) (105,0.08411711) (114,0.06409131) (123,0.05006643) (133,0.03930301) (144,0.03150896) (155,0.02661277) (168,0.02325931) (181,0.02150417) (196,0.02054524) (212,0.02012885) (229,0.01997761) (248,0.01994108) (267,0.01994845) (289,0.01996754) (312,0.01998471) (338,0.0199977) (365,0.0200056) (394,0.02001021) (426,0.02001278) (461,0.02001409) (498,0.0200147) (538,0.02001497) (581,0.02001509) (628,0.02001513) (679,0.02001515) (734,0.02001516) (793,0.02001516) (857,0.02001516) (926,0.02001516) (1001,0.02001516)};
\addlegendentry{WF}
\addplot[teal,thick,dashed] coordinates {(1,1.018209) (2,0.84624) (3,0.8356564) (4,0.8029748) (5,0.7869551) (6,0.7702201) (7,0.7547601) (8,0.7391968) (9,0.7240949) (10,0.7088511) (11,0.6929189) (12,0.6761825) (13,0.6588451) (14,0.6417918) (15,0.6248641) (16,0.6063898) (17,0.5881241) (18,0.5689983) (19,0.5477345) (20,0.5256405) (21,0.5028282) (22,0.4801144) (23,0.4573225) (24,0.4342994) (25,0.4111201) (26,0.3876872) (27,0.3638697) (28,0.3406145) (29,0.3176935) (30,0.2954824) (31,0.2744152) (33,0.2357129) (36,0.1861496) (38,0.1586315) (42,0.1150208) (45,0.09089677) (48,0.07219255) (52,0.05389882) (57,0.03863722) (61,0.03049005) (66,0.02393065) (71,0.0202495) (77,0.01811308) (83,0.01723402) (90,0.01688875) (98,0.01678924) (105,0.01677485) (114,0.01679401) (123,0.01679903) (133,0.01680632) (144,0.01681314) (155,0.01681408) (168,0.01681502) (181,0.01682193) (196,0.01681711) (212,0.01681942) (229,0.01681971) (248,0.01682013) (267,0.01681834) (289,0.01681669) (312,0.01681925) (338,0.01681883) (365,0.01682353) (394,0.0168221) (426,0.01681853) (461,0.01682262) (498,0.01682221) (538,0.01681946) (581,0.0168179) (628,0.01681599) (679,0.01682441) (734,0.01681761) (793,0.01682261) (857,0.01681632) (926,0.01681805) (1001,0.01681453)};
\addlegendentry{TWF}
\addplot[blue,thick,dashdotted] coordinates {(1,1.018209) (2,0.7741887) (3,0.7572209) (4,0.7358111) (5,0.7120734) (6,0.689907) (7,0.6682733) (8,0.6479425) (9,0.6277294) (10,0.6077948) (11,0.5877601) (12,0.5676449) (13,0.5467039) (14,0.5245944) (15,0.5011073) (16,0.4760896) (17,0.4503326) (18,0.4235049) (19,0.3956859) (20,0.3671522) (21,0.3388128) (22,0.3109474) (23,0.2830001) (24,0.2559497) (25,0.2291868) (26,0.2046778) (27,0.1824262) (28,0.1617298) (29,0.1434473) (30,0.1269913) (31,0.1123724) (32,0.09940897) (34,0.07829112) (36,0.0621544) (39,0.04499884) (41,0.03697427) (44,0.02850171) (47,0.02309942) (50,0.01977263) (54,0.01733987) (57,0.0164045) (61,0.01577179) (65,0.01549844) (70,0.01536901) (75,0.01533106) (80,0.01532389) (85,0.01532589) (91,0.01533044) (97,0.01533424) (103,0.0153368) (110,0.0153386) (118,0.01533969) (126,0.01534022) (134,0.01534047) (143,0.0153406) (153,0.01534067) (163,0.01534069) (174,0.0153407) (186,0.0153407) (199,0.01534071) (212,0.01534071) (226,0.01534071) (242,0.01534071) (258,0.01534071) (275,0.01534071) (294,0.01534071) (314,0.01534071) (335,0.01534071)};
\addlegendentry{TAF}
\end{loglogaxis}
\end{tikzpicture}
\caption{Noisy CDP convergence with $n=1000$ and $L=6$. Gaussian noise has standard deviation $0.1$; the Poisson exposure is one. }
\label{fig:cdp-noisy-curves}
\end{figure}
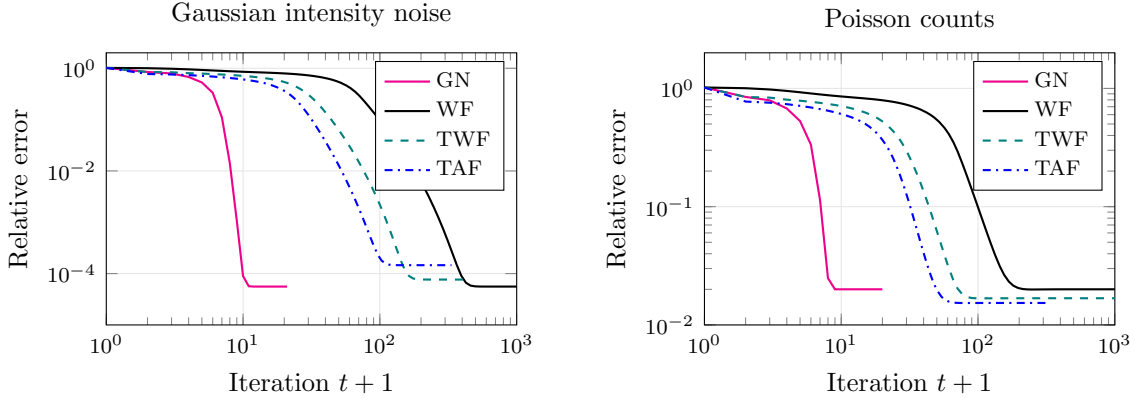

\begin{example}
We vary the Gaussian and Poisson noise levels to examine the relationship between the measurement SNR and the MSE of the recovered signal.
For each nominal SNR in \(\{20,25,\ldots,60\}\)~dB, we generate
20 independent trials with \(n=1000\) and \(L=10\).
Figure~\ref{fig:cdp-snr} reports the mean and median of the squared
relative errors for each method.  Both the mean and median errors decrease at every SNR step,
and their difference is below \(0.08\)~dB for each method and setting. GN and WF have nearly identical mean-error curves. In the tested settings, GN and WF have smaller final errors than
TWF and TAF under Gaussian noise, whereas TAF has the smallest
errors under Poisson noise. 

\end{example}

\begin{figure}[htbp]
\centering
\begin{tikzpicture}
\begin{axis}[width=0.46\textwidth,height=0.34\textwidth,tick label style={font=\scriptsize},label style={font=\small},title style={font=\small},grid=major,grid style={black!10},legend style={font=\scriptsize,cells={anchor=west}},legend pos=south west,xlabel={Realized SNR (dB)},ylabel={Mean relative MSE (dB)},title={Gaussian: mean},legend pos=north east]
\addlegendimage{magenta,thick,dashed,mark=o,mark size=3pt,mark options={solid}}
\addlegendentry{GN}
\addlegendimage{black,thick,mark=x,mark size=2pt,mark options={solid}}
\addlegendentry{WF}
\addlegendimage{teal,thick,dashed,mark=square*,mark size=1.8pt,mark options={solid}}
\addlegendentry{TWF}
\addlegendimage{blue,thick,dashdotted,mark=triangle*,mark size=2.2pt,mark options={solid}}
\addlegendentry{TAF}
\addplot[black,thick,mark=x,mark size=2pt,mark options={solid},forget plot] coordinates {(20.00221,-25.18335) (25.0006,-30.21775) (29.99427,-35.21853) (35.0133,-40.18791) (39.97627,-45.26137) (45.00652,-50.19369) (50.00355,-55.18107) (55.02044,-60.23748) (60.00636,-65.17559)};
\addplot[teal,thick,dashed,mark=square*,mark size=1.8pt,mark options={solid},forget plot] coordinates {(20.00221,-22.59619) (25.0006,-27.55448) (29.99427,-32.57405) (35.0133,-37.55027) (39.97627,-42.56085) (45.00652,-47.50177) (50.00355,-52.53174) (55.02044,-57.61716) (60.00636,-62.61463)};
\addplot[blue,thick,dashdotted,mark=triangle*,mark size=2.2pt,mark options={solid},forget plot] coordinates {(20.00221,-21.65857) (25.0006,-25.94615) (29.99427,-30.31648) (35.0133,-34.65109) (39.97627,-39.03824) (45.00652,-43.65143) (50.00355,-48.26082) (55.02044,-52.80434) (60.00636,-57.36631)};
\addplot[magenta,thick,dashed,mark=o,mark size=3pt,mark options={solid},forget plot] coordinates {(20.00221,-25.18336) (25.0006,-30.21776) (29.99427,-35.21854) (35.0133,-40.18791) (39.97627,-45.26137) (45.00652,-50.19369) (50.00355,-55.18107) (55.02044,-60.23749) (60.00636,-65.17559)};
\end{axis}
\end{tikzpicture}
\hfill
\begin{tikzpicture}
\begin{axis}[width=0.46\textwidth,height=0.34\textwidth,tick label style={font=\scriptsize},label style={font=\small},title style={font=\small},grid=major,grid style={black!10},legend style={font=\scriptsize,cells={anchor=west}},legend pos=south west,xlabel={Realized SNR (dB)},ylabel={Mean relative MSE (dB)},title={Poisson: mean},legend pos=north east]
\addlegendimage{magenta,thick,dashed,mark=o,mark size=3pt,mark options={solid}}
\addlegendentry{GN}
\addlegendimage{black,thick,mark=x,mark size=2pt,mark options={solid}}
\addlegendentry{WF}
\addlegendimage{teal,thick,dashed,mark=square*,mark size=1.8pt,mark options={solid}}
\addlegendentry{TWF}
\addlegendimage{blue,thick,dashdotted,mark=triangle*,mark size=2.2pt,mark options={solid}}
\addlegendentry{TAF}
\addplot[black,thick,mark=x,mark size=2pt,mark options={solid},forget plot] coordinates {(20.03991,-23.6373) (24.99793,-28.54661) (30.00122,-33.53468) (35.00935,-38.58756) (39.98322,-43.48994) (44.9964,-48.60461) (50.01348,-53.59594) (54.97938,-58.57975) (60.03725,-63.63209)};
\addplot[teal,thick,dashed,mark=square*,mark size=1.8pt,mark options={solid},forget plot] coordinates {(20.03991,-25.47695) (24.99793,-30.45842) (30.00122,-35.44779) (35.00935,-40.45607) (39.98322,-45.42554) (44.9964,-50.59093) (50.01348,-55.47192) (54.97938,-60.49129) (60.03725,-65.51419)};
\addplot[blue,thick,dashdotted,mark=triangle*,mark size=2.2pt,mark options={solid},forget plot] coordinates {(20.03991,-25.79168) (24.99793,-30.94631) (30.00122,-35.97144) (35.00935,-41.01441) (39.98322,-45.96444) (44.9964,-51.12282) (50.01348,-55.99472) (54.97938,-61.025) (60.03725,-66.07251)};
\addplot[magenta,thick,dashed,mark=o,mark size=3pt,mark options={solid},forget plot] coordinates {(20.03991,-23.63729) (24.99793,-28.54661) (30.00122,-33.53467) (35.00935,-38.58755) (39.98322,-43.48994) (44.9964,-48.60461) (50.01348,-53.59594) (54.97938,-58.57975) (60.03725,-63.63209)};
\end{axis}
\end{tikzpicture}
\par\medskip
\begin{tikzpicture}
\begin{axis}[width=0.46\textwidth,height=0.34\textwidth,tick label style={font=\scriptsize},label style={font=\small},title style={font=\small},grid=major,grid style={black!10},legend style={font=\scriptsize,cells={anchor=west}},legend pos=south west,xlabel={Realized SNR (dB)},ylabel={Median relative MSE (dB)},title={Gaussian: median},legend pos=north east]
\addlegendimage{magenta,thick,dashed,mark=o,mark size=3pt,mark options={solid}}
\addlegendentry{GN}
\addlegendimage{black,thick,mark=x,mark size=2pt,mark options={solid}}
\addlegendentry{WF}
\addlegendimage{teal,thick,dashed,mark=square*,mark size=1.8pt,mark options={solid}}
\addlegendentry{TWF}
\addlegendimage{blue,thick,dashdotted,mark=triangle*,mark size=2.2pt,mark options={solid}}
\addlegendentry{TAF}
\addplot[black,thick,mark=x,mark size=2pt,mark options={solid},forget plot] coordinates {(20.00221,-25.16425) (25.0006,-30.19709) (29.99427,-35.24831) (35.0133,-40.17809) (39.97627,-45.25487) (45.00652,-50.17578) (50.00355,-55.23074) (55.02044,-60.23026) (60.00636,-65.15105)};
\addplot[teal,thick,dashed,mark=square*,mark size=1.8pt,mark options={solid},forget plot] coordinates {(20.00221,-22.55356) (25.0006,-27.53331) (29.99427,-32.53904) (35.0133,-37.50522) (39.97627,-42.55004) (45.00652,-47.48737) (50.00355,-52.55646) (55.02044,-57.59629) (60.00636,-62.62584)};
\addplot[blue,thick,dashdotted,mark=triangle*,mark size=2.2pt,mark options={solid},forget plot] coordinates {(20.00221,-21.66165) (25.0006,-25.87614) (29.99427,-30.35502) (35.0133,-34.68403) (39.97627,-39.07504) (45.00652,-43.65044) (50.00355,-48.27567) (55.02044,-52.76798) (60.00636,-57.37361)};
\addplot[magenta,thick,dashed,mark=o,mark size=3pt,mark options={solid},forget plot] coordinates {(20.00221,-25.16424) (25.0006,-30.1971) (29.99427,-35.24831) (35.0133,-40.17809) (39.97627,-45.25486) (45.00652,-50.17579) (50.00355,-55.23073) (55.02044,-60.23027) (60.00636,-65.15104)};
\end{axis}
\end{tikzpicture}
\hfill
\begin{tikzpicture}
\begin{axis}[width=0.46\textwidth,height=0.34\textwidth,tick label style={font=\scriptsize},label style={font=\small},title style={font=\small},grid=major,grid style={black!10},legend style={font=\scriptsize,cells={anchor=west}},legend pos=south west,xlabel={Realized SNR (dB)},ylabel={Median relative MSE (dB)},title={Poisson: median},legend pos=north east]
\addlegendimage{magenta,thick,dashed,mark=o,mark size=3pt,mark options={solid}}
\addlegendentry{GN}
\addlegendimage{black,thick,mark=x,mark size=2pt,mark options={solid}}
\addlegendentry{WF}
\addlegendimage{teal,thick,dashed,mark=square*,mark size=1.8pt,mark options={solid}}
\addlegendentry{TWF}
\addlegendimage{blue,thick,dashdotted,mark=triangle*,mark size=2.2pt,mark options={solid}}
\addlegendentry{TAF}
\addplot[black,thick,mark=x,mark size=2pt,mark options={solid},forget plot] coordinates {(20.03991,-23.62329) (24.99793,-28.52415) (30.00122,-33.56354) (35.00935,-38.55795) (39.98322,-43.45744) (44.9964,-48.60542) (50.01348,-53.62495) (54.97938,-58.61155) (60.03725,-63.58536)};
\addplot[teal,thick,dashed,mark=square*,mark size=1.8pt,mark options={solid},forget plot] coordinates {(20.03991,-25.4317) (24.99793,-30.46054) (30.00122,-35.38621) (35.00935,-40.4925) (39.98322,-45.44771) (44.9964,-50.5954) (50.01348,-55.45272) (54.97938,-60.51262) (60.03725,-65.54883)};
\addplot[blue,thick,dashdotted,mark=triangle*,mark size=2.2pt,mark options={solid},forget plot] coordinates {(20.03991,-25.764) (24.99793,-30.95685) (30.00122,-35.93248) (35.00935,-41.00459) (39.98322,-45.94161) (44.9964,-51.16387) (50.01348,-55.99357) (54.97938,-61.00865) (60.03725,-66.11661)};
\addplot[magenta,thick,dashed,mark=o,mark size=3pt,mark options={solid},forget plot] coordinates {(20.03991,-23.62328) (24.99793,-28.52414) (30.00122,-33.56356) (35.00935,-38.55794) (39.98322,-43.45746) (44.9964,-48.60544) (50.01348,-53.62496) (54.97938,-58.61155) (60.03725,-63.58536)};
\end{axis}
\end{tikzpicture}
\caption{Noise sensitivity for $n=1000$ and $L=10$, with 20 independent trials at each nominal SNR in $\{20,25,\ldots,60\}$ dB. }
\label{fig:cdp-snr}
\end{figure}
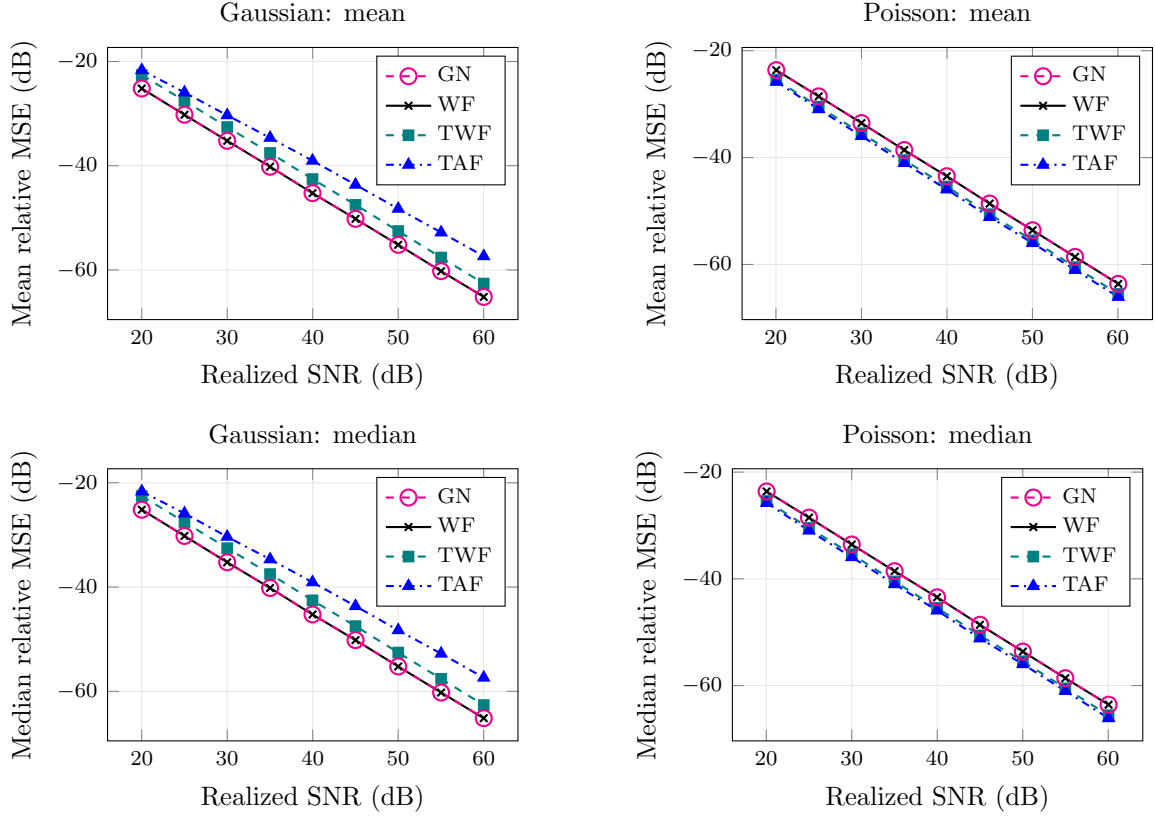

\subsection{Recovery of a natural image}
We compare the four methods on the recovery of a natural image from masked Fourier intensity measurements. The image used is the Milky Way
{Galaxy, represented as a \(1080\times1920\) RGB image.}  Each color channel is recovered
separately, with \(L=20\) octanary masks.  Each GN subproblem is solved using at most five LSQR iterations.
We record the first iterations attaining relative errors \(10^{-5}\) and \(10^{-10}\).
Elapsed times include the measured cost of spectral initialization
and exclude mask generation and file I/O.
Table~\ref{tab:cdp-image} reports the arithmetic mean of the
three channel-wise iteration counts and elapsed times at each
accuracy threshold.

All four methods recover all three channels to the stopping tolerance
in this experiment. The recovered GN channels have final relative
errors \(1.93\times10^{-12}\), \(3.23\times10^{-12}\), and
\(1.28\times10^{-12}\), respectively.
Figure~\ref{fig:cdp-image-curves} shows the largest channel error
as a function of the outer iteration. In this experiment, GN requires fewer outer iterations and less
measured time than WF, TWF, and TAF at both accuracy thresholds.

\begin{table}[htbp]
\centering
\caption{Full-resolution image recovery from 20 masks. Iterations and seconds are averaged over the RGB channels. Elapsed times include spectral initialization.}
\label{tab:cdp-image}
\begin{tabular}{lrrrr}
\hline
& \multicolumn{2}{c}{Relative error \(10^{-5}\)}
& \multicolumn{2}{c}{Relative error \(10^{-10}\)}\\
\cline{2-3}\cline{4-5}
Method & Iterations & Seconds & Iterations & Seconds\\
\hline
GN & 8.7 & 89.98 & 13.3 & 119.16\\
WF & 144.3 & 156.73 & 222.7 & 227.26\\
TWF & 79.3 & 220.99 & 149.3 & 388.05\\
TAF & 64.0 & 127.86 & 125.0 & 217.89\\
\hline
\end{tabular}
\end{table}

\begin{figure}[htbp]
\centering
\begin{tikzpicture}
\begin{loglogaxis}[width=0.46\textwidth,height=0.34\textwidth,tick label style={font=\scriptsize},label style={font=\small},title style={font=\small},grid=major,grid style={black!10},legend style={font=\scriptsize,cells={anchor=west}},legend pos=south west,xlabel={Iteration $t+1$},ylabel={Largest channel error},title={Full-resolution image recovery},xmin=1,xmax=1001,ymin=1e-12,ymax=2,width=0.75\textwidth,height=0.43\textwidth]
\addplot[magenta,thick] coordinates {(1,1.281186) (2,1.073581) (3,0.9965477) (4,0.9071508) (5,0.7591501) (6,0.4953744) (7,0.1387099) (8,0.0128959) (9,0.0006059252) (10,4.290927e-05) (11,3.415698e-06) (12,2.748849e-07) (13,2.337979e-08) (14,1.955146e-09) (15,1.703577e-10) (16,1.450905e-11) (17,3.230718e-12)};
\addlegendentry{GN}
\addplot[black,thick] coordinates {(1,1.281186) (2,1.270697) (3,1.2513) (4,1.226163) (5,1.198816) (6,1.172072) (7,1.147625) (8,1.126193) (9,1.107859) (10,1.09237) (11,1.079346) (12,1.068384) (13,1.059109) (14,1.051192) (15,1.044356) (16,1.038364) (17,1.033017) (18,1.02815) (19,1.023618) (20,1.019299) (21,1.015088) (22,1.010893) (23,1.006633) (24,1.002238) (25,0.9976458) (26,0.992801) (27,0.9876543) (28,0.9821607) (29,0.9762785) (30,0.9699675) (31,0.9631883) (32,0.9559009) (34,0.9396331) (36,0.9207943) (38,0.8989433) (41,0.8593057) (43,0.8272939) (46,0.7687499) (49,0.6947074) (52,0.6028824) (55,0.4954745) (59,0.3482291) (63,0.2291145) (66,0.1652417) (71,0.09592173) (75,0.06237279) (80,0.0365592) (85,0.0214336) (91,0.01123685) (96,0.00651128) (102,0.003342043) (109,0.001504534) (116,0.0006610942) (123,0.0002829646) (131,0.0001037135) (139,3.661372e-05) (148,1.083753e-05) (158,2.643637e-06) (168,6.062115e-07) (178,1.345218e-07) (190,2.235563e-08) (202,3.753835e-09) (215,5.481471e-10) (228,8.063817e-11) (243,8.899083e-12)};
\addlegendentry{WF}
\addplot[teal,thick,dashed] coordinates {(1,1.281186) (2,1.019456) (3,1.111698) (4,1.037193) (5,1.069191) (6,1.041243) (7,1.050109) (8,1.038158) (9,1.037851) (10,1.031107) (11,1.027438) (12,1.021901) (13,1.016917) (14,1.011255) (15,1.005427) (16,0.9991663) (17,0.9925105) (18,0.9853789) (19,0.9777223) (20,0.9694597) (21,0.9604869) (22,0.9506817) (23,0.9398871) (24,0.9279376) (25,0.9146151) (26,0.8996417) (27,0.8826643) (28,0.8632511) (29,0.8408196) (30,0.8146614) (31,0.7839161) (33,0.7045012) (35,0.5949062) (37,0.457628) (39,0.316894) (42,0.1605131) (44,0.0998664) (47,0.05008281) (49,0.0322405) (52,0.01709651) (56,0.007630037) (59,0.004260467) (62,0.002414558) (66,0.001153418) (70,0.0005601957) (74,0.0002756382) (79,0.0001152231) (84,4.876318e-05) (89,2.083388e-05) (94,8.967224e-06) (99,3.882304e-06) (105,1.430382e-06) (112,4.491287e-07) (118,1.67132e-07) (126,4.495634e-08) (133,1.430331e-08) (141,3.877524e-09) (149,1.054459e-09) (158,2.444424e-10) (168,4.833272e-11) (178,9.586003e-12)};
\addlegendentry{TWF}
\addplot[blue,thick,dashdotted] coordinates {(1,1.281186) (2,0.9202842) (3,0.9195156) (4,0.9183592) (5,0.9163125) (6,0.9126187) (7,0.9063627) (8,0.8968419) (9,0.8840267) (10,0.8686027) (11,0.8514108) (12,0.8329554) (13,0.8133811) (14,0.7926554) (15,0.770653) (16,0.7471903) (17,0.7220395) (18,0.6949057) (19,0.6654472) (20,0.6331783) (21,0.5975453) (22,0.5578005) (23,0.5130702) (24,0.462382) (25,0.405002) (26,0.3414647) (27,0.2752513) (28,0.2129672) (29,0.1602119) (30,0.1188837) (31,0.08810829) (32,0.06569105) (34,0.03748319) (36,0.02210084) (38,0.01337893) (40,0.00826784) (43,0.004137027) (45,0.002647886) (48,0.001380228) (50,0.000902612) (53,0.0004828771) (56,0.0002612379) (60,0.0001167277) (63,6.430199e-05) (67,2.927559e-05) (71,1.342623e-05) (75,6.192231e-06) (79,2.868487e-06) (83,1.333451e-06) (88,5.138381e-07) (93,1.98689e-07) (99,6.376107e-08) (104,2.478947e-08) (110,7.997845e-09) (117,2.142986e-09) (123,6.94501e-10) (131,1.550003e-10) (138,4.181609e-11) (146,9.388669e-12)};
\addlegendentry{TAF}
\end{loglogaxis}
\end{tikzpicture}
\caption{Recovery of the $1080\times1920$ RGB image from $L=20$ masks.  Each curve is the largest relative error among the three channels at the indicated outer iteration. }
\label{fig:cdp-image-curves}
\end{figure}
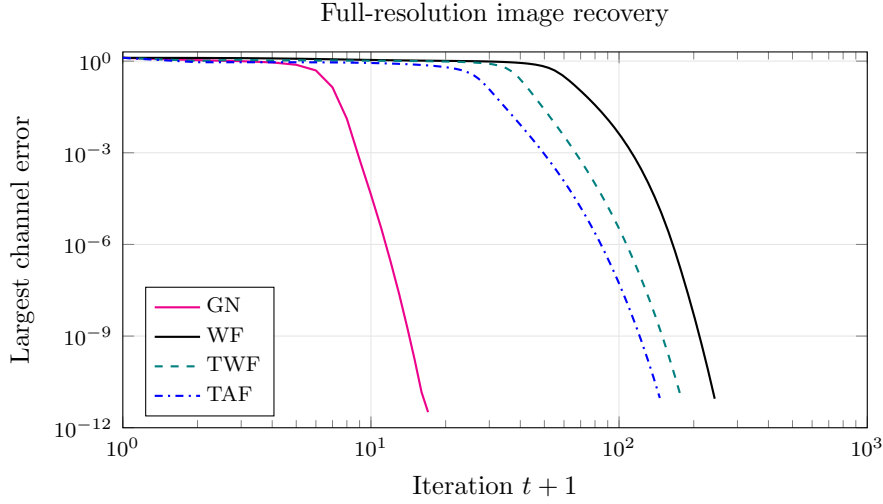

\section{{Proofs of the resampled results in Section~\ref{sec:resampled-results}}}
\label{sec:resampled-proof}

{This section proves the results in Section~\ref{sec:resampled-results}. Throughout the proofs,} \(\mathbb C^n\) is regarded as the real Hilbert
space defined by \eqref{eq:real-inner-product}.
The horizontal space at a nonzero point is defined as
\begin{equation}
 \mathcal H_{\bm z_t}
 =\{\bm h\in\mathbb C^n:\Ima(\bm z_t^*\bm h)=0\}
 =\operatorname{span}_{\mathbb R}\{\mathrm i\bm z_t\}^{\perp_{\mathbb R}}.
 \label{eq:horizontal-space}
\end{equation}
The horizontal space \(\mathcal H_{\bm z_t}\) has real dimension \(2n-1\).
Let \(U_{\bm z_t}\in\mathbb C^{n\times(2n-1)}\) have real-orthonormal
columns spanning that space, and define
\begin{equation}
 \mathsf R_{\bm z_t}
 =J_{\bm z_t}
   \begin{bmatrix}U_{\bm z_t}\\\overline{U_{\bm z_t}}\end{bmatrix}
 \in\mathbb R^{m\times(2n-1)},
 \label{rs:eq:restricted-block-jacobian}
\end{equation}
where $J_{\bm z_t}$ is given in \eqref{eq:jacobian}. The following lemma gives an equivalent form of the update that is convenient for the convergence analysis.

\begin{lemma}
\label{rs:lem:equivalent-update}
Let \(\bm z_t\ne\bm0\). The update rule \eqref{eq:gn1} is equivalent to 
\begin{equation}
\vz_{t+1}=\vz_t -U_{\bm z_t}\mathsf R_{\bm z_t}^{\dagger}\bm F(\bm z_t).
\label{rs:eq:pseudoinverse-step}
\end{equation}
\end{lemma}

\begin{proof}
By the defining property of the Moore--Penrose pseudoinverse,
\(-J_{\bm z_t}^{\dagger}\bm F(\bm z_t)\) is the unique minimum-norm
least-squares solution of
\begin{equation} \label{rs:eq:uncmini}
\min_{\bm w\in\mathbb C^{2n}}
\left\|J_{\bm z_t}\bm w+\bm F(\bm z_t)\right\|_2^2.
\end{equation}
We next show that
\begin{equation} \label{rs:eq:mininorm}
\left(J_{\bm z_t}^{\dagger}\bm F(\bm z_t)\right)_{1:n}
=
\overline{
\left(J_{\bm z_t}^{\dagger}\bm F(\bm z_t)\right)_{n+1:2n}}.
\end{equation}
Indeed, note that the two block columns of
\(J_{\bm z_t}\) in \eqref{eq:jacobian} are conjugate partners and
\(\bm F(\bm z_t)\in\mathbb R^m\) by \eqref{eq:residual}. Therefore, for any vector \(\bm w=[\bm w_1;\bm w_2] \in \C^{2n}\), we have
\[
J_{\bm z_t}
\begin{bmatrix}
\overline{\bm w_2}\\
\overline{\bm w_1}
\end{bmatrix}
+\bm F(\bm z_t)
= \overline{
J_{\bm z_t}
\begin{bmatrix}
\bm w_1\\
\bm w_2
\end{bmatrix}
+\bm F(\bm z_t)}.
\]
This implies that if \(\bm w=[\bm w_1;\bm w_2] \in \C^{2n}\) is a least-squares solution of
\eqref{rs:eq:uncmini}, then
\(\widetilde{\bm w}:=[\overline{\bm w_2}; \overline{\bm w_1}]\) is also a least-squares solution and
\(\|\widetilde{\bm w}\|_2=\|\bm w\|_2\). In particular,  \(-J_{\bm z_t}^{\dagger}\bm F(\bm z_t)\) is a
minimum-norm least-squares solution of \eqref{rs:eq:uncmini}. {Uniqueness of the minimum-norm least-squares solution therefore implies \eqref{rs:eq:mininorm}.}

Next, we prove \(\left(-J_{\bm z_t}^{\dagger}\bm F(\bm z_t)\right)_{1:n} \in\mathcal H_{\bm z_t}\), where
\(\mathcal H_{\bm z_t}\) is defined in \eqref{eq:horizontal-space}. To see this, observe that 
\[
2\operatorname{Im}\!\left(
\bm z_t^* \left(J_{\bm z_t}^{\dagger}\bm F(\bm z_t)\right)_{1:n}
\right) = \begin{bmatrix}
\mathrm{i}\bm z_t\\
\overline{\mathrm{i}\bm z_t}
\end{bmatrix}^{\!*}
\left(J_{\bm z_t}^{\dagger}\bm F(\bm z_t)\right) =0,
\]
where the first equality follows from \eqref{rs:eq:mininorm}, while the
last equality follows because the Moore--Penrose solution \(J_{\bm z_t}^{\dagger}\bm F(\bm z_t)\) lies in
\(\operatorname{range}(J_{\bm z_t}^*)\), which is orthogonal to \(\ker(J_{\bm z_t})\), and $ \begin{bmatrix}
\mathrm{i}\bm z_t\\
\overline{\mathrm{i}\bm z_t}
\end{bmatrix} \in \ker(J_{\bm z_t})$ due to the identity
\[
J_{\bm z_t}
\begin{bmatrix}
\mathrm{i}\bm z_t\\
\overline{\mathrm{i}\bm z_t}
\end{bmatrix}
=\bm0.
\]
 Consequently, \(\left(-J_{\bm z_t}^{\dagger}\bm F(\bm z_t)\right)_{1:n}\) is also
the minimum-norm solution of
\begin{equation} \label{rs:eq:equi1}
\min_{\bm h\in\mathbb C^n}\quad
\frac12
\left\|
J_{\bm z_t}
\begin{bmatrix}
\bm h\\
\overline{\bm h}
\end{bmatrix}
+\bm F(\bm z_t)
\right\|_2^2 \qquad 
\text{s.t. }\quad \operatorname{Im}(\bm z_t^*\bm h)=0.
\end{equation}
Since the columns of \(U_{\bm z_t}\) form a real-orthonormal basis of
\(\mathcal H_{\bm z_t}\), every \(\bm h\in\mathcal H_{\bm z_t}\)
has a unique representation \(\bm h=U_{\bm z_t}\bm q\) with
\(\bm q\in\mathbb R^{2n-1}\).  By the definition of  \(\mathsf R_{\bm z_t}\) given in
\eqref{rs:eq:restricted-block-jacobian}, the program \eqref{rs:eq:equi1} is equivalent to
\begin{equation} \label{rs:eq:equi2}
\min_{\bm q\in\mathbb R^{2n-1}}
\left\|\mathsf R_{\bm z_t}\bm q+\bm F(\bm z_t)\right\|_2^2.
\end{equation}
{Since \(U_{\bm z_t}\) is a real isometry, \(\|U_{\bm z_t}\bm q\|_2=\|\bm q\|_2\), minimizing the norm of \(\bm h\) is equivalent to minimizing the norm of \(\bm q\). The minimum-norm solution of \eqref{rs:eq:equi2} is}
\[
\bm q=-\mathsf R_{\bm z_t}^{\dagger}\bm F(\bm z_t),
\]
and hence
\[
\left(J_{\bm z_t}^{\dagger}\bm F(\bm z_t)\right)_{1:n}
= U_{\bm z_t}\mathsf R_{\bm z_t}^{\dagger}\bm F(\bm z_t).
\]
This completes the proof.
\end{proof}

The next lemma establishes concentration estimates for CDP sampling.

\begin{lemma} \label{rs:prop:cdp-regularity}
Assume that the entries of the masks $\dkh{\vd_l}_{l=1}^L$ are independent copies of a complex random variable $d$ satisfying \eqref{eq:mask}. {Fix \(\bm u\in\mathbb C^n\) and \(\bm v\in\mathbb C^n\setminus\{\bm0\}\), independently of the masks. Then the following statements hold.}
\begin{itemize}
\item[(1)] For any fixed  \(\delta\in(0,1)\), there is a constant \(C_{M,\delta}>0\), depending only on \(\delta\) and the mask bound \(M\), such that, if \(L\ge C_{M,\delta}\log^2 n\), then with probability at least $1-n^{-3}$,
\begin{equation}
\|\mathcal J_{\bm v}\bm h\|_2^2
\ge 2(1-\delta) \|\bm v\|_2^2\|\bm h\|_2^2
\quad
\text{for all }\bm h\in\mathcal H_{\bm v}.
\label{rs:eq:jacobian-lower}
\end{equation}
\item[(2)] For each fixed \(\beta>0\), with probability at least \(1-\tfrac18(nL)^{-\beta}\),
\begin{equation}
 \begin{aligned}
 &\left\|\frac1{nL}\sum_{\ell=1}^L\sum_{k=0}^{n-1}
       |\bm f_k^*{\bm D}_\ell^*\bm u|^2{\bm D}_\ell\bm f_k\bm f_k^*{\bm D}_\ell^*
       -(\|\bm u\|_2^2I_n+\bm u\bm u^*)\right\|_2\\
 &\qquad\le C_{M,\beta} \left(\sqrt{\frac{\log(nL)}L}+\frac{\log^2(nL)}L\right) \|\bm u\|_2^2.
 \end{aligned}
 \label{eq:common-spectral-rate}
\end{equation}
\end{itemize}
%\begin{equation}
%\Bigg\| \frac1{nL}\sum_{\ell=1}^L\sum_{k=0}^{n-1}
%|\bm f_k^*D_\ell^*\bm u|^2
%D_\ell\bm f_k\bm f_k^*D_\ell^* -\bigl(\|\bm u\|_2^2I_n+\bm u\bm u^*\bigr)
%\Bigg\|_2
%\le
%\delta\|\bm u\|_2^2.
%\label{rs:eq:fixed-vector-operator}
%\end{equation}
Here, $\mathcal J_{\bm v}$ and $\mathcal H_{\bm v}$ are defined in \eqref{eq:jac} and  \eqref{eq:horizontal-space}, respectively. 
\end{lemma}

\begin{proof}
We first prove \eqref{rs:eq:jacobian-lower}. For any $\bm h\in\mathcal H_{\bm v}$, let \(Z=\bm v\bm h^*+\bm h\bm v^*\). Then
\begin{equation}
\|Z\|_F^2 =
2\|\bm v\|_2^2\|\bm h\|_2^2
+2\operatorname{Re}\!\left[(\bm v^*\bm h)^2\right] \ge 2\|\bm v\|_2^2\|\bm h\|_2^2.
\label{rs:eq:lifted-horizontal-norm}
\end{equation}
Define the lifted measurement map $\mathcal A: \mathcal H^{n\times n} \to \R^{nL}$ on Hermitian matrices as 
\[
\mathcal A(X) = \bigl(\bm f_k^*{\bm D}_\ell^* X{\bm D}_\ell\bm f_k\bigr)_{1\le \ell \le L,\,0\le k<n},
\]
where $\mathcal H^{n\times n}$ denotes the set of all Hermitian matrices. The definition of $\mathcal J_{\bm v}\bm h$ in \eqref{eq:jac} gives 
\[
\mathcal J_{\bm v}\bm h=\frac1{\sqrt{nL}} \mathcal A(Z).
\]
Note that $Z\in\mathcal T_{\bm v}$,  where $\mathcal T_{\bm v} = \{\bm v\bm h^*+\bm h\bm v^*: \bm h\in\mathbb C^n\}$ is the tangent space to the manifold of all rank-$1$ Hermitian matrices at $\bm v \bm v^*$.  The parameterized robust-injectivity estimate in the proof of
\cite[Proposition~8, equation~(36)]{gross2017improved} gives, for every fixed \(\delta\in(0,1)\),
with probability at least \(1- n^2\exp(-c\delta^2L/M^8)\), \footnote{{The original proposition concerns real masks; the corresponding complex-mask robust-injectivity estimate is used in \cite[Section~2.1]{huang2026cdp}. Under \eqref{eq:mask}, symmetry supplies the cancellation \(\E(|d|^2d)=0\) needed in the mask-block moment argument.}}
\begin{equation}
\frac1{nL}\|\mathcal A(Z)\|_2^2
\ge
(1-\delta) \|Z\|_F^2 \qquad \mbox{for all} \quad Z\in\mathcal T_{\bm v}.
\label{rs:eq:tangent-injectivity}
\end{equation}
 {Substituting \eqref{rs:eq:lifted-horizontal-norm} into \eqref{rs:eq:tangent-injectivity} and taking} \(L\ge C_{M,\delta}\log n\) for some constant \(C_{M,\delta}\) depending only on \(\delta\) and the mask bound \(M\), we have that with probability at least $1-n^{-3}$,
 \[
 \|\mathcal J_{\bm v}\bm h\|_2^2
\ge 2(1-\delta) \|\bm v\|_2^2\|\bm h\|_2^2
\quad
\text{for all} \quad \bm h\in\mathcal H_{\bm v}.
\]
 This proves \eqref{rs:eq:jacobian-lower}.

We next prove \eqref{eq:common-spectral-rate}. First, {the bound is immediate if}
\(\bm u=\bm0\). Otherwise, by homogeneity we assume
\(\|\bm u\|_2=1\). For each complete mask, set
\begin{equation}
B_\ell
=
\frac1n\sum_{k=0}^{n-1}
|\bm f_k^*{\bm D}_\ell^*\bm u|^2
{\bm D}_\ell\bm f_k\bm f_k^*{\bm D}_\ell^*.
\label{rs:eq:mask-block-matrix}
\end{equation}
{The CDP moment identity \cite[Lemma~3.1]{candes2015cdp} gives the expectation below. The second-moment bound is dimension independent:}
\begin{equation}
\mathbb E B_\ell=I_n+\bm u\bm u^*,
\qquad
{\|\mathbb E B_\ell^2\|_2\le C_M}.
\label{eq:mask-block-moments}
\end{equation}
{
Indeed, the scalar Hoeffding inequality gives \(\Prob\{|\bm f_k^*{\bm D}_\ell^*\bm v|>a\}\le4e^{-a^2/(4M^2)}\) for every unit vector \(\bm v\) and \(a>0\); integrating this tail bounds the fourth and eighth moments by constants depending only on \(M\). Fourier orthogonality and \({\bm D}_\ell^*{\bm D}_\ell\preceq M^2I_n\) imply
\[
 B_\ell^2\preceq\frac{M^2}{n}\sum_k|\bm f_k^*{\bm D}_\ell^*\bm u|^4{\bm D}_\ell\bm f_k\bm f_k^*{\bm D}_\ell^*.
\]
For every unit vector \(\bm w\), Cauchy--Schwarz therefore gives
\[
 \bm w^*\E B_\ell^2\bm w\le\frac{M^2}{n}\sum_k\bigl(\E|\bm f_k^*{\bm D}_\ell^*\bm u|^8\bigr)^{1/2}\bigl(\E|\bm f_k^*{\bm D}_\ell^*\bm w|^4\bigr)^{1/2}\le C_M.
\]
}
Fix \(\beta>0\) and write \(m=nL\).  Introduce the event
\[
 \mathcal E_\ell
 =\left\{\max_k|\bm f_k^*{\bm D}_\ell^*\bm u|^2 \le 8M^2 \xkh{\beta+11}\log m\right\}.
\]
For any fixed \(\bm v\in\mathbb C^n\), Hoeffding's inequality applied to the real and imaginary parts gives
\begin{equation}
 \Prob\{|\bm f_k^*{\bm D}_\ell^*\bm v|> \sqrt{8M^2 \xkh{\beta+11}\log m} \|\bm v\|_2\} \le4 m^{-(2\beta+22)}.
 \label{eq:scalar-mask-tail}
\end{equation}
A union bound gives
\begin{equation}
 \Prob(\mathcal E_\ell^c)
 \le4n\cdot  m^{-(2\beta+22)}
 \le {4m^{-2\beta-21}}.
 \label{rs:eq:block-truncation-probability}
\end{equation}
On \(\bigcap_{\ell=1}^L\mathcal E_\ell\), the identity
\(\E B_\ell=I_n+\bm u\bm u^*\) gives
\begin{align}
 &\left\|\frac1L\sum_{\ell=1}^L B_\ell
                         -(I_n+\bm u\bm u^*)\right\|_2\notag\\
 &=\left\|\frac1L\sum_{\ell=1}^L
       (B_\ell\mathbf1_{\mathcal E_\ell}
                    -\E[B_\ell\mathbf1_{\mathcal E_\ell}])
       -\frac1L\sum_{\ell=1}^L
                    \E[B_\ell\mathbf1_{\mathcal E_\ell^c}]\right\|_2
       \notag\\
 &\le\left\|\frac1L\sum_{\ell=1}^L
       (B_\ell\mathbf1_{\mathcal E_\ell}
                    -\E[B_\ell\mathbf1_{\mathcal E_\ell}])\right\|_2
       +\frac1L\sum_{\ell=1}^L
             \|\E[B_\ell\mathbf1_{\mathcal E_\ell^c}]\|_2
 =:I+II.
 \label{rs:eq:IpulsII}
\end{align}

\emph{Bound for \(II\).}
Fourier orthogonality and the mask bound give
\[
 \|B_\ell\|_2
 \le\max_k|\bm f_k^*{\bm D}_\ell^*\bm u|^2
       \left\|\frac1n\sum_k{\bm D}_\ell\bm f_k\bm f_k^*{\bm D}_\ell^*\right\|_2
 \le M^4n,
\]
where \(|\bm f_k^*{\bm D}_\ell^*\bm u|\le M\sqrt n\).
Consequently,
\begin{equation}
 II\le\frac1L\sum_{\ell=1}^L
       \E\|B_\ell\mathbf1_{\mathcal E_\ell^c}\|_2
 \le M^4 n \cdot \mathbb P(\mathcal E_\ell^c)
 \le {4M^4m^{-2\beta-20}}.
 \label{rs:eq:truncation-bias}
\end{equation}

\emph{Bound for \(I\).} We next use the {matrix Bernstein inequality} to bound the term $I$. To this end, note that each matrix $B_\ell$ is {positive semidefinite}. Thus, 
\[
\left\|\mathbb E\!\left[ B_\ell\mathbf{1}_{\mathcal E_\ell}\right]\right\|_2\le \left\|\mathbb E\!  B_\ell \right\|_2 =\norm{I_n+\bm u\bm u^*}  \le2. 
\]
This gives
\begin{eqnarray*}
\norm{B_\ell\mathbf{1}_{\mathcal E_\ell} -\mathbb E[B_\ell\mathbf{1}_{\mathcal E_\ell}]} &\le & \norm{B_\ell\mathbf{1}_{\mathcal E_\ell}} +\norm{\mathbb E[B_\ell\mathbf{1}_{\mathcal E_\ell}]} \\
& \le & \max_k|\bm f_k^*{\bm D}_\ell^*\bm u|^2 \mathbf{1}_{\mathcal E_\ell} \cdot 
       \left\|\frac1n\sum_k{\bm D}_\ell\bm f_k\bm f_k^*{\bm D}_\ell^*\right\|_2 +2 \\
&= & \max_k |{\bm f_k^*{\bm D}_\ell^*}\bm u|^2  \mathbf{1}_{\mathcal E_\ell} \cdot  \|{\bm D}_\ell {\bm D}_\ell^*\|_2+2 \\
&\le & 8M^4 \xkh{\beta+11}\log(nL)+2.
\end{eqnarray*}
Moreover, we have
\[
0 \preceq \mathbb E\!\left[ \bigl(B_\ell\mathbf{1}_{\mathcal E_\ell} -\mathbb E[B_\ell\mathbf{1}_{\mathcal E_\ell}]\bigr)^2 \right] = \mathbb E\!\left[
\bigl(B_\ell\mathbf{1}_{\mathcal E_\ell} \bigr)^2 \right] - \xkh{\mathbb E[B_\ell\mathbf{1}_{\mathcal E_\ell}]}^2 
\preceq
\mathbb E B_\ell^2,
\]
and hence 
\[
\norm{\sum_{\ell=1}^L \mathbb E\!\left[ \bigl(B_\ell\mathbf{1}_{\mathcal E_\ell} -\mathbb E[B_\ell\mathbf{1}_{\mathcal E_\ell}]\bigr)^2 \right] } \le \sum_{\ell=1}^L \|\mathbb E B_\ell^2\|_2\le {C_M L}.
\]
Matrix Bernstein \cite[Theorem~1.4]{tropp2012} yields, for every \(t>0\),
\begin{equation}
 \Prob\{I>t\}
 \le2n\exp\left(-\frac{Lt^2}{{2\bigl[C_M+(8M^4(\beta+11)\log(nL)+2)t/3\bigr]}}\right).
 \label{rs:eq:truncated-matrix-bernstein}
\end{equation}
{Choose \(t=\frac12 C_{M,\beta}(\sqrt{\log(nL)/L}+\log^2(nL)/L)\), with \(C_{M,\beta}\) large enough that the exponent is at least \((\beta+1)\log m+\log32\). Since \(n\le m\), this gives}
\begin{equation} \label{rs:eq:truncation-bias1}
 \Prob\dkh{I> \frac{C_{M,\beta}}2 \left(\sqrt{\frac{\log(nL)}L}+\frac{\log^2(nL)}L\right)} \le \frac1{16}m^{-\beta}
\end{equation}
Substituting \eqref{rs:eq:truncation-bias} and \eqref{rs:eq:truncation-bias1} into \eqref{rs:eq:IpulsII} and noting that
\({m^{-2\beta-20}}\le m^{-1}\le\sqrt{\log(nL)/L}\),  we obtain \eqref{eq:common-spectral-rate} for unit \(\bm u\), with total failure probability at most
\[
 \frac1{16}m^{-\beta}
       +\sum_{\ell=1}^L\Prob(\mathcal E_\ell^c)
 \le\frac1{16}m^{-\beta}+{4Lm^{-2\beta-21}}
 \le\frac18m^{-\beta}.
\]
The bias is absorbed by increasing \(C_{M,\beta}\).  This completes the proof.
\end{proof}

{We now prove the main results in Section~\ref{sec:resampled-results}.}

\subsection{Proof of Theorem \ref{rs:cor:spectral-initialization}} \label{pr:spec}
\begin{proof}
Let \(m=nL\). By the definition of \(Y\) in \eqref{eq:init}, the observation model $y_{\ell,k}=|\bm f_k^*{\bm D}_\ell^*\bm x_\star|^2 +\xi_{\ell,k}$, and the triangle inequality, we obtain, with probability at least $1-n^{-3}$,
\begin{align}
 &\left\|Y-(\bm x_\star\bm x_\star^*
                    +\|\bm x_\star\|_2^2I_n)\right\|_2\notag\\
 &\le\left\|\frac1m\sum_{\ell,k}
       |\bm f_k^*{\bm D}_\ell^*\bm x_\star|^2
       {\bm D}_\ell\bm f_k\bm f_k^*{\bm D}_\ell^*
       -(\bm x_\star\bm x_\star^*+\|\bm x_\star\|_2^2I_n)
       \right\|_2+\left\|\frac1m\sum_{\ell,k}\xi_{\ell,k}
                      {\bm D}_\ell\bm f_k\bm f_k^*{\bm D}_\ell^*\right\|_2\notag\\
 &\le\frac{\|\bm x_\star\|_2^2}{128}
       +\|\bm\xi\|_\infty
           \left\|\frac1m\sum_{\ell,k}
                     {\bm D}_\ell\bm f_k\bm f_k^*{\bm D}_\ell^*\right\|_2\notag\\
 &=\frac{\|\bm x_\star\|_2^2}{128}
       +\|\bm\xi\|_\infty
           \left\|\frac1L\sum_{\ell=1}^L {\bm D}_\ell {\bm D}_\ell^*\right\|_2\notag\\
 &\le\frac{\|\bm x_\star\|_2^2}{128}
        +M^2\|\bm\xi\|_\infty
 \le\frac{\|\bm x_\star\|_2^2}{64},
 \label{rs:eq:noisy-initializer-matrix-event}
\end{align}
The second inequality uses part~(2) of Lemma~\ref{rs:prop:cdp-regularity} with \(\bm u=\bm x_\star\), \(\beta=3\), and \(C_M\) large enough to make the relative bound at most \(1/128\), together with the positive semidefiniteness of each \({\bm D}_\ell\bm f_k\bm f_k^*{\bm D}_\ell^*\). The last inequality uses the assumption $M^2\|\bm\xi\|_\infty \le\frac{\|\bm x_\star\|_2^2}{128}$.
Let \(\bm u_\star=\bm x_\star/\|\bm x_\star\|_2\), and let \(\bm v_0\) be a unit
leading eigenvector of \(Y\), with eigenvalue \(\lambda_1(Y)\).  {Weyl's inequality}
and \eqref{rs:eq:noisy-initializer-matrix-event} give
\begin{align*}
 \lambda_1(Y) \ge2\|\bm x_\star\|_2^2
   -\|Y-(\bm x_\star\bm x_\star^*
                   +\|\bm x_\star\|_2^2I_n)\|_2 \ge(2-1/64)\|\bm x_\star\|_2^2.
\end{align*}
The Davis--Kahan \(\sin\Theta\) theorem \cite{davis1970rotation} gives
\begin{align*}
 \dist(\bm v_0,\bm u_\star)
 &\le
 \frac{\sqrt2\|Y-(\bm x_\star\bm x_\star^*
                     +\|\bm x_\star\|_2^2I_n)\|_2}
      {\lambda_1(Y)-\|\bm x_\star\|_2^2}\\
 &\le\frac{\sqrt2\,\|\bm x_\star\|_2^2/64}
              {[(2-1/64)-1]\|\bm x_\star\|_2^2}
 =\frac{\sqrt2}{63}.
\end{align*}

It remains to control the scale \(\lambda\). {Observe that} 
\[
\frac1m\sum_{\ell=1}^L\sum_{k=0}^{n-1}|\bm f_k^*{\bm D}_\ell^*\bm x_\star|^2
= \frac1L\sum_{\ell=1}^{L} \bm x_\star^*{\bm D}_{\ell}{\bm D}_{\ell}^*\bm x_\star .
\]
The random variables $\bm x_\star^*{\bm D}_{\ell}{\bm D}_{\ell}^*\bm x_\star$ are independent, have mean \(\|\bm x_\star\|_2^2\), and belong to
\([0,M^2\|\bm x_\star\|_2^2]\). Hoeffding's inequality gives, when \(L\ge C_M\log n\), with probability at least \(1-n^{-3}\),
\[
\left|
\frac1m\sum_{\ell=1}^L\sum_{k=0}^{n-1} |\bm f_k^*{\bm D}_\ell^*\bm x_\star|^2-\|\bm x_\star\|_2^2
\right|
\le\frac{\|\bm x_\star\|_2^2}{128}.
\]
This gives
\begin{align*}
 |\lambda-\|\bm x_\star\|_2^2|
 &=\left|\frac1m\sum_{\ell,k}|\bm f_k^*{\bm D}_\ell^*\bm x_\star|^2
       -\|\bm x_\star\|_2^2+\frac1m\sum_{\ell,k}\xi_{\ell,k}\right|\\
 &\le\frac{\|\bm x_\star\|_2^2}{128}
       +\left|\frac1m\sum_{\ell,k}\xi_{\ell,k}\right|
 \le\frac{\|\bm x_\star\|_2^2}{64},
\end{align*}
where the last inequality follows from 
\[
\left| \frac1m\sum_{\ell=1}^L\sum_{k=0}^{n-1}\xi_{\ell,k}\right|
\le\frac{\|\bm\xi\|_1}{m}
\le\frac{\|\bm x_\star\|_2^2}{128}.
\]
In particular, \(\lambda\ge(63/64)\|\bm x_\star\|_2^2>0\).
This gives 
\[
\left|\sqrt\lambda-\|\bm x_\star\|_2\right| =\frac{\Abs{\lambda - \|\bm x_\star\|_2^2}}{\sqrt\lambda+\|\bm x_\star\|_2}
\le\frac{\|\bm x_\star\|_2}{64}.
\]
  Since
\(\bm z_0=\sqrt{\lambda}\,\bm v_0\) and
\(\bm x_\star=\|\bm x_\star\|_2\bm u_\star\), the triangle inequality gives
\begin{align*}
\dist(\bm z_0,\bm x_\star)
&= \min_{\phi \in \R} \left\|\sqrt{\lambda}\,\bm v_0
       -\|\bm x_\star\|_2\bm u_\star e^{i\phi}\right\|_2 \\
&=  \min_{\phi \in \R} \left\|\|\bm x_\star\|_2(\bm v_0-\bm u_\star e^{i\phi} )
       +(\sqrt{\lambda}-\|\bm x_\star\|_2)\bm v_0\right\|_2 \\
&\le \|\bm x_\star\|_2\,\dist(\bm v_0,\bm u_\star)
       +\left|\sqrt{\lambda}-\|\bm x_\star\|_2\right| \\
       & \le \xkh{\frac{\sqrt{2}}{63}+\frac1{64}}\|\bm x_\star\|_2
       <\frac18\|\bm x_\star\|_2.
\end{align*}
{Both required events hold with probability at least \(1-2n^{-3}\), provided \(L\ge C_M\log^2 n\). This completes the proof.}
\end{proof}

\subsection{Proof of Theorem~\ref{rs:thm:complete}}
\label{sec:resampled-onestep-proof}

\begin{proof}
{Let \(m=nL\), and condition on the current iterate \(\bm z_t\), which is independent of the current mask block. The basin assumption gives \(\|\bm z_t\|_2\ge7\|\bm x_\star\|_2/8>0\). By \eqref{eq:alignment},}
\[
\bm z_t-\widetilde{\bm x}_\star(\bm z_t)\in\mathcal H_{\bm z_t},
\qquad
\|\bm z_t-\widetilde{\bm x}_\star(\bm z_t)\|_2
=
\dist(\bm z_t,\bm x_\star).
\]
For any measurement vector \({\bm D}_\ell\bm f_k\), the definition
\eqref{eq:jac} gives
\begin{align}
 &\sqrt m\,
 [\mathcal J_{\bm z_t}
       (\bm z_t-\widetilde{\bm x}_\star(\bm z_t))]_{\ell,k}\notag\\
 &=2\Rea\left[
 \overline{\bm f_k^*{\bm D}_\ell^*\bm z_t}\,
 \bm f_k^*{\bm D}_\ell^*(\bm z_t-\widetilde{\bm x}_\star(\bm z_t))
 \right]\notag\\
 &=2\Rea\left[
 \overline{\bm f_k^*{\bm D}_\ell^*\widetilde{\bm x}_\star(\bm z_t)}\,
 \bm f_k^*{\bm D}_\ell^*(\bm z_t-\widetilde{\bm x}_\star(\bm z_t))
 \right]+2|\bm f_k^*{\bm D}_\ell^*
             (\bm z_t-\widetilde{\bm x}_\star(\bm z_t))|^2\notag\\
 &=|\bm f_k^*{\bm D}_\ell^*\bm z_t|^2
   -|\bm f_k^*{\bm D}_\ell^*\bm x_\star|^2
   +|\bm f_k^*{\bm D}_\ell^*
             (\bm z_t-\widetilde{\bm x}_\star(\bm z_t))|^2,
 \label{rs:eq:jacobian-expansion}
\end{align}
where the last equality follows by expanding the squared magnitude:
\[
|\bm f_k^*{\bm D}_\ell^*\bm z_t|^2
       -|\bm f_k^*{\bm D}_\ell^*\bm x_\star|^2=2\Rea\left[
 \overline{\bm f_k^*{\bm D}_\ell^*\widetilde{\bm x}_\star(\bm z_t)}\,
 \bm f_k^*{\bm D}_\ell^*(\bm z_t-\widetilde{\bm x}_\star(\bm z_t))
 \right]+|\bm f_k^*{\bm D}_\ell^*
             (\bm z_t-\widetilde{\bm x}_\star(\bm z_t))|^2.
\]
Combining \eqref{rs:eq:jacobian-expansion} with \eqref{eq:residual}
and the observation model gives
\begin{equation}
\mathcal J_{\bm z_t}
       (\bm z_t-\widetilde{\bm x}_\star(\bm z_t))-\bm F(\bm z_t)=\frac1{\sqrt m}
   \left(|\bm f_k^*{\bm D}_\ell^*
       (\bm z_t-\widetilde{\bm x}_\star(\bm z_t))|^2\right)_{\ell,k}
       +\frac{\bm\xi}{\sqrt m}.
 \label{rs:eq:componentwise-cancellation}
\end{equation}
By the definition of \(\widetilde{\bm x}_\star(\bm z_t)\) in \eqref{eq:alignment}, we have \(\bm z_t-\widetilde{\bm x}_\star(\bm z_t) \in \mathcal H_{\bm z_t}\). Hence, 
\[
\bm z_t-\widetilde{\bm x}_\star(\bm z_t)=U_{\bm z_t}\bm q
\]
for some \(\bm q\in\mathbb R^{2n-1}\).
For any $0<\delta<1$, apply Lemma~\ref{rs:prop:cdp-regularity} with \(\bm v=\bm z_t\) and \(\bm u=\bm z_t-\widetilde{\bm x}_\star(\bm z_t)\). If \(L\ge {C_{M,\delta}} \log^2 n\), where ${C_{M,\delta}}>0$ depends only on $\delta, M$, then, with probability at least $1-n^{-3}$,
\begin{equation}
\sigma_{\min}(\mathsf R_{\bm z_t})
\ge \sqrt{2(1-\delta)}\,\|\bm z_t\|_2.
\label{rs:eq:horizontal-singular-value-step}
\end{equation}
Therefore,  the matrix \(\mathsf R_{\bm z_t} \in \mathbb R^{m\times (2n-1)}\) has full column rank. Thus,  \(\mathsf R_{\bm z_t}^\dagger\mathsf R_{\bm z_t}=I_{2n-1}\), and hence
\begin{align}
 \bm z_t-\widetilde{\bm x}_\star(\bm z_t)
 &=U_{\bm z_t}\bm q
 =U_{\bm z_t}\mathsf R_{\bm z_t}^\dagger\mathsf R_{\bm z_t}\bm q
 \notag\\
 &=U_{\bm z_t}\mathsf R_{\bm z_t}^\dagger
       J_{\bm z_t}
       \begin{bmatrix}U_{\bm z_t}\\\overline{U_{\bm z_t}}\end{bmatrix}
       \bm q\notag\\
 &=U_{\bm z_t}\mathsf R_{\bm z_t}^\dagger
       \mathcal J_{\bm z_t}
       \bigl(\bm z_t-\widetilde{\bm x}_\star(\bm z_t)\bigr).
 \label{rs:eq:zminusx}
\end{align}
Here, the second line uses \eqref{rs:eq:restricted-block-jacobian}, and the last line uses \(\bm q\in\mathbb R^{2n-1}\), so that
\(\overline{U_{\bm z_t}\bm q}=\overline{U_{\bm z_t}}\bm q\).
Also, \(U_{\bm z_t}\) is an isometry on real vectors.
Applying Lemma ~\ref{rs:lem:equivalent-update}  and using \eqref{rs:eq:componentwise-cancellation} and \eqref{rs:eq:zminusx}, we have
\begin{align}
 &\dist(\bm z_{t+1},\bm x_\star)\notag\\
 &\le\left\|\bm z_t-\widetilde{\bm x}_\star(\bm z_t)
       -U_{\bm z_t}\mathsf R_{\bm z_t}^\dagger\bm F(\bm z_t)
       \right\|_2\notag\\
 &=\left\|U_{\bm z_t}\mathsf R_{\bm z_t}^\dagger
       \left[\mathcal J_{\bm z_t}
        (\bm z_t-\widetilde{\bm x}_\star(\bm z_t))-\bm F(\bm z_t)
       \right]\right\|_2\notag\\
 &=\left\|\mathsf R_{\bm z_t}^\dagger
       \left[\frac1{\sqrt m}
       \left(|\bm f_k^*{\bm D}_\ell^*
        (\bm z_t-\widetilde{\bm x}_\star(\bm z_t))|^2\right)_{\ell,k}
       +\frac{\bm\xi}{\sqrt m}\right]\right\|_2.
 \label{rs:eq:distzk1}
\end{align}
{As noted above, \(\bm z_t\ne\bm0\).}
By \eqref{rs:eq:horizontal-singular-value-step}, one has
\begin{equation}
\|\mathsf R_{\bm z_t}^\dagger\|_2
\le \frac{1} {\sqrt{2(1-\delta)}\,\|\bm z_t\|_2}.
\label{rs:eq:pseudoinverse-bound}
\end{equation}
Part~(2) of Lemma~\ref{rs:prop:cdp-regularity} gives, with probability at least \(1-n^{-3}\),
\begin{equation}
\left\|
\frac1{\sqrt{m}} \left( \left| \bm f_k^*{\bm D}_\ell^* \bigl(\bm z_t-\widetilde{\bm x}_\star(\bm z_t)\bigr) \right|^2
\right)_{\ell,k} \right\|_2
\le  \sqrt{2+\delta}\, \dist^2(\bm z_t,\bm x_\star),
\label{rs:eq:remainder-bound}
\end{equation}
provided  \(L\ge {C_{M,\delta}} \log^2 n\).
{On the intersection of these two events, which has conditional probability at least \(1-2n^{-3}\), substituting \eqref{rs:eq:pseudoinverse-bound} and \eqref{rs:eq:remainder-bound} into \eqref{rs:eq:distzk1} yields}
\begin{equation}
\dist(\bm z_{t+1},\bm x_\star) \le
\sqrt{\frac{2+\delta}{2(1-\delta)}} \cdot \frac{\dist(\bm z_t,\bm x_\star)^2}{\|\bm z_t\|_2}+ \frac{\|\bm\xi\|_2} {\sqrt{2(1-\delta)m}\,\|\bm z_t\|_2}.
\label{rs:eq:absolute-one-step}
\end{equation}
Finally,  the reverse triangle inequality gives
\begin{equation}
\|\bm z_t\|_2
\ge
\|\bm x_\star\|_2-\dist(\bm z_t,\bm x_\star)
\ge \frac78 \|\bm x_\star\|_2.
\label{rs:eq:iterate-norm-lower}
\end{equation}
Substituting this lower bound into \eqref{rs:eq:absolute-one-step}
and using \(8/7\le2\) gives the conclusion of Theorem~\ref{rs:thm:complete} with the stated \(\kappa_\delta\).
\end{proof}

\section{Proofs of the fixed-data results in {Section} \ref{sec:fixed-results}}
\label{sec:fixed-proof}

To prove Theorem~\ref{thm:noisy}, we first establish deterministic estimates throughout a neighborhood and then verify their hypotheses for random masks.

\subsection{A deterministic contraction argument} \label{sec:deterministic}
For a set of measurement pairs
\(S\subseteq\{1,\ldots,L\}\times\{0,\ldots,n-1\}\), \(\mathcal J_{\bm z,S}\) denotes the row restriction of \(\mathcal J_{\bm z}\) to \(S\), where $\mathcal J_{\bm z}$ is defined in \eqref{eq:jac}. Fourier orthogonality {gives} 
\begin{equation}
 \frac1m\sum_{\ell,k}|\bm f_k^*{\bm D}_\ell^*\bm v|^2
 \le M^2\|\bm v\|_2^2,
 \qquad \max_{\ell,k}\|{\bm D}_\ell\bm f_k\|_2\le M\sqrt n.
 \label{eq:frame}
\end{equation}
For any vector  $\vz \in \C^n$, define
\begin{equation}
 \mathcal H_{\bm z}
 =\{\bm h\in\mathbb C^n:\Ima(\bm z^*\bm h)=0\}
 =\operatorname{span}_{\mathbb R}\{\mathrm i\bm z\}^{\perp_{\mathbb R}}.
 \label{eq:horizontal-space1}
\end{equation}

\begin{proposition} \label{prop:linear}
{Fix \(\bm x_\star\ne\bm0\) and deterministic masks satisfying \eqref{eq:frame}. Assume that}
\begin{equation}
 \frac32\|\bm x_\star\|_2^2\|\bm v\|_2^2
 \le\|\mathcal J_{\bm x_\star}\bm v\|_2^2
 \le9\|\bm x_\star\|_2^2\|\bm v\|_2^2,
 \quad\bm v\in\mathcal H_{\bm x_\star}.
 \label{eq:target-gram}
\end{equation}
Suppose that, for some \(0<\gamma\le1\),
\begin{equation}
 \|\mathcal J_{\bm x_\star,S}\|_2
 \le\frac{\|\bm x_\star\|_2}{16}
 \quad\text{for all } S\text{ with }|S|\le\gamma m.
 \label{eq:erasure}
\end{equation}
{Suppose that the current iterate satisfies}
\begin{equation}
 \frac{\dist(\bm z_t,\bm x_\star)}{\|\bm x_\star\|_2}
 \le\min\left\{\frac14,\frac{\sqrt\gamma}{64M^2}\right\}.
 \label{eq:prop-basin}
\end{equation}
Then the update $ \bm z_{t+1}=\bm z_t-\mathcal J_{\bm z_t}^{\dagger}\bm F_{\mbox{clean}}(\bm z_t)$ satisfies
\begin{equation}
 \dist(\bm z_{t+1},\bm x_\star) \le\frac{15}{16}\dist(\bm z_t,\bm x_\star).
 \label{eq:prop-contraction}
\end{equation}
{Moreover, \(\bm z_t\ne\bm0\) and \(\|\mathcal J_{\bm z_t}^{\dagger}\|_2\le1/\|\bm x_\star\|_2\).}
Here, $\mathcal J_{\bm z_t}$ is given in \eqref{eq:jac}, and $\bm F_{\mbox{clean}}(\bm z_t)$ is given in \eqref{eq:residual} with $ y_{\ell,k}=|\bm f_k^*{\bm D}_\ell^*\bm x_\star|^2$.
\end{proposition}

\begin{proof}
By scaling \(\bm x_\star\) and \(\bm z_t\) by \(\|\bm x_\star\|_2\), we may assume \(\|\bm x_\star\|_2=1\).  Using the phase alignment in \eqref{eq:phase-distance} and \eqref{eq:alignment}, we choose a representative \(\bm z_t\) with \(\bm x_\star^*\bm z_t \in \R_+\) and set $\bm h=\bm z_t-\bm x_\star$. The definition of $\mathcal H_{\bm z_t}$ in \eqref{eq:horizontal-space1} then gives \(\bm h\in\mathcal H_{\bm z_t}\).  Let
\begin{equation} \label{eq:setS}
 S=\{(\ell,j):|\bm f_j^*{\bm D}_\ell^*\bm h|>1/(64M)\}.
\end{equation}
Note that by assumption \eqref{eq:prop-basin}, $\norm{\vh} \le \frac{\sqrt\gamma}{64M^2}$. Therefore, the definition of $S$ together with ~\eqref{eq:frame} gives
\begin{equation} \label{eq:numS}
 |S|\le4096M^4m \norm{\vh}^2 \le\gamma m.
 \end{equation}
We claim that, for every nonzero \(\bm v\in\mathcal H_{\bm z_t}\),
\begin{equation}
\|\bm v\|_2  \le   \|\mathcal J_{\bm z_t,S^c}\bm v\|_2  \le \sqrt{10}\|\bm v\|_2.
 \label{cl:good-upper}
\end{equation}
Note that  $ \bm z_{t+1}=\bm z_t-\mathcal J_{\bm z_t}^{\dagger}\bm F_{\mbox{clean}}(\bm z_t)$. The definitions of $\mathcal J_{\bm z_t}$ and $\bm F_{\mbox{clean}}(\bm z_t)$ give 
\[
\bm F_{\mbox{clean}}(\bm z_t)=\mathcal J_{\bm z_t}\bm h-\bm b,
\]
where \(\bm b=(|\bm f_j^*{\bm D}_\ell^*\bm h|^2/\sqrt m)_{\ell,j} \in \R^m\) is the quadratic remainder.  This immediately gives
\begin{equation} \label{eq:upzt}
 \bm z_{t+1}-\bm x_\star=\vh-\mathcal J_{\bm z_t}^{\dagger}\bm F_{\mbox{clean}}(\bm z_t) =\mathcal J_{\bm z_t}^\dagger\bm b.
\end{equation}
{The last equality uses \(\mathcal J_{\bm z_t}^\dagger\mathcal J_{\bm z_t}\bm h=\bm h\): the claim \eqref{cl:good-upper} implies \(\ker\mathcal J_{\bm z_t}=\operatorname{span}_{\mathbb R}\{\mathrm i\bm z_t\}\), and \(\bm h\in\mathcal H_{\bm z_t}\) is orthogonal to this kernel.} Use \(U_{\bm z_t}\) and \(\mathsf R_{\bm z_t}\) from
\eqref{rs:eq:restricted-block-jacobian}. Denote by
\[
 \mathsf R_{\bm z_t,S^c}\in\mathbb R^{|S^c|\times(2n-1)},
 \qquad
 \mathsf R_{\bm z_t,S}\in\mathbb R^{|S|\times(2n-1)}
\]
the row submatrices of \(\mathsf R_{\bm z_t}\). {In the block formulas below, rows and residual entries are ordered consistently with \(S^c\) first and \(S\) second.} For every \(\bm q\in\mathbb R^{2n-1}\),  the definitions give
\begin{equation} \label{eq:property3}
 \mathsf R_{\bm z_t,S^c}\bm q
 =\mathcal J_{\bm z_t,S^c}(U_{\bm z_t}\bm q),
 \qquad
 \mathsf R_{\bm z_t,S}\bm q
 =\mathcal J_{\bm z_t,S}(U_{\bm z_t}\bm q).
\end{equation}
Moreover, we have
\begin{equation} \label{eq:property2}
\norm{ U_{\bm z_t} \bm q} =\norm{\bm q} \qquad \mbox{for all} \quad \bm q \in \R^{2n-1}.
\end{equation}
Indeed, $\norm{ U_{\bm z_t} \bm q}^2= \bm q^\top U_{\bm z_t}^*U_{\bm z_t} \bm q=\bm q^\top  \Rea(U_{\bm z_t}^*U_{\bm z_t})\bm q =\norm{\bm q}^2$ due to $\Rea(U_{\bm z_t}^*U_{\bm z_t})=I_{2n-1}$.
Because \(U_{\bm z_t}\bm q\in\mathcal H_{\bm z_t}\) and
\(\|U_{\bm z_t}\bm q\|_2=\|\bm q\|_2\), the bounds \eqref{cl:good-upper} give
\[
 \|\bm q\|_2^2
 \le\|\mathsf R_{\bm z_t,S^c}\bm q\|_2^2
 \le10\|\bm q\|_2^2,
 \qquad \bm q\in\mathbb R^{2n-1},
\]
which means 
\begin{equation} \label{eq:RzSbound}
I_{2n-1}\preceq
 \mathsf R_{\bm z_t,S^c}^{\top}\mathsf R_{\bm z_t,S^c}
 \preceq10I_{2n-1}
\end{equation}
and 
\begin{equation} \label{eq:Rzbound}
 \mathsf R_{\bm z_t}^{\top}\mathsf R_{\bm z_t}
 =\mathsf R_{\bm z_t,S^c}^{\top}\mathsf R_{\bm z_t,S^c}
  +\mathsf R_{\bm z_t,S}^{\top}\mathsf R_{\bm z_t,S}
 \succeq I_{2n-1}.
\end{equation}
{Applying the minimum-norm characterization in Lemma~\ref{rs:lem:equivalent-update} to the real right-hand side \(\bm b\) gives} $\mathcal J_{\bm z_t}^{\dagger}\bm b=U_{\bm z_t} \mathsf R_{\bm z_t}^{\dagger}  \bm b$. Note that  \(\mathsf R_{\bm z_t}\) has full column rank. Thus, by \eqref{eq:upzt}, we have
\begin{eqnarray*}
 \bm z_{t+1}-\bm x_\star &= & \mathcal J_{\bm z_t}^{\dagger}\bm b=U_{\bm z_t} \mathsf R_{\bm z_t}^{\dagger}  \bm b=U_{\bm z_t} (\mathsf R_{\bm z_t}^{\top}\mathsf R_{\bm z_t})^{-1}  \mathsf R_{\bm z_t}^{\top} \bm b \\
 &=& U_{\bm z_t} (\mathsf R_{\bm z_t}^{\top}\mathsf R_{\bm z_t})^{-1} \left(\mathsf R_{\bm z_t,S^c}^{\top}\bm b_{S^c}
         +\mathsf R_{\bm z_t,S}^{\top}\bm b_S\right)\notag\\
  &\overset{\text{(i)}}{=}&    \frac12U_{\bm z_t} (\mathsf R_{\bm z_t}^{\top}\mathsf R_{\bm z_t})^{-1}
    \mathsf R_{\bm z_t,S}^{\top}\mathsf R_{\bm z_t,S}
    \Rea(U_{\bm z_t}^*\bm h)\\
   &&  +U_{\bm z_t} (\mathsf R_{\bm z_t}^{\top}\mathsf R_{\bm z_t})^{-1}  \left(\mathsf R_{\bm z_t,S^c}^{\top}\bm b_{S^c}
         -\frac12\mathsf R_{\bm z_t,S}^{\top}   \mathcal J_{\bm x_\star,S}\bm h\right)\notag\\
 &=& \frac12U_{\bm z_t}
    (\mathsf R_{\bm z_t}^{\top}\mathsf R_{\bm z_t})^{-1}
    \mathsf R_{\bm z_t,S}^{\top}\mathsf R_{\bm z_t,S}
    \Rea(U_{\bm z_t}^*\bm h)\notag+ U_{\bm z_t} \mathsf R_{\bm z_t}^{\dagger} 
    \begin{bmatrix}
       \bm b_{S^c}\\-\frac12\mathcal J_{\bm x_\star,S}\bm h
    \end{bmatrix}\\
       &:=& I+II,
 \label{eq:good-bad}
\end{eqnarray*}
To justify (i), the definition of $\mathcal J_{\bm z}$ in \eqref{eq:jac} and the identity $\vh=\vz_t-\vx_\star$ give 
\begin{eqnarray*}
 \bm b_S=(\frac1{\sqrt m}|\bm f_j^*{\bm D}_\ell^*\bm h|^2)_{S}=\tfrac12\mathcal J_{\bm h,S}\bm h &=& \tfrac12\mathcal J_{\bm z_t,S}\bm h -\tfrac12\mathcal J_{\bm x_\star,S}\bm h  \\
 &=& \frac12\mathsf R_{\bm z_t,S}\Rea(U_{\bm z_t}^*\bm h)-\frac12\mathcal J_{\bm x_\star,S}\bm h.
\end{eqnarray*}
Here, we use the fact that  $\bm h=U_{\bm z_t}\Rea(U_{\bm z_t}^*\bm h)$ and the relationship between ${\mathcal J}_{\bm z_t,S}$ and $\mathsf R_{\bm z_t,S}$ in \eqref{eq:property3}.  To verify $\bm h=U_{\bm z_t}\Rea(U_{\bm z_t}^*\bm h)$, note that  the columns of  \(U_{\bm z_t}\in\mathbb C^{n\times(2n-1)}\) {form a real-orthonormal basis of \(\mathcal H_{\bm z_t}\)}, which means $\Rea(U_{\bm z_t}^*U_{\bm z_t})=I_{2n-1}$. Since \(\bm h\in\mathcal H_{\bm z_t}\), this gives 
\begin{equation} \label{eq:property1}
\bm h=U_{\bm z_t} \bm q= U_{\bm z_t} \Rea(U_{\bm z_t}^*U_{\bm z_t}) \bm q =U_{\bm z_t} \Rea(U_{\bm z_t}^*U_{\bm z_t} \bm q) =U_{\bm z_t}\Rea(U_{\bm z_t}^*\bm h)
\end{equation}
for some $\bm q \in \R^{2n-1}$. 

We next give {upper bounds} for $I$ and $II$.  First, applying Lemma~\ref{lem:half} on \(\mathbb R^{2n-1}\) to the matrices
\[
 \mathsf R_{\bm z_t,S^c}^{\top}\mathsf R_{\bm z_t,S^c}
 \quad\text{and}\quad
 \mathsf R_{\bm z_t,S}^{\top}\mathsf R_{\bm z_t,S},
\]
together with \eqref{eq:RzSbound}, we have
\[
 \left\|(\mathsf R_{\bm z_t}^{\top}\mathsf R_{\bm z_t})^{-1}
          \mathsf R_{\bm z_t,S}^{\top}\mathsf R_{\bm z_t,S}\right\|_2
 \le\frac{11}{2\sqrt{10}}.
\]
The vector
\((\mathsf R_{\bm z_t}^{\top}\mathsf R_{\bm z_t})^{-1}
\mathsf R_{\bm z_t,S}^{\top}\mathsf R_{\bm z_t,S}
\Rea(U_{\bm z_t}^*\bm h)\) is real.  Therefore,
\[
\norm{I} \le \frac12
  \left\|(\mathsf R_{\bm z_t}^{\top}\mathsf R_{\bm z_t})^{-1}
          \mathsf R_{\bm z_t,S}^{\top}\mathsf R_{\bm z_t,S}\right\|_2
  \|\Rea(U_{\bm z_t}^*\bm h)\|_2 \le \frac{11}{4\sqrt{10}} \|\bm h\|_2,
\]
The last inequality uses \(\|\Rea(U_{\bm z_t}^*\bm h)\|_2=\|\bm h\|_2\) due to \eqref{eq:property1} and \eqref{eq:property2}.

For $II$,  by the definition of $S$ in \eqref{eq:setS}, one has
\[
 \|\bm b_{S^c}\|_2^2
 \le\frac1{4096M^2m}\sum_{\ell,j}|\bm f_j^*{\bm D}_\ell^*\bm h|^2
 \le\frac{\|\bm h\|_2^2}{4096}.
\]
By \eqref{eq:Rzbound}, one has $\norm{ \mathsf R_{\bm z_t}^{\dagger} } \le 1$. 
Therefore,  \eqref{eq:property2} yields 
\[
 \begin{aligned}
 \norm{II} \le \norm{ \mathsf R_{\bm z_t}^{\dagger}} \left\|
   \begin{bmatrix}\bm b_{S^c}\\
     -\frac12\mathcal J_{\bm x_\star,S}\bm h\end{bmatrix}\right\|_2\le\|\bm b_{S^c}\|_2
       +\frac12\|\mathcal J_{\bm x_\star,S}\bm h\|_2\le\left(\frac1{64}+\frac1{32}\right)\|\bm h\|_2.
 \end{aligned}
\]
The last inequality uses \eqref{eq:erasure}, which applies because $|S| \le \gamma m$ by \eqref{eq:numS}.
Combining the {bounds} for $I$ and $II$, we have
\begin{align*}
 \dist(\bm z_{t+1},\bm x_\star) \le\|\bm z_{t+1}-\bm x_\star\|_2  \le\left(\frac{11}{4\sqrt{10}}+\frac3{64}\right)\|\bm h\|_2
 \le\frac{15}{16}\dist(\bm z_t,\bm x_\star).
\end{align*}

It remains to prove the claim \eqref{cl:good-upper}.  Since \(\bm z_t=\bm x_\star+\bm h\), the definition
\eqref{eq:jac} gives \(\mathcal J_{\bm z_t,S^c}\bm v  =\mathcal J_{\bm x_\star,S^c}\bm v+\mathcal J_{\bm h,S^c}\bm v\). The reverse triangle inequality implies
\begin{equation} \label{eq:Jz0}
 \|\mathcal J_{\bm z_t,S^c}\bm v\|_2
 \ge\|\mathcal J_{\bm x_\star,S^c}\bm v\|_2 -\|\mathcal J_{\bm h,S^c}\bm v\|_2
 \end{equation}
 We first give an upper bound for $\|\mathcal J_{\bm h,S^c}\bm v\|_2$. Observe that
  \begin{align}
 \|\mathcal J_{\bm h,S^c}\bm v\|_2^2
 &=\frac4m\sum_{(\ell,j)\in S^c}
   \bigl[\Rea(\overline{\bm f_j^*{\bm D}_\ell^*\bm h}\,
                         \bm f_j^*{\bm D}_\ell^*\bm v)\bigr]^2 \notag\\
 &\le\frac4m\sum_{(\ell,j)\in S^c}
       |\bm f_j^*{\bm D}_\ell^*\bm h|^2|\bm f_j^*{\bm D}_\ell^*\bm v|^2 \notag\\
 &\le\frac{4}{(64M)^2m}\sum_{\ell,j}|\bm f_j^*{\bm D}_\ell^*\bm v|^2 
 \le\frac{\|\bm v\|_2^2}{32^2}, \label{eq:Jz1}
\end{align}
where the last inequality comes from  \eqref{eq:frame}.  For the  term $\|\mathcal J_{\bm x_\star,S^c}\bm v\|_2$, 
let  
\[
  \bm v_\perp=\bm v- \langle\bm v,\mathrm i\bm x_\star\rangle_{\mathbb R} \mathrm i\bm x_\star \in\mathcal H_{\bm x_\star}
 \]
   be the real orthogonal projection onto \(\mathcal H_{\bm x_\star}\). Since $\mathcal J_{\bm x_\star}(\mathrm i\bm x_\star) =\bm 0$, we have \(\mathcal J_{\bm x_\star}\bm v_\perp=\mathcal J_{\bm x_\star}\bm v\). Therefore,
 \[
  \|\mathcal J_{\bm x_\star,S^c}\bm v\|_2^2 =\|\mathcal J_{\bm x_\star}\bm v\|_2^2
      -\|\mathcal J_{\bm x_\star,S}\bm v\|_2^2=\|\mathcal J_{\bm x_\star}\bm v_\perp\|_2^2 -\|\mathcal J_{\bm x_\star,S}\bm v\|_2^2.
 \] 
  Let $\rho=\norm{\vh}$.  According to \eqref{eq:target-gram}, one has 
\begin{equation} \label{eq:Jxv}
\|\mathcal J_{\bm x_\star}\bm v\|_2^2= \|\mathcal J_{\bm x_\star}\bm v_\perp\|_2^2 \ge \frac32\|\bm v_\perp\|_2^2 \ge \frac32(1-\rho^2){\|\bm v\|_2^2},
\end{equation}
where the last inequality follows from 
\[
 \|\bm v_\perp\|_2^2
 =\|\bm v\|_2^2
   -\bigl|\langle\bm v,\mathrm i\bm x_\star\rangle_{\mathbb R}\bigr|^2
 \ge(1-\rho^2)\|\bm v\|_2^2
\]
since 
\[
\Abs{ \langle\bm v,\mathrm i\bm x_\star\rangle_{\mathbb R} }= \Abs{\langle\bm v,\mathrm i(\bm z_t-\vh)\rangle_{\mathbb R}} =\Abs{\langle\bm v,\mathrm i\bm h\rangle_{\mathbb R}} \le \|\bm v\|_2\|\bm h\|_2,
\]
where the second equality {uses} \(\bm v\in\mathcal H_{\bm z_t}\). On the other hand, \eqref{eq:erasure} together with \eqref{eq:numS} gives
\[
\|\mathcal J_{\bm x_\star,S}\bm v\|_2^2 \le\|\bm v\|_2^2/256.
\]
 Combining the two bounds gives
\begin{equation} \label{eq:Jz2}
 \|\mathcal J_{\bm x_\star,S^c}\bm v\|_2^2 \ge\left[\frac32(1-\rho^2)-\frac1{256}\right]\|\bm v\|_2^2.
\end{equation}
Substituting \eqref{eq:Jz1} and \eqref{eq:Jz2} into \eqref{eq:Jz0} and using \(\rho\le1/4\), we obtain
\begin{align*}
 \|\mathcal J_{\bm z_t,S^c}\bm v\|_2 \ge\left[\sqrt{\tfrac32(1-\rho^2)-\tfrac1{256}}
                          -\tfrac1{32}\right]\|\bm v\|_2  >\|\bm v\|_2.
 \label{eq:good-lower}
\end{align*}
{For the upper bound, the identity \(\mathcal J_{\bm x_\star}\bm v=\mathcal J_{\bm x_\star}\bm v_\perp\) and \(\|\bm v_\perp\|_2\le\|\bm v\|_2\) extend the upper bound in \eqref{eq:target-gram} to every \(\bm v\). Together with \eqref{eq:Jz1} and \(3+1/32<\sqrt{10}\), this gives}
\[
 \|\mathcal J_{\bm z_t,S^c}\bm v\|_2
 \le \|\mathcal J_{\bm x_\star,S^c}\bm v\|_2 +\|\mathcal J_{\bm h,S^c}\bm v\|_2 \le \|\mathcal J_{\bm x_\star}\bm v\|_2 +\|\mathcal J_{\bm h,S^c}\bm v\|_2 \le \sqrt{10}\|\bm v\|_2.
\]
{The claim also gives \(\sigma_{\min}(\mathsf R_{\bm z_t})\ge1\), hence \(\|\mathcal J_{\bm z_t}^{\dagger}\|_2=\|\mathsf R_{\bm z_t}^{\dagger}\|_2\le1\) by the real isometry \(U_{\bm z_t}\). Restoring the signal scale yields the stated bound \(1/\|\bm x_\star\|_2\).}
This completes the proof.
\end{proof}

{By Proposition~\ref{prop:linear}, it remains to establish \eqref{eq:target-gram} and \eqref{eq:erasure}, and to verify that every iterate \(\bm z_t\) remains in the basin \eqref{eq:prop-basin}.}  To this end, we need the following  uniform bound for arbitrary measurement subset.

\begin{lemma} \label{lem:subset}
Suppose that the entries of the masks $\dkh{\vd_l}_{l=1}^L$ are independent copies of a complex random variable $d$ satisfying \eqref{eq:mask}. For every fixed \(\beta>0\), with probability
at least \(1-\frac14(nL)^{-\beta}\), simultaneously for every set
\(S\subseteq\{1,\ldots,L\}\times\{0,\ldots,n-1\}\),
\begin{equation}
 \left\|\frac1m\sum_{(\ell,k)\in S}
              {\bm D}_\ell\bm f_k\bm f_k^*{\bm D}_\ell^*\right\|_2
 \le C_{M,\beta}\left(\frac1L+\frac{|S|}{m}\log^4(nL)\right).
 \label{eq:subset}
\end{equation}
Here, $\bm f_k$ and ${\bm D}_\ell$ are defined in \eqref{eq:fourier-vector0}, {\(m=nL\), and \(C_{M,\beta}\ge32M^2\) depends only on \(M\) and \(\beta\)}.
\end{lemma}
\begin{proof}
{See Appendix~\ref{sec:proLesubset}.}
\end{proof}

We are now ready to prove Theorem~\ref{thm:noisy}.

\subsection{Proof of Theorem~\ref{thm:noisy}} \label{sec:mainproof}

\begin{proof}
{Write \(m=nL\). We verify \eqref{eq:target-gram}, \eqref{eq:erasure}, and basin entry, then prove that the noisy updates preserve the basin.}

{The tangent estimate \eqref{rs:eq:tangent-injectivity} at \(\bm x_\star\), with \(\delta=1/4\), gives the following bound with probability at least \(1-m^{-\beta}/16\) when \(L\ge C_{M,\beta}\log m\):}
\begin{equation} \label{eq:cond1}
 \|\mathcal J_{\bm x_\star}\bm v\|_2^2  \ge\frac32{\|\bm x_\star\|_2^2}\|\bm v\|_2^2,\qquad
 \bm v\in\mathcal H_{\bm x_\star}.
\end{equation}
For the upper bound, Lemma~\ref{rs:prop:cdp-regularity} gives, with probability at least $1-m^{-\beta}/8$, 
\begin{equation} \label{eq:cond2}
 \|\mathcal J_{\bm x_\star}\bm v\|_2^2  \le \frac4m\sum_{(\ell,k)}|\bm f_k^*{\bm D}_\ell^*\bm x_\star|^2  |\bm f_k^*{\bm D}_\ell^*\bm v|^2 \le 4(2+1/4){\|\bm x_\star\|_2^2}\|\bm v\|_2^2=9{\|\bm x_\star\|_2^2}\|\bm v\|_2^2,
\end{equation}
provided \(L\ge C_{M,\beta}\log^2(nL)\) for some constant $C_{M,\beta}>0$. This gives  \eqref{eq:target-gram}.

For every fixed \(\beta>0\), apply \eqref{eq:scalar-mask-tail} at \(\bm x_\star\) and take a
{union bound over the \(m\) measurement pairs. It follows} with probability at least \(1-\frac1{16}(nL)^{-\beta}\) that 
\[
\max_{\ell,k}|\bm f_k^*{\bm D}_\ell^*\bm x_\star|^2
           \le C_{M,\beta}\log(nL)\|\bm x_\star\|_2^2.
 \]
Combining this estimate with Lemma~\ref{lem:subset}, we obtain, with probability at least \({1-\frac5{16}(nL)^{-\beta}}\), the following bound simultaneously for all \(S\) and \(\bm v\):
\begin{align*}
 \|\mathcal J_{\bm x_\star,S}\bm v\|_2^2
 &\le\frac4m\sum_{(\ell,k)\in S}|\bm f_k^*{\bm D}_\ell^*\bm x_\star|^2
                              |\bm f_k^*{\bm D}_\ell^*\bm v|^2\notag\\
 &\le C_*\left(\frac{\log(nL)}{L}
             +\frac{|S|}{m}\log^5(nL)\right)
                 \|\bm x_\star\|_2^2\|\bm v\|_2^2,
\end{align*}
where \(C_*=C_{M,\beta}\ge1\) is {a fixed constant}.  Let 
\[
 \gamma=\frac1{512C_*\log^5(nL)}.
\]
{Set \(r_b=\sqrt\gamma/(64M^2)\). Since \(m\ge2\) and \(M\ge1\), we have \(\gamma\le1\) and \(r_b\le1/64\). Fix \(C_*\) first and choose the theorem's noise constant so that}
{
\[
 c_{M,\beta}\le\frac1{4096M^2(1+M^2)\sqrt{512C_*}},\qquad
 \frac{\|\bm\xi\|_\infty}{\|\bm x_\star\|_2^2}\le\frac{r_b}{64(1+M^2)}.
\]
}
{Only the sampling constant is enlarged in the remaining estimates.}
Hence, if  \(L\ge512C_*\log(nL)\), then
\begin{equation} \label{eq:cond3}
 \|\mathcal J_{\bm x_\star,S}\|_2
 \le\frac{\|\bm x_\star\|_2}{{16}}
 \quad\text{for all } S\text{ with }|S|\le\gamma m.
\end{equation}
{Indeed, the preceding squared bound is at most \((1/512+1/512)\|\bm x_\star\|_2^2\|\bm v\|_2^2\); taking square roots gives \(1/16\). Thus \eqref{eq:erasure} holds.}

{For basin entry, we follow the initializer argument in Section~\ref{pr:spec}. Lemma~\ref{rs:prop:cdp-regularity} and Hoeffding's inequality \cite{hoeffding1963} give, with combined failure probability at most \(3m^{-\beta}/16\),}
\begin{align*}
 &\frac1{\|\bm x_\star\|_2^2}
 \left\|\frac1m\sum_{\ell,k}
  |\bm f_k^*{\bm D}_\ell^*\bm x_\star|^2
        {\bm D}_\ell\bm f_k\bm f_k^*{\bm D}_\ell^*
  -(\|\bm x_\star\|_2^2I+\bm x_\star\bm x_\star^*)\right\|_2
 \notag\\
 &\hspace{20mm}\le
 C_{M,\beta}\left(\sqrt{\frac{\log(nL)}L}
                          +\frac{\log^2(nL)}L\right)
 \le\frac{r_b}{16},
 \end{align*}
 and 
 \begin{align*}
 &\frac1{\|\bm x_\star\|_2^2}
 \left|\frac1m\sum_{\ell,k}|\bm f_k^*{\bm D}_\ell^*\bm x_\star|^2
                              -\|\bm x_\star\|_2^2\right|
 \le C_{M,\beta}\sqrt{\frac{\log(nL)}L}
 \le\frac{r_b}{8},
\end{align*}
provided $L\ge C'_{M,\beta} \log^6(nL)$ for a sufficiently large constant $C'_{M,\beta}>1$.
{These choices are possible because \(\sqrt{\log m/L}=O(\log^{-5/2}m)\) and \(\log^2m/L=O(\log^{-4}m)\). The chosen noise bound and \eqref{rs:eq:noisy-initializer-matrix-event} then give}
\[
\|Y-(\|\bm x_\star\|_2^2I+\bm x_\star\bm x_\star^*)\|_2 \le  \xkh{\frac{r_b}{16}+\frac{M^2r_b}{64(1+M^2)}} {\|\bm x_\star\|_2^2} \le \frac{5r_b}{64}{\|\bm x_\star\|_2^2}
\]
  and 
  \[
  |\lambda-\|\bm x_\star\|_2^2| \le  \xkh{\frac{r_b}{8}+\frac{r_b}{64(1+M^2)}} \|\bm x_\star\|_2^2 \le \frac{9r_b}{64}\|\bm x_\star\|_2^2.
  \]   
 {Since \(9r_b/64<1/2\), the energy estimate implies \(\lambda>0\). Weyl's inequality and the Davis--Kahan theorem \cite{davis1970rotation} give the following bound; its factor \(2\) is valid because the relative matrix perturbation is at most \(5r_b/64<1/4\):}
\[
 \dist\!\left(\bm v_0,
                  \frac{\bm x_\star}{\|\bm x_\star\|_2}\right)
 \le2\frac{\|Y-(\|\bm x_\star\|_2^2I+
                        \bm x_\star\bm x_\star^*)\|_2}
                {\|\bm x_\star\|_2^2}.
\]
Using \(\|\bm v_0\|_2=1\), the triangle inequality, and
\(\lambda>0\), we obtain
\begin{equation} \label{eq:cond4}
 \frac{\dist(\bm z_0,\bm x_\star)}{\|\bm x_\star\|_2} \le
 \dist\!\left(\bm v_0,
              \frac{\bm x_\star}{\|\bm x_\star\|_2}\right)
       +\frac{|\sqrt{\lambda}-\|\bm x_\star\|_2|}
                   {\|\bm x_\star\|_2}\le\frac{10r_b}{64}+\frac{9r_b}{64} <\frac{r_b}2 \le \frac14.
\end{equation}
This proves part~(1) of Theorem~\ref{thm:noisy}.

{The five underlying events \eqref{eq:cond1}--\eqref{eq:cond4} hold together with probability at least \(1-m^{-\beta}\), and the hypotheses of Proposition~\ref{prop:linear} hold at \(\bm z_0\).} Recall that the update rule is $\bm z_{t+1}=\bm z_t-\mathcal J_{\bm z_t}^{\dagger} \bm F(\bm z_t)$ and 
\begin{equation} \label{eq:noiseF}
\bm F({\bm z_t})= \bm F_{\mbox{clean}}(\bm z_t)- \frac{\bm\xi}{\sqrt m}.
\end{equation}
Hence, 
\begin{eqnarray}
 \dist(\bm z_{t+1},\bm x_\star) &\le&  \dist\!\left(\bm z_t-\mathcal J_{\bm z_t}^{\dagger}\bm F_{\mbox{clean}}(\bm z_t),
       \bm x_\star\right) +\frac{\|\mathcal J_{\bm z_t}^{\dagger}\|_2\|\bm\xi\|_2}
                   {\sqrt m} \notag \\
                   & \le &\frac{15}{16}\dist(\bm z_t,\bm x_\star)        +\frac{\|\bm\xi\|_2}{\sqrt m\,\|\bm x_\star\|_2} \label{eq:noisy-linear-step}.
\end{eqnarray}
Here, the last inequality follows from \eqref{eq:prop-contraction} and {the scaled pseudoinverse bound \(\|\mathcal J_{\bm z_t}^{\dagger}\|_2\le1/\|\bm x_\star\|_2\) in Proposition~\ref{prop:linear}}.
Furthermore, the noise condition \( \|\bm\xi\|_\infty \le c_{M,\beta} \log^{-5/2}(nL) \|\bm x_\star\|_2^2\) gives 
\begin{equation*} 
 \frac{\|\bm\xi\|_2}{\sqrt m\,\|\bm x_\star\|_2^2}
 \le\frac{\|\bm\xi\|_\infty}{\|\bm x_\star\|_2^2}
 \le\frac{r_b}{64(1+M^2)}\le\frac{r_b}{32}.
\end{equation*}
Therefore, if the current relative error is at most \(r_b/2\), then the above two inequalities give
\[
 \frac{\dist(\bm z_{t+1},\bm x_\star)}{\|\bm x_\star\|_2}
 \le\frac{15}{16}\frac{r_b}{2}+\frac{r_b}{32}
 =\frac{r_b}{2}.
\]
By induction, \eqref{eq:noisy-linear-step} holds for all ${t\ge0}$. This proves part~(2) of Theorem~\ref{thm:noisy}.

Finally, iterating \eqref{eq:noisy-linear-step} yields
\begin{align*}
 \dist(\bm z_t,\bm x_\star) \le
 \left(\frac{15}{16}\right)^t\dist(\bm z_0,\bm x_\star) +\frac{16\|\bm\xi\|_2}{\sqrt m\,\|\bm x_\star\|_2}.
  \end{align*}
  Take
  \begin{equation}
 T=\left\lceil
 \frac{{\max\{0,\log(16M^2\sqrt n\,r_b)\}}}{\log(16/15)}
 \right\rceil.
 \label{eq:noisy-entry-time}
\end{equation}
  If in addition
 \begin{equation*} 
  \frac{\|\bm\xi\|_2}{\sqrt m\,\|\bm x_\star\|_2^2}  \le\frac1{512M^2\sqrt n},
 \end{equation*}
then, for all $t\ge T$,
\begin{eqnarray*}
 {\frac{\dist(\bm z_t,\bm x_\star)}{\|\bm x_\star\|_2}} &\le &  \left(\frac{15}{16}\right)^T
       \frac{\dist(\bm z_0,\bm x_\star)}{\|\bm x_\star\|_2} + \frac{16\|\bm\xi\|_2}{\sqrt m\,\|\bm x_\star\|_2^2} \\
       &\le & \left(\frac{15}{16}\right)^T\frac{r_b}{2} + \frac1{32M^2\sqrt n} \\
 &\le & \frac1{16M^2\sqrt n}.
\end{eqnarray*}
For all $\vw, \vv \in \C^n$, we have
\begin{equation} \label{eq:Jhvlatter}
 \|\mathcal J_{\bm w}\bm v\|_2 \le 2\sqrt{\frac1m\sum_{(\ell,k)}|\bm f_k^*{\bm D}_\ell^* \bm w|^2  |\bm f_k^*{\bm D}_\ell^*\bm v|^2}
 \le2M^2\sqrt n\|\bm w\|_2\|\bm v\|_2.
\end{equation}
{As in Proposition~\ref{prop:linear}, phase-align \(\bm z_t\) and set \(\bm h=\bm z_t-\bm x_\star\). Then \(\|\bm h\|_2=\dist(\bm z_t,\bm x_\star)\), \(\bm h\in\mathcal H_{\bm z_t}\), and \(\|\mathcal J_{\bm z_t}^{\dagger}\|_2\le1/\|\bm x_\star\|_2\). For \(\bm b=(|\bm f_j^*{\bm D}_\ell^*\bm h|^2/\sqrt m)_{\ell,j}\), \eqref{eq:Jhvlatter} gives \(\|\bm b\|_2=\frac12\|\mathcal J_{\bm h}\bm h\|_2\le M^2\sqrt n\|\bm h\|_2^2\). Thus \eqref{eq:noiseF} and \eqref{eq:upzt} imply}
{
\begin{align*}
 \dist(\bm z_{t+1},\bm x_\star)
 &\le\left\|\mathcal J_{\bm z_t}^{\dagger}\left(\bm b+\frac{\bm\xi}{\sqrt m}\right)\right\|_2\\
 &\le\frac{2M^2\sqrt n\,\dist^2(\bm z_t,\bm x_\star)}{\|\bm x_\star\|_2}+\frac{\|\bm\xi\|_2}{\sqrt m\,\|\bm x_\star\|_2}.
\end{align*}
}
Dividing by \(\|\bm x_\star\|_2\) proves (3); the definition of \(T\) and \(r_b\le1/64\) give \(T\le C_{M,\beta}\log(en)\).

\end{proof}

\section{Conclusions and discussion}
\label{sec:discussion}
We have established convergence guarantees for the minimum-norm Gauss--Newton method for coded diffraction phase retrieval under two data-use protocols, which require different probabilistic arguments.
With resampling, the current iterate is independent of the next mask block, so a Jacobian lower bound and a fourth-moment estimate in the aligned-error direction suffice.  
With fixed data, the current iterate depends on the masks. Our uniform subset estimate controls the rows selected by the current error without requiring independence between the iterate and the data.

Several questions remain open. First, our guarantees hold for each fixed signal $\vx_\star$, rather than simultaneously for all $\vx_\star$. Extending the analysis to a uniform recovery guarantee would be of interest. Second, the mask complexities are $O(\log^2 n)$ per independent block with resampling and $O(\log^6 n)$ in total with fixed data. Reducing these logarithmic exponents is another important direction. Third, our convergence guarantees rely on spectral initialization; establishing convergence from a random initialization remains an interesting challenge.

\appendix
\section{Proof of Lemma \ref{lem:subset}} \label{sec:proLesubset}
\begin{proof}
{First consider \(n<32\log^4(nL)\). In this case,} 
\begin{align*}
 \left\|\frac1m\sum_{(\ell,k)\in S}
              {\bm D}_\ell\bm f_k\bm f_k^*{\bm D}_\ell^*\right\|_2
 &\le\frac1m\sum_{(\ell,k)\in S}
              \|{\bm D}_\ell\bm f_k\bm f_k^*{\bm D}_\ell^*\|_2\\
 &\le\frac{nM^2|S|}{m}
 \le32M^2\frac{|S|}{m}\log^4(nL).
\end{align*}
Here, we use the fact that \(\norm{{\bm D}_{\ell}}\le M\) and $\norm{\vf_k}=\sqrt n$.  {Thus \eqref{eq:subset} holds deterministically in this case. Henceforth assume \(n\ge32\log^4(nL)\).}

Let \(U=n^{-1/2}[\bm f_0,\ldots,\bm f_{n-1}]\) and
\(W=[{\bm D}_1U,\ldots,{\bm D}_LU] \in \C^{n\times m}\).  For any \(S\subseteq\{1,\ldots,L\}\times\{0,\ldots,n-1\}\), let \(W_S\) be the submatrix of \(W\) consisting of its columns in
\(S\).  {Since the column of \(W\) indexed by \((\ell,k)\) is \(n^{-1/2}{\bm D}_\ell\bm f_k\),} 
\[
 \frac1m\sum_{(\ell,k)\in S}{\bm D}_\ell\bm f_k\bm f_k^*{\bm D}_\ell^*
 =\frac n m W_SW_S^*=\frac1L W_SW_S^*.
\]
To prove the result, it suffices to show that there exists a constant $C_{M,\beta}$ such that 
\begin{equation} \label{eq:supWS}
 {\|W_S\|_2^2} \le C_{M,\beta}\left(1+\frac{|S|}{n}\log^4(nL)\right){\quad\text{for every }S}.
\end{equation}
To this end,  let \(T_r=\{\bm w{\in\mathbb C^m}:\|\bm w\|_2\le1, |\supp\bm w|\le r\}\), where \(1\le r\le m\). We claim that for each fixed integer \(1\le r\le m\) and each fixed \(p\ge1\), with probability at least \(1-2e^{-p}\),
\begin{equation}  \label{eq:restricted-upper}
 \sup_{\bm w\in T_r}\|W\bm w\|_2
 \le C_M\left[1+\sqrt{r/n}\log^2(2m)+\sqrt{pr/n}\right].
\end{equation}
Here \(C_M\) depends only on the mask bound \(M\).
{Choose \(r_0=\lfloor n/(16\log^4(nL))\rfloor\). The assumption on \(n\) gives \(r_0\ge n/(32\log^4(nL))\ge1\). Applying \eqref{eq:restricted-upper} with \(r=r_0\) and \(p=(\beta+4)\log(nL)\), and using \(2m^{-\beta-4}\le m^{-\beta}/4\), gives, with probability at least \(1-m^{-\beta}/4\),}
\begin{equation}  \label{eq:restricted-upper1}
 \sup_{\bm w\in T_{r_0}}\|W\bm w\|_2  \le  C_M\xkh{2+\sqrt{\beta+4}}:=K_{M,\beta}.
\end{equation}
On this event, partition any fixed nonempty set \(S\) into
\(b=\lceil|S|/r_0\rceil\) disjoint sets
\(S_1,\ldots,S_b\), each of size at most \(r_0\).
For any \(\bm w\in\mathbb C^m\) with  \(\mbox{supp}(\bm w) \subset S\), let
\(\bm w_{S_a}\) agree with \(\bm w\) on \(S_a\) and vanish outside \(S_a\).  
Whenever \(\bm w_{S_a}\ne\bm0\), the normalized vector
\(\bm w_{S_a}/\|\bm w_{S_a}\|_2\) belongs to \(T_{r_0}\).
Applying \eqref{eq:restricted-upper1} to this vector yields
\[
 \|W\bm w_{S_a}\|_2  \le K_{M,\beta}\|\bm w_{S_a}\|_2.
\]
Note that 
\[
 \bm w=\sum_{a=1}^b\bm w_{S_a},
 \qquad
 \sum_{a=1}^b\|\bm w_{S_a}\|_2^2=\|\bm w\|_2^2.
\]
Hence  the triangle inequality and Cauchy--Schwarz give
\begin{align*}
 \|W\bm w\|_2
 &=\left\|\sum_{a=1}^bW\bm w_{S_a}\right\|_2
 \le\sum_{a=1}^b\|W\bm w_{S_a}\|_2
 \le K_{M,\beta}\sum_{a=1}^b\|\bm w_{S_a}\|_2\\
 &\le K_{M,\beta}\sqrt b
       \left(\sum_{a=1}^b\|\bm w_{S_a}\|_2^2\right)^{1/2}
 =K_{M,\beta}\sqrt b\|\bm w\|_2,
\end{align*}
which means 
 \[
 \|W_S\|_2\le K_{M,\beta}\sqrt b  \le K_{M,\beta}\sqrt{1+|S|/r_0} \le K_{M,\beta} \sqrt{1+32|S|\log^4(nL)/n}.
 \]
 On the event in \eqref{eq:restricted-upper1}, this inequality holds simultaneously for all $S$.
This immediately gives
\[
 {\|W_S\|_2^2} \le K_{M,\beta}^2\left(1+{\frac{32|S|\log^4(nL)}{n}}\right){\quad\text{for every }S},
\]
which proves \eqref{eq:supWS} and hence \eqref{eq:subset}.

It remains to prove \eqref{eq:restricted-upper}. We use the norm-process bound in Lemma~\ref{lem:kmr}. 
For a compact matrix class \(\mathcal V\),
write \(\mathcal N(\mathcal V,u)\) for its operator-norm covering
number, and put
\[
 d_F=\sup_{V\in\mathcal V}\|V\|_F,\qquad
 d_{\rm op}=\sup_{V\in\mathcal V}\|V\|_2.
\]
For \(\bm w=[\bm w_1^\top,\ldots,\bm w_L^\top]^\top
\in\mathbb C^{nL}\) with \(\bm w_\ell=(w_{\ell,0},\ldots,w_{\ell,n-1})^\top \in \C^n\), define
\begin{equation}
 V_{\bm w}=[\diag(U\bm w_1),\ldots,\diag(U\bm w_L)]
       \in\mathbb C^{n\times m},
 \label{eq:V}
\end{equation}
where  \(U=n^{-1/2}[\bm f_0,\ldots,\bm f_{n-1}]\).
Let \(\bm d=\zkh{\vd_1^\top, \ldots, \vd_L^\top}^\top \in \C^m\).  Recall that \(W=[{\bm D}_1U,\ldots,{\bm D}_LU]\).  It is {straightforward to verify} that  \(W\bm w=V_{\bm w}\bm d\). Thus, to prove the claim \eqref{eq:restricted-upper}, it suffices to show, with probability at least \(1-2e^{-p}\),
\begin{equation}  \label{cla:Vw}
 \sup_{\bm w\in T_r}\|V_{\bm w}\bm d\|_2 \le C_M\left[1+\sqrt{r/n}\log^2(2m)+\sqrt{pr/n}\right],
\end{equation}
{We apply Lemma~\ref{lem:kmr}, which is stated for real matrices. To this end, set} \(\bm d_{\mathrm R}=\sqrt2\Rea\bm d\), \(\bm d_{\mathrm I}=\sqrt2\Ima\bm d\), and let
\[
 \widetilde V_{\bm w}
 =\begin{bmatrix}\Rea V_{\bm w}\\ \Ima V_{\bm w}\end{bmatrix}
 \in\mathbb R^{2n\times m}.
\]
{We have \(\bm d=(\bm d_{\mathrm R}+\mathrm i\bm d_{\mathrm I})/\sqrt2\), and for every \(\bm u\in\mathbb R^m\),} 
\begin{equation} \label{eq:VtildeV}
 \|\widetilde V_{\bm w}\bm u\|_2^2
 =\|(\Rea V_{\bm w})\bm u\|_2^2
  +\|(\Ima V_{\bm w})\bm u\|_2^2
 =\|V_{\bm w}\bm u\|_2^2.
\end{equation}
Therefore, 
\begin{align}
 \sup_{\bm w\in T_r}\|V_{\bm w}\bm d\|_2  & \le\frac1{\sqrt2}\left(
    \sup_{\bm w\in T_r}\|V_{\bm w}\bm d_{\mathrm R}\|_2
   +\sup_{\bm w\in T_r}\|V_{\bm w}\bm d_{\mathrm I}\|_2\right) \notag \\
   &= \frac1{\sqrt2}\left(
    \sup_{\bm w\in T_r}\| \widetilde V_{\bm w}\bm d_{\mathrm R}\|_2
   +\sup_{\bm w\in T_r}\| \widetilde V_{\bm w}\bm d_{\mathrm I}\|_2\right). \label{eq:supvwd}
\end{align}
We apply Lemma~\ref{lem:kmr} separately to $ \sup_{\bm w\in T_r}\|V_{\bm w}\bm d_{\mathrm R}\|_2$ and $\sup_{\bm w\in T_r}\|V_{\bm w}\bm d_{\mathrm I}\|_2$. We first verify the hypotheses of \eqref{eq:kmr}. 
Under the assumption \eqref{eq:mask}, one has 
\[
 \E[(\Rea d)^2]+\E[(\Ima d)^2]=\E|d|^2=1,
 \qquad
 \E[(\Rea d)^2]-\E[(\Ima d)^2]=\Rea(\E d^2)=0.
\]
Thus each of \(\Rea d\) and \(\Ima d\) has second moment
\(1/2\). Since \(\E d=0\), the coordinates of \(\bm d_{\mathrm R}\)  are independent,  centered, variance-one real
random variables bounded by \(\sqrt2M\). These variables are subgaussian by Hoeffding's inequality \cite{hoeffding1963}.  {For \(\widetilde{\mathcal V}=\{\widetilde V_{\bm w}:\bm w\in T_r\}\), Lemma~\ref{lem:kmr} gives}
\begin{equation} \label{eq:Esup0}
 \left(\E\sup_{\bm w\in T_r }\|\widetilde V_{\bm w}\bm d_{\mathrm R} \|_2^p\right)^{1/p}
 \le C_1 \left[
  \int_0^{d_{\rm op}(\widetilde{\mathcal V})}\sqrt{\log\mathcal N(\widetilde{\mathcal V},u)}\,du
  +d_F(\widetilde{\mathcal V})+\sqrt p\,d_{\rm op}(\widetilde{\mathcal V})\right],
\end{equation}
where $C_1>0$ is a  constant depending only on $M$,
\begin{equation} \label{eq:relanorm1}
d_F(\widetilde{\mathcal V}):= \sup_{\bm w\in T_r } \|\widetilde V_{\bm w}\|_F=\sup_{\bm w\in T_r }  \|V_{\bm w}\|_F= \sup_{\bm w\in T_r } \|\bm w\|_2\le1,
\end{equation}
and 
\[
d_{\rm op}(\widetilde{\mathcal V}):= \sup_{\bm w\in T_r } \|\widetilde V_{\bm w}\|_2.
\]
Observe that by \eqref{eq:VtildeV}, we have 
\begin{equation} \label{eq:relanorm2}
\|\widetilde V_{\bm w} \|_2 =\sup_{\bm u \in\mathbb R^m,\ \|\bm u\|_2=1}
       \|\widetilde V_{\bm w}  \bm u \|_2 =\sup_{\bm u \in\mathbb R^m,\ \|\bm u\|_2=1}
       \|V_{\bm w} \bm u \|_2{\le} \|V_{\bm w}\|_2
\end{equation}
and from the definition of $V_{\bm w}$ in \eqref{eq:V}, we have
\begin{equation}
 \|V_{\bm w}\|_2
 =\max_j\left(\sum_\ell|(U\bm w_\ell)_j|^2\right)^{1/2}
 =:\|\bm w\|_X.
 \label{eq:Vnorm}
\end{equation}
 {Since every entry of \(U\) has modulus \(n^{-1/2}\),} for each \(\bm w\in T_r\), we have
\begin{equation}
 \|\bm w\|_X\le\|\bm w\|_1/\sqrt n
 \le\sqrt{r/n}\|\bm w\|_2
 \label{eq:sparse-radius}
\end{equation}
Consequently, $d_{\rm op}(\widetilde{\mathcal V}) \le \sqrt{r/n}$.  Next, we claim that there is an absolute constant \(C>0\) such that
\begin{equation} \label{eq:intbound}
  \int_0^{d_{\rm op}(\widetilde{\mathcal V})}\sqrt{\log\mathcal N(\widetilde{\mathcal V},u)}\,du \le  C\sqrt{r/n}\log^2(2m).
\end{equation}
Combining these estimates with \eqref{eq:Esup0} gives
\[
 \left(\E\sup_{\bm w\in T_r }\|\widetilde V_{\bm w}\bm d_{\mathrm R} \|_2^p\right)^{1/p}
 \le C_M\left[1+\sqrt{r/n}\log^2(2m)   +\sqrt{pr/n}\right].
\]
For a nonnegative random variable \(Z\), Markov's
inequality applied to \(Z^p\) gives \(\Prob\{Z>e(\E Z^p)^{1/p}\}\le e^{-p}\). Therefore, with probability at least $1-e^{-p}$,
\[
  \sup_{\bm w\in T_r}\| \widetilde V_{\bm w}\bm d_{\mathrm R}\|_2 \le eC_M\left[1+\sqrt{r/n}\log^2(2m)   +\sqrt{pr/n}\right].
\]
The same argument gives an identical bound for \(\sup_{\bm w\in T_r}\| \widetilde V_{\bm w}\bm d_{\mathrm I}\|_2\). Substituting these bounds into \eqref{eq:supvwd} and taking a union bound, we obtain, with probability at least \(1-2e^{-p}\),
\[
 \sup_{\bm w\in T_r}\|W \bm w \|_2= \sup_{\bm w\in T_r}\|V_{\bm w}\bm d\|_2 \le \sqrt2eC_M\left[1+\sqrt{r/n}\log^2(2m)+\sqrt{pr/n}\right],
\]
which proves \eqref{eq:restricted-upper}.

Finally, we prove \eqref{eq:intbound}.  Let \(\mathcal V=\{V_{\bm w}:\bm w\in T_r\}\). Recall from \eqref{eq:relanorm1} and \eqref{eq:relanorm2} that 
\begin{align*}
 \|\widetilde V_{\bm w}\|_F=\|V_{\bm w}\|_F, \qquad  \|\widetilde V_{\bm w} \|_2{\le}\| V_{\bm w} \|_2.
\end{align*}
{The same argument applies to differences: \(\|\widetilde V_{\bm w}-\widetilde V_{\bm v}\|_2\le\|V_{\bm w}-V_{\bm v}\|_2\), since the supremum over real inputs is no larger than that over complex inputs.} Therefore, any operator-norm \(u\)-net of \(\mathcal V=\{V_{\bm w}:\bm w\in T_r\}\) corresponds to  an operator-norm \(u\)-net of $\widetilde{\mathcal V}$. Consequently, the covering numbers satisfy \( \mathcal N(\widetilde{\mathcal V},u) \le \mathcal N(\mathcal V,u)\).   Moreover, the map
\(\bm w\mapsto V_{\bm w}\) is linear, so \eqref{eq:Vnorm} and
\eqref{eq:sparse-radius} imply
\[
 \|V_{\bm w}-V_{\bm v}\|_2=\|\bm w-\bm v\|_X,
 \qquad d_F=1,\qquad d_{\rm op}\le d_r:=\sqrt{r/n}.
\]
In particular, for every \(u>0\), covering \(\mathcal V\) in
operator norm is exactly equivalent to covering \(T_r\) in
\(\|\cdot\|_X\), with centers in the respective sets:
\[
 \mathcal N(\mathcal V,u)
 =\mathcal N(T_r,\|\cdot\|_X,u).
\]
Here, $\|\cdot\|_X$  is defined in \eqref{eq:Vnorm}. 
From Lemma \ref{le:covernum}, {for every \(u>0\)}, 
\begin{equation}  \label{eq:coarse-cover11}
 \log\mathcal N(T_r,\|\cdot\|_X,u)  \le r\log(em/r)+2r\log(1+2\sqrt{r/n}/u).
\end{equation}
and for $0<u\le\sqrt{r/n}$,
\begin{equation}  \label{eq:coarse-cover21}
 \log\mathcal N(T_r,\|\cdot\|_X,u)
 \le\frac{Cr\log(2n)\log(2m)}{nu^2}.
\end{equation}
We now integrate the two covering estimates for the matrix class
\(\mathcal V\) to prove \eqref{eq:intbound}.  Specifically, 
\begin{align*}
&  \int_0^{d_{\rm op}(\widetilde{\mathcal V})}\sqrt{\log\mathcal N(\widetilde{\mathcal V},u)}\,du \\
 &\quad\overset{\mathrm{(i)}}{\le}
       \int_0^{ \sqrt{r/n}}\sqrt{\log\mathcal N(T_r,\|\cdot\|_X,u)}\,du\\
 &\quad=
       \int_0^{1/\sqrt n}\sqrt{\log\mathcal N(T_r,\|\cdot\|_X,u)}\,du
       +\int_{1/\sqrt n}^{\sqrt{r/n}}
          \sqrt{\log\mathcal N(T_r,\|\cdot\|_X,u)}\,du\\
 &\quad\overset{\mathrm{(ii)}}{\le}
       \sqrt{\frac rn}\int_0^1
          \sqrt{\log(em/r)+2\log(1+2\sqrt r/v)}\,dv +C\sqrt{\frac rn}\sqrt{\log(2n)\log(2m)}
           \log\sqrt r\\
 &\quad\overset{\mathrm{(iii)}}{\le}
       C\sqrt{\frac rn}\left[
          \sqrt{\log(2m)}
          +\int_0^1\sqrt{\log(1/v)}\,dv
          +\sqrt{\log(2n)\log(2m)}\log \sqrt r
       \right]\\
 &\quad\overset{\mathrm{(iv)}}{\le}
       C\sqrt{\frac rn}\log^2(2m).
\end{align*}
Here, (i) uses $\mathcal N(\widetilde{\mathcal V},u) \le \mathcal N(T_r,\|\cdot\|_X,u)$ and $d_{\rm op}(\widetilde{\mathcal V}) \le \sqrt{r/n}$, as established in \eqref{eq:sparse-radius}. For (ii), we apply \eqref{eq:coarse-cover11} and make the substitution \(u=v/\sqrt n\) in the first term, and  use \eqref{eq:coarse-cover21} to obtain the second term.  For (iii), we use \(\sqrt{a+b}\le\sqrt a+\sqrt b\) for \(a,b\ge0\), together with \(\log(1+2\sqrt r/v)\le\log(1+2\sqrt r)+\log(1/v)\) and \(\log(em/r)+2\log(1+2\sqrt r)\le C\log(2m)\) for \(0<v\le1\).
Finally, (iv) uses \(\int_0^1\sqrt{\log(1/v)}\,dv=\sqrt\pi/2\) {and \(n,r\le m\), so \(\log\sqrt r\le\frac12\log m\)}. 
This completes the proof.
\end{proof}

\begin{lemma} \label{le:covernum}
For any fixed \(1\le r\le m\), let \(T_r=\{\bm w{\in\mathbb C^m}:\|\bm w\|_2\le1, |\supp\bm w|\le r\}\) and $\|\cdot\|_X$ be the norm defined in \eqref{eq:Vnorm}. {Let \(\mathcal N(T_r,\|\cdot\|_X,u)\) be the covering number with centers in \(T_r\). Then, for every \(u>0\),} 
\begin{equation}  \label{eq:coarse-cover1}
 \log\mathcal N(T_r,\|\cdot\|_X,u)  \le r\log(em/r)+2r\log(1+2\sqrt{r/n}/u).
\end{equation}
and for $0<u\le\sqrt{r/n}$,
\begin{equation}  \label{eq:coarse-cover2}
 \log\mathcal N(T_r,\|\cdot\|_X,u)
 \le\frac{Cr\log(2n)\log(2m)}{nu^2}.
\end{equation}
\end{lemma}

\begin{proof}
{We first prove \eqref{eq:coarse-cover1}. Fix a support of size \(r\). For two vectors on this same support, \eqref{eq:sparse-radius} gives \(\|\bm w-\bm v\|_X\le\sqrt{r/n}\|\bm w-\bm v\|_2\). Thus a Euclidean net of radius \(u/\sqrt{r/n}\) for its unit ball is a radius-\(u\) net in \(\|\cdot\|_X\). A volume comparison in real dimension \(2r\) bounds the size of a radius-\(v\) Euclidean net by \((1+2/v)^{2r}\). Taking the union over at most \(\binom mr\le(em/r)^r\) supports covers \(T_r\)}; choosing \(v=u/\sqrt{r/n}\) then gives
\begin{equation*}
 \log\mathcal N(T_r,\|\cdot\|_X,u)
 \le r\log(em/r)+2r\log(1+2\sqrt{r/n}/u).
\end{equation*}

We next {prove} \eqref{eq:coarse-cover2}.  To establish this bound, for any $0<u\le\sqrt{r/n}$, we construct an explicit $u$-net of $T_r$ under the norm $\|\cdot\|_X$.  Specifically,  let \(\bm e_k\) denote the coordinate unit vectors of \(\mathbb C^m\)
and define the finite set of \(4m+1\) vectors
\[
 \mathcal D=\{\bm0\}\cup
       \{\pm\bm e_k,\pm\mathrm i\bm e_k:1\le k\le m\}.
\]
For any fixed \(\bm w\in T_r\), the real and imaginary parts of \(\bm w\)
have at most \(2r\) nonzero coordinates.  Cauchy--Schwarz
and \(\|\bm w\|_2\le1\) therefore give
\[
 \|\Rea\bm w\|_1+\|\Ima\bm w\|_1
 \le\sqrt{2r}\,
   \bigl(\|\Rea\bm w\|_2^2+\|\Ima\bm w\|_2^2\bigr)^{1/2}
 \le\sqrt{2r}.
\]
For a real number \(a\), write
\(a_+=\max\{a,0\}\). Define a \(\mathcal D\)-valued random vector
\(\bm Y\) by
\begin{align}
 \Prob\{\bm Y=\bm e_k\}
   &=\frac{(\Rea w_k)_+}{\sqrt{2r}},
 &\Prob\{\bm Y=-\bm e_k\}
   &=\frac{(-\Rea w_k)_+}{\sqrt{2r}}, \label{ran:Y1}\\
 \Prob\{\bm Y=\mathrm i\bm e_k\}
   &=\frac{(\Ima w_k)_+}{\sqrt{2r}},
 &\Prob\{\bm Y=-\mathrm i\bm e_k\}
   &=\frac{(-\Ima w_k)_+}{\sqrt{2r}}, \label{ran:Y2}
 \qquad 1\le k\le m,
\end{align}
and 
\begin{equation} \label{ran:Y0}
\Prob\{\bm Y=\bm 0\}= 1-(\|\Rea\bm w\|_1+\|\Ima\bm w\|_1)/\sqrt{2r}.
\end{equation}
Direct calculation gives \(\sqrt{2r}\E\bm Y=\bm w\).  For any fixed integer $h \ge 1$, take \(h\) independent copies
\(\bm Y_1,\ldots,\bm Y_h\) of $\bm Y$.  According to Lemma \ref{le:Yh}, we have
\[
 \E\left\|\frac{\sqrt{2r}}h
          \sum_{a=1}^h\bm Y_a-\bm w\right\|_X
 \le C\sqrt{\frac{r\log(2n)}{nh}}.
\]
Now fix \(0<u\le\sqrt{r/n}\), and choose \(h=\lceil C_1r\log(2n)/(nu^2)\rceil\), where \(C_1 \ge 16C^2\). Thus, 
\begin{equation} \label{eq:ExpYa}
 \E\left\|\frac{\sqrt{2r}}h  \sum_{a=1}^h\bm Y_a-\bm w\right\|_X  \le \frac{u}4.
\end{equation}
Define the set
\[
\mathcal C:= \left\{\frac{\sqrt{2r}}h\sum_{a=1}^h\bm y_a:  \bm y_a\in\mathcal D\right\},
\]
which  has at most \((4m+1)^h\) elements and does not depend on
\(\bm w\).  According to \eqref{eq:ExpYa},  for each \(\bm w\in T_r\), there exists $\bm x \in \mathcal C$ such that \(\bm w \in B(\bm x, \|\cdot\|_X, u/4) \). {Here, \(B(\bm x,\|\cdot\|_X,u/4)\) is the closed ball of radius \(u/4\) centered at \(\bm x\) in the norm \(\|\cdot\|_X\).}  Denote
\[
{\mathcal C}_1:=\dkh{\vx \in \mathcal C: B(\bm x, \|\cdot\|_X, u/4) \cap T_r \neq \emptyset}.
\]
It is easy to see that 
\[
T_r  \subset \cup_{\bm x \in \mathcal C_1} ~ B(\bm x, \|\cdot\|_X, u/4).
\]
Note that the centers $\bm x \in \mathcal C_1$ may lie outside \(T_r\).  For those {points} $\bm x \in \mathcal C_1$ but $\vx \notin T_r$, choose any point $\widetilde\vx \in B(\bm x, \|\cdot\|_X, u/4) \cap T_r$.  Since $B(\bm x, \|\cdot\|_X, u/4) \subset B(\widetilde {\bm x}, \|\cdot\|_X, u/2)$, replacing these centers and retaining those already in \(T_r\) yields a set $\mathcal N \subset T_r$ that is a radius-\(u\) net of \(T_r\) with at most \((4m+1)^h\) points. Thus, for any $0<u\le\sqrt{r/n}$,
\begin{equation*}
 \log\mathcal N(T_r,\|\cdot\|_X,u)  \le h \log(4m+1) \le  \frac{Cr\log(2n)\log(2m)}{nu^2}
\end{equation*}
by the choice of  \(h=\lceil C_1r\log(2n)/(nu^2)\rceil\).
This completes the proof.
\end{proof}

\begin{lemma} \label{le:Yh}
For any fixed \(\bm w\in T_r\), suppose that  \(\bm Y_1,\ldots,\bm Y_h\) are independent copies of $\bm Y$ defined in \eqref{ran:Y1}--\eqref{ran:Y0}. Then 
\[
 \E\left\|\frac{\sqrt{2r}}h
          \sum_{a=1}^h\bm Y_a-\bm w\right\|_X
 \le C\sqrt{\frac{r\log(2n)}{nh}}.
\]
\end{lemma}

\begin{proof}
 Introduce independent copies \(\bm Y'_1,\ldots,\bm Y'_h\) and independent
Rademacher random variables \(\epsilon_1,\ldots,\epsilon_h\).  Jensen's inequality, followed by the symmetry of \(\bm Y_a-\bm Y'_a\) and the triangle inequality, gives
\begin{align}
 \E\left\|\frac{\sqrt{2r}}h\sum_{a=1}^h\bm Y_a-\bm w\right\|_X
 &\le\frac{\sqrt{2r}}h
       \E\left\|\sum_{a=1}^h(\bm Y_a-\bm Y'_a)\right\|_X
       \notag\\
 &\le\frac{2\sqrt{2r}}h
       \E\left\|\sum_{a=1}^h\epsilon_a\bm Y_a\right\|_X.
 \label{eq:maurey}
\end{align}
We next bound the right-hand side of \eqref{eq:maurey} using the definition of \(\|\cdot\|_X\). Write each $\bm Y_a \in \C^{m}$ in block form as $\bm Y_a=\xkh{\bm Y_{a,1}^\top,\ldots, \bm Y_{a,L}^\top}{^\top}$, with $\bm Y_{a,1}, \ldots, \bm Y_{a,L} \in \C^n$. For \(0\le j<n\), let
\[
 \bm g_{a,j}
 =\xkh{(U\bm Y_{a,1})_j, \ldots, (U\bm Y_{a,L})_j}^\top  \in\mathbb C^L.
\]
With this notation,  the definition of \(\|\cdot\|_X\) gives
\[
 \left\|\sum_{a=1}^h\epsilon_a\bm Y_a\right\|_X
 =\max_{0\le j<n}   \left\|\sum_{a=1}^h\epsilon_a\bm g_{a,j}\right\|_2
\]
For any fixed row \(j\), we apply McDiarmid's bounded differences
inequality, stated as Lemma~\ref{lem:mcdiarmid} in
Appendix~\ref{app:auxiliary}, to establish a tail bound for $\left\|\sum_{a=1}^h\epsilon_a\bm g_{a,j}\right\|_2$.
{Each atom \(\bm Y_a\) has at most one nonzero component. Since the entries of \(U\) have modulus \(n^{-1/2}\), this implies \(\|\bm g_{a,j}\|_2\le n^{-1/2}\). Independence and centering of the Rademacher signs give}
\[
 \E_\epsilon\left\|\sum_{a=1}^h
             \epsilon_a\bm g_{a,j}\right\|_2^2
 =\sum_{a=1}^h\|\bm g_{a,j}\|_2^2\le\frac hn.
\]
Here and below, \(\E_\epsilon\)  denotes {expectation with respect to the Rademacher random variables} \(\epsilon_1,\ldots,\epsilon_h\),  with \(\bm Y_1,\ldots,\bm Y_h\) held fixed. 
Cauchy--Schwarz therefore gives
\[
 \E_\epsilon\left\|\sum_{a=1}^h\epsilon_a\bm g_{a,j}\right\|_2
 \le\left(\E_\epsilon
       \left\|\sum_{a=1}^h\epsilon_a\bm g_{a,j}\right\|_2^2
       \right)^{1/2}
 \le\sqrt{h/n}.
\]
Next, if two sign sequences  differ only at coordinate \(a\), the  triangle inequality gives
\[
 \left|
  \left\|\sum_{b=1}^h\epsilon_b\bm g_{b,j}\right\|_2
  -\left\|\sum_{b=1}^h\epsilon'_b\bm g_{b,j}\right\|_2
 \right|
 \le|\epsilon_a-\epsilon'_a|\,\|\bm g_{a,j}\|_2
 \le\frac2{\sqrt n}.
\]
Thus, Lemma~\ref{lem:mcdiarmid} gives, for every \(t>0\),
\begin{align*}
 \Prob_\epsilon\left\{
  \left\|\sum_{a=1}^h\epsilon_a\bm g_{a,j}\right\|_2
     >\sqrt{h/n}+t
  \,\middle|\,\bm Y_1,\ldots,\bm Y_h\right\}
 \le\exp\left(-\frac{nt^2}{2h}\right).
\end{align*}
Taking a union bound over all $0\le j\le n-1$ gives
\[
 \Prob_\epsilon\left\{
  \left\|\sum_{a=1}^h\epsilon_a\bm Y_a\right\|_X
     >\sqrt{h/n}+t
  \,\middle|\,\bm Y_1,\ldots,\bm Y_h\right\}
 \le\min\left\{1,n\exp\left(-\frac{nt^2}{2h}\right)\right\}.
\]
Therefore, 
\begin{align*}
 \E_\epsilon\left\|\sum_{a=1}^h\epsilon_a\bm Y_a\right\|_X & =\int_0^{\infty}  \Prob_\epsilon\left\{
  \left\|\sum_{a=1}^h\epsilon_a\bm Y_a\right\|_X
     >t  \,\middle|\,\bm Y_1,\ldots,\bm Y_h\right\}  dt \\
 &\le\sqrt{\frac hn}
  +\int_0^\infty
       \min\left\{1,n\exp\left(-\frac{nt^2}{2h}\right)\right\}dt\\
 &\le\sqrt{\frac hn}+t_0
       +n\int_{t_0}^\infty\exp\left(-\frac{nt^2}{2h}\right)dt\\
 &\le\sqrt{\frac hn}+t_0
       +\frac{h}{t_0}\exp\left(-\frac{nt_0^2}{2h}\right)\\
 &=\sqrt{\frac hn}+t_0+\frac{h}{2nt_0}
 \le C\sqrt{\frac{h\log(2n)}n},
\end{align*}
where \(t_0=\sqrt{2h\log(2n)/n}\), and the {fourth} line comes from 
\(\int_{t_0}^\infty e^{-at^2}\,dt \le e^{-at_0^2}/(2at_0)\) for \(a,t_0>0\).
The bound holds uniformly over all realizations of \(\bm Y_1,\ldots,\bm Y_h\). Taking expectation with respect to these variables and using \eqref{eq:maurey} therefore gives
\[
 \E\left\|\frac{\sqrt{2r}}h
          \sum_{a=1}^h\bm Y_a-\bm w\right\|_X
 \le C\sqrt{\frac{r\log(2n)}{nh}}.
\]
This completes the proof.
\end{proof}

\section{Auxiliary lemmas} \label{app:auxiliary}

\begin{lemma} \label{lem:half}
{Let \(\bm A,\bm B\in\mathbb R^{d\times d}\) be symmetric matrices and suppose} \({\alpha}I\preceq \bm A\preceq {\beta}I\), where \(0<{\alpha}\le {\beta}\), and \(\bm B\succeq0\).
Then
\begin{equation}
 \|(\bm A+\bm B)^{-1}\bm B\|_2\le\frac{{\alpha}+{\beta}}{2\sqrt{{\alpha\beta}}}.
 \label{eq:half}
\end{equation}
\end{lemma}
\begin{proof}
By definition, there exists $\vv \in \R^d$ with $\norm{\vv}=1$ such that 
\[
 \|(\bm A+\bm B)^{-1}\bm B\|_2=\norm{(\bm A+\bm B)^{-1}\bm B \vv}.
 \]
Set \(\bm w:=(\bm A+\bm B)^{-1}\bm B \vv\). Then \(\bm A\bm w=\bm B({\bm v}-\bm w)\). Thus
\[
 \langle{\bm v}-\bm w,\bm A\bm w\rangle
 =\langle{\bm v}-\bm w,\bm B({\bm v}-\bm w)\rangle\ge0,
\]
which gives 
\begin{equation} \label{eq:ineqw}
 \langle\bm w,\bm A\bm w\rangle \le  \langle\bm v,\bm A\bm w\rangle \le \|\bm A\bm w\|_2.
\end{equation}
We next give an upper bound for $\norm{\vw}$.
If \(\bm w\ne0\), put \(\bm u=\bm w/\|\bm w\|_2\) and
\(t=\langle\bm u,\bm A\bm u\rangle\). Then \(\|\bm u\|_2=1\)
and \({\alpha}\le t\le {\beta}\). 
From \eqref{eq:ineqw}, one has
\begin{equation} \label{eq:ineqw1}
\norm{\vw} \le \frac{ \|\bm A\bm u\|_2}{\langle\bm u,\bm A\bm u\rangle}= \frac{ \|\bm A\bm u\|_2}{t}.
\end{equation}
Since \({\alpha}I\preceq \bm A\preceq {\beta}I\), the spectral theorem gives \((\bm A-{\alpha}I)(\bm A-{\beta}I)\preceq0\). Expanding the product yields
\[
 \bm A^2-({\alpha}+{\beta})\bm A+{\alpha\beta}I\preceq0.
\]
This immediately gives 
\[
 \|\bm A\bm u\|_2^2=\langle\bm u,\bm A^2\bm u\rangle
 \le({\alpha}+{\beta})\langle\bm u,\bm A\bm u\rangle-{\alpha\beta}\|\bm u\|_2^2
 =({\alpha}+{\beta})t-{\alpha\beta}.
\]
Since \(t\ge {\alpha}>0\), division by \(t^2\) gives
\[
 \frac{\|\bm A\bm u\|_2^2}{t^2}
 \le\frac{({\alpha}+{\beta})t-{\alpha\beta}}{t^2}
 \le\frac{({\alpha}+{\beta})^2}{4{\alpha\beta}},\qquad {\alpha}\le t\le {\beta}.
\]
The scalar function in the middle expression attains its maximum at \(t=2{\alpha\beta}/({\alpha}+{\beta})\). Taking square roots and substituting the resulting bound into \eqref{eq:ineqw1} proves \eqref{eq:half}. The case \(\bm w=0\) is immediate. 
\end{proof}

The following entropy-integral inequality follows by combining \cite[Theorem~3.5(a)]{kmr2014} with the covering-number estimate in \cite[equation~(2.1)]{kmr2014}.

\begin{lemma}
\label{lem:kmr}
Let \(\mathcal V\subset\mathbb R^{d\times N}\) be a fixed nonempty
compact class, and let \(\bm g=(g_1,\ldots,g_N)^\top\) have
independent coordinates satisfying
\[
 \E g_j=0,\qquad \E g_j^2=1,\qquad
 \E\exp(g_j^2/K^2)\le2,\qquad 1\le j\le N,
\]
for some constant \(K>0\). Define
\[
 d_F=\sup_{V\in\mathcal V}\|V\|_F,\qquad
 d_{\rm op}=\sup_{V\in\mathcal V}\|V\|_2.
\]
For \(u>0\), let \(\mathcal N(\mathcal V,u)\) be the smallest
number of operator-norm balls of radius \(u\), with centers in
\(\mathcal V\), that cover \(\mathcal V\). Then, for every \(p\ge1\),
\begin{equation}
 \left(\E\sup_{V\in\mathcal V}\|V\bm g\|_2^p\right)^{1/p}
 \le C_K\left[
  \int_0^{d_{\rm op}}\sqrt{\log\mathcal N(\mathcal V,u)}\,du
  +d_F+\sqrt p\,d_{\rm op}\right],
 \label{eq:kmr}
\end{equation}
where \(C_K\) depends only on \(K\).
\end{lemma}

The following lemma states McDiarmid's bounded differences inequality. 

\begin{lemma}[{\cite{mcdiarmid1989}}]
\label{lem:mcdiarmid}
Let \(X_1,\ldots,X_N\) be independent random variables taking
values in measurable spaces \(\mathcal X_1,\ldots,\mathcal X_N\),
and let \(f:\mathcal X_1\times\cdots\times\mathcal X_N\to\mathbb R\)
be measurable with \(\E|f(X_1,\ldots,X_N)|<\infty\).
Suppose that constants \(c_1,\ldots,c_N\ge0\), with
\(\sum_{a=1}^N c_a^2>0\), satisfy
\[
 \bigl|f(x_1,\ldots,x_a,\ldots,x_N)
       -f(x_1,\ldots,x'_a,\ldots,x_N)\bigr|
 \le c_a
\]
for each \(1\le a\le N\) and all choices of the arguments in
their respective spaces. Then, for every \(t>0\),
\begin{equation}
 \Prob\left\{
 f(X_1,\ldots,X_N)-\E f(X_1,\ldots,X_N)\ge t
 \right\}
 \le\exp\left(-\frac{2t^2}{\sum_{a=1}^N c_a^2}\right).
 \label{eq:mcdiarmid}
\end{equation}
\end{lemma}


\begin{thebibliography}{99}
 % References reordered alphabetically by author; citation keys are unchanged.
\bibitem{balan2010signal}
{{R.}~Balan}, {\em On signal reconstruction from its spectrogram}, in 44th Annu. Conf. Inf. Sci. Syst. CISS, IEEE, 2010, pp.~1--4.

\bibitem{balan2006signal}
{{R.}~Balan, P.~Casazza, and D.~Edidin}, {\em On signal reconstruction without phase}, Appl. Comput. Harmon. Anal., 20 (2006), pp.~345--356.

\bibitem{beinert2018}
{R.~Beinert and G.~Plonka}, {\em Enforcing uniqueness in one-dimensional phase retrieval by additional signal information in time domain}, Appl. Comput. Harmon. Anal., 45 (2018), pp.~505--525.

\bibitem{bendory}
{T.~Bendory, R.~Beinert, and Y.~C.~Eldar}, {\em Fourier phase retrieval: uniqueness and algorithms}, in Compressed Sensing and its Applications, 2017, pp.~55--91.

\bibitem{candes2015cdp}
E.~J.~Cand\`es, X.~Li, and M.~Soltanolkotabi,
``Phase retrieval from coded diffraction patterns,''
\emph{Applied and Computational Harmonic Analysis}, vol.~39,
no.~2, pp.~277--299, 2015.


\bibitem{candes2015wf}
E.~J.~Cand\`es, X.~Li, and M.~Soltanolkotabi,
``Phase retrieval via Wirtinger flow: Theory and algorithms,''
\emph{IEEE Transactions on Information Theory}, vol.~61,
no.~4, pp.~1985--2007, 2015.


\bibitem{chai2010array}
{A.~Chai, M.~Moscoso, and G.~Papanicolaou}, {\em Array imaging using intensity-only measurements}, Inverse Probl., 27 ({2011}), 015005.

\bibitem{chen2017twf}
Y.~Chen and E.~J.~Cand\`es,
``Solving random quadratic systems of equations is nearly as easy
as solving linear systems,''
\emph{Communications on Pure and Applied Mathematics}, vol.~70,
no.~5, pp.~822--883, 2017.

\bibitem{davis1970rotation}
C.~Davis and W.~M.~Kahan,
``The rotation of eigenvectors by a perturbation. III,''
\emph{SIAM Journal on Numerical Analysis}, vol.~7,
no.~1, pp.~1--46, 1970.


\bibitem{edidin2019}
{D.~Edidin}, {\em The geometry of ambiguity in one-dimensional phase retrieval}, SIAM J. Appl. Algebr. Geom., 3 (2019), pp.~644--660.

\bibitem{gao2017phaseless}
B.~Gao and Z.~Xu,
``Phaseless recovery using the Gauss--Newton method,''
\emph{IEEE Transactions on Signal Processing}, vol.~65,
no.~22, pp.~5885--5896, 2017.


\bibitem{gross2017improved}
D.~Gross, F.~Krahmer, and R.~Kueng,
``Improved recovery guarantees for phase retrieval from coded
diffraction patterns,''
\emph{Applied and Computational Harmonic Analysis}, vol.~42,
no.~1, pp.~37--64, 2017.


\bibitem{Xcry1}
{R.~W.~Harrison}, {\em Phase problem in crystallography}, J. Opt. Soc. Am. A, 10 (1993), pp.~1046--1055.

\bibitem{hoeffding1963}
W.~Hoeffding, ``Probability inequalities for sums of bounded random
variables,'' \emph{Journal of the American Statistical Association},
vol.~58, no.~301, pp.~13--30, 1963.





\bibitem{huangGao}
{ G.~Huang and S. Li}, {\em PhaseLift for Coded Diffraction Patterns: Optimal Sampling Rate}, arXiv preprint arXiv:2608.02450, 2026.

\bibitem{huang2026near}
M.~Huang,
``Near-quadratic convergence of the Gauss--Newton method for
complex phase retrieval,''
\emph{IEEE Transactions on Information Theory}, vol.~72,
no.~1, pp.~222--245, 2026.


\bibitem{huang2026cdp}
M.~Huang, J.~Wen, and R.~Zhang,
``Recovery performance of PhaseLift for phase retrieval from coded
diffraction patterns,''
\emph{Inverse Problems}, vol.~42, Art.~no.~045021, 2026.


\bibitem{huangPerfm}
{ M.~Huang and Z.~Xu}, {\em Performance bounds of the intensity-based estimators for noisy phase retrieval}, Appl. Comput. Harmon. Anal., 68 (2024), 101584.

\bibitem{huangJoC}
{M.~Huang and Z.~Xu}, {\em No existence of a linear algorithm for the one-dimensional Fourier phase retrieval}, J. Complexity, 86 (2025), 101886.

\bibitem{kmr2014}
F.~Krahmer, S.~Mendelson, and H.~Rauhut,
``Suprema of chaos processes and the restricted isometry property,''
\emph{Communications on Pure and Applied Mathematics}, vol.~67,
no.~11, pp.~1877--1904, 2014.


\bibitem{li2025taf}
H.~Li and J.~Li,
``Truncated amplitude flow with coded diffraction patterns,''
\emph{Inverse Problems}, vol.~41, Art.~no.~015002, 2025.


\bibitem{li2022sampling}
H.~Li, S.~Li, and Y.~Xia,
``Sampling complexity on phase retrieval from masked Fourier
measurements via Wirtinger flow,''
\emph{Inverse Problems}, vol.~38, Art.~no.~105004, 2022.



\bibitem{liu2008phase}
{Y.~J.~Liu et al.}, {\em Phase retrieval in X-ray imaging based on using structured illumination}, Phys. Rev. A, 78 (2008), 023817.

\bibitem{Loewen}
{E.~G.~Loewen and E.~Popov}, {\em Diffraction gratings and applications}, CRC Press, 1997.

\bibitem{mcdiarmid1989}
C.~McDiarmid, ``On the method of bounded differences,'' in
\emph{Surveys in Combinatorics, 1989}, J.~Siemons, Ed.
Cambridge University Press, 1989, pp.~148--188.


\bibitem{miao2008extending}
{J.~Miao, T.~Ishikawa, Q.~Shen, and T.~Earnest}, {\em Extending x-ray crystallography to allow the imaging of noncrystalline materials, cells, and single protein complexes}, Annu. Rev. Phys. Chem., 59 (2008), pp.~387--410.

\bibitem{Xcry2}
{R.~P.~Millane}, {\em Phase retrieval in crystallography and optics}, J. Opt. Soc. Am. A, 7 (1990), pp.~394--411.

\bibitem{paige1982lsqr}
C.~C.~Paige and M.~A.~Saunders,
``LSQR: An algorithm for sparse linear equations and sparse least squares,''
\emph{ACM Transactions on Mathematical Software}, vol.~8,
no.~1, pp.~43--71, 1982.

\bibitem{tropp2012}
J.~A.~Tropp, ``User-friendly tail bounds for sums of random matrices,''
\emph{Foundations of Computational Mathematics}, vol.~12,
pp.~389--434, 2012.


\bibitem{vershynin2018hdp}
R.~Vershynin,
\emph{High-Dimensional Probability: An Introduction With Applications in Data Science}.
Cambridge, U.K.: Cambridge Univ. Press, 2018.

\bibitem{walther1963question}
{A.~Walther}, {\em The question of phase retrieval in optics}, Optica Acta, 10 (1963), pp.~41--49.

\bibitem{wang2018taf}
G.~Wang, G.~B.~Giannakis, and Y.~C.~Eldar,
``Solving systems of random quadratic equations via truncated
amplitude flow,''
\emph{IEEE Transactions on Information Theory}, vol.~64,
no.~2, pp.~773--794, 2018.
\end{thebibliography}
\end{document}